\documentclass[10pt]{article}

\usepackage[T1]{fontenc}
\usepackage[utf8]{inputenc}
\usepackage{lmodern}        
\usepackage[activate={true,nocompatibility},final]{microtype}  

\usepackage[margin=0.9in]{geometry}    

\usepackage{authblk}                    

\usepackage{amsmath}
\usepackage{amssymb,amsfonts}
\usepackage{mathtools}      
\usepackage{bm}             
\usepackage{nicefrac}
\allowdisplaybreaks         

\usepackage{amsthm}
\usepackage{thm-restate}    

\usepackage{graphicx}
\graphicspath{{figures/}{./}}                       
\usepackage{booktabs}                               
\usepackage{multirow,makecell,array}
\usepackage{longtable}                              
\usepackage[font=small,labelfont=bf]{caption}       
\usepackage{subcaption}                             
\usepackage{algorithm}
\usepackage{algpseudocode}                          

\usepackage{enumitem}
\setlist[itemize]{leftmargin=2.2em,itemsep=2pt,topsep=2pt}
\setlist[enumerate]{leftmargin=2.2em,itemsep=2pt,topsep=2pt}

\usepackage{xcolor}
\definecolor{LinkColor}{rgb}{0.10,0.40,0.75}        
\definecolor{CiteColor}{rgb}{0.70,0.25,0.20}        
\definecolor{UrlColor} {rgb}{0.20,0.50,0.50}        
\definecolor{TodoColor}{rgb}{0.80,0.30,0.10}        

\usepackage{tikz}
\usetikzlibrary{positioning,calc,arrows.meta,decorations.pathreplacing,patterns,shapes.geometric}

\usepackage[round,sort&compress]{natbib}

\usepackage{url}
\usepackage{hyperref}
\hypersetup{
  colorlinks=true,
  linkcolor=LinkColor,
  citecolor=CiteColor,
  urlcolor=UrlColor,
  breaklinks=true,
  bookmarksnumbered=true,
}
\usepackage{bookmark}                               
\usepackage[capitalise,nameinlink,noabbrev,sort&compress]{cleveref}  

\numberwithin{equation}{section}

\newif\ifdraft \draftfalse
\ifdraft
  \usepackage{lineno}\linenumbers
  \newcommand{\todo}[1]{\textcolor{TodoColor}{\textbf{[TODO:}~#1\textbf{]}}}
\else
  \newcommand{\todo}[1]{}                            
\fi

\usepackage[scaled=0.94]{helvet}
\definecolor{detectblue}{HTML}{1F4FA8}   
\definecolor{attribgold}{HTML}{B97A0A}   
\definecolor{bluetint}  {HTML}{F3F7FC}   
\definecolor{goldtint}  {HTML}{FCF8EE}   
\definecolor{inkgray}   {HTML}{4D4D4D}   
\definecolor{skyblue}   {HTML}{E3F1FC}   

\usepackage{titlesec}
\titleformat*{\section}{\Large\sffamily\bfseries\color{detectblue}}
\titleformat*{\subsection}{\large\sffamily\bfseries\color{detectblue}}
\titleformat*{\subsubsection}{\normalsize\sffamily\bfseries\color{detectblue}}
\titleformat{\paragraph}[runin]{\normalsize\bfseries}{}{0pt}{}
\titlespacing*{\paragraph}{0pt}{1.4ex plus .4ex}{0.9em}

\usepackage{tocloft}

\usepackage{fancyhdr}
\fancyheadoffset{0pt}
\newcommand{\gradientrule}{\begin{tikzpicture}%
  \shade[left color=detectblue, right color=attribgold]
    (0,0) rectangle (\textwidth, 0.06em);\end{tikzpicture}}

\usepackage[most]{tcolorbox}
\newcommand{\resultboxset}[1]{\tcolorboxenvironment{#1}{enhanced jigsaw, breakable,
  boxrule=0pt, frame hidden, colback=bluetint, arc=1.2mm,
  borderline west={1.5pt}{0pt}{detectblue!70},
  left=7pt, right=7pt, top=4pt, bottom=4pt, before skip=11pt, after skip=11pt}}
\newcommand{\setupboxset}[1]{\tcolorboxenvironment{#1}{enhanced jigsaw, breakable,
  boxrule=0pt, frame hidden, colback=goldtint, arc=1.2mm,
  borderline west={1.5pt}{0pt}{attribgold!75},
  left=7pt, right=7pt, top=4pt, bottom=4pt, before skip=11pt, after skip=11pt}}
\AtBeginDocument{%
  \resultboxset{theorem}\resultboxset{lemma}\resultboxset{proposition}%
  \resultboxset{corollary}\resultboxset{fact}\resultboxset{conjecture}%
  \resultboxset{restatable}%
  \setupboxset{definition}\setupboxset{assumption}\setupboxset{remark}%
  \setupboxset{example}\setupboxset{problem}%
}
\newtcolorbox{statementbox}{enhanced jigsaw, breakable, boxrule=0pt, frame hidden,
  colback=bluetint, arc=1.2mm, borderline west={1.5pt}{0pt}{detectblue!70},
  left=7pt, right=7pt, top=4pt, bottom=4pt, before skip=11pt, after skip=11pt}

\newtcolorbox{poincarebox}{enhanced jigsaw, boxrule=0pt, frame hidden,
  colback=skyblue, arc=2mm,
  left=9pt, right=9pt, top=5pt, bottom=5pt}

\newcommand{\E}{\mathbb{E}}                       

\newcommand{\eps}{\varepsilon}

\DeclareMathOperator{\supp}{supp}

\DeclarePairedDelimiter{\norm}{\lVert}{\rVert}

\DeclarePairedDelimiterX{\inner}[2]{\langle}{\rangle}{#1,#2}   

\newcommand{\pnat}{p_{\mathrm{nat}}}            
\newcommand{\pwm}{p_{\mathrm{wm}}}              
\newcommand{\pbar}{\bar p}                      
\newcommand{\Xs}{\mathcal{X}}                   
\newcommand{\Keys}{\mathcal{K}}                 
\DeclareMathOperator{\KLop}{KL}
\newcommand{\KL}[2]{\KLop\!\left(#1 \,\|\, #2\right)}   
\DeclareMathOperator{\I}{I}                      
\newcommand{\Hb}{H_{\mathrm{b}}}                
\newcommand{\Lmap}{L}                           

\newcommand{\Unif}{\mathrm{Unif}}               
\DeclareMathOperator{\JSD}{JSD}                  
\newcommand{\Phibud}{\Phi}                       
\newcommand{\ntask}[1]{n_{#1}}                   
\newcommand{\att}{\mathrm{att}}                  
\newcommand{\ext}{\mathrm{ext}}                  
\providecommand{\det}{\mathrm{det}}              

\theoremstyle{plain}
\newtheorem{theorem}{Theorem}[section]
\newtheorem{proposition}[theorem]{Proposition}
\newtheorem{lemma}[theorem]{Lemma}

\newtheorem{fact}[theorem]{Fact}

\theoremstyle{definition}
\newtheorem{definition}[theorem]{Definition}
\newtheorem{assumption}[theorem]{Assumption}
\newtheorem{example}[theorem]{Example}

\theoremstyle{remark}
\newtheorem{remark}[theorem]{Remark}

\crefname{theorem}{Theorem}{Theorems}          \Crefname{theorem}{Theorem}{Theorems}
\crefname{proposition}{Proposition}{Propositions}
\Crefname{proposition}{Proposition}{Propositions}
\crefname{lemma}{Lemma}{Lemmas}                \Crefname{lemma}{Lemma}{Lemmas}
\crefname{corollary}{Corollary}{Corollaries}   \Crefname{corollary}{Corollary}{Corollaries}
\crefname{conjecture}{Conjecture}{Conjectures} \Crefname{conjecture}{Conjecture}{Conjectures}
\crefname{fact}{Fact}{Facts}                   \Crefname{fact}{Fact}{Facts}
\crefname{definition}{Definition}{Definitions} \Crefname{definition}{Definition}{Definitions}
\crefname{assumption}{Assumption}{Assumptions} \Crefname{assumption}{Assumption}{Assumptions}
\crefname{example}{Example}{Examples}          \Crefname{example}{Example}{Examples}
\crefname{problem}{Problem}{Problems}          \Crefname{problem}{Problem}{Problems}
\crefname{remark}{Remark}{Remarks}             \Crefname{remark}{Remark}{Remarks}
\crefname{claim}{Claim}{Claims}                \Crefname{claim}{Claim}{Claims}
\crefname{algorithm}{Algorithm}{Algorithms}    \Crefname{algorithm}{Algorithm}{Algorithms}

\hypersetup{pdftitle={Watermark Forensics for Generative Models: An Information-Theoretic Perspective},
  pdfauthor={Xiaoyu Li, Zheng Gao, Xiaoyan Feng, Jiaojiao Jiang, Yulei Sui, Jiankun Hu}}

\begin{document}

\thispagestyle{empty}
\begin{center}
  {\huge\sffamily\bfseries Watermark Forensics for Generative Models\par}
  \vspace{0.5em}
  {\Large\sffamily\bfseries\color{black!85}
   An Information-Theoretic Perspective\par}
  \vspace{0.7em}
  \gradientrule
  \vspace{0.9em}
  {\normalsize
   \begin{tabular}{@{}c@{\hspace{2.8em}}c@{\hspace{2.8em}}c@{}}
   Xiaoyu Li$^{1}$ & Zheng Gao$^{1}$ & Xiaoyan Feng$^{2}$ \\[0.6em]
   Jiaojiao Jiang$^{1}$ & Yulei Sui$^{1}$ & Jiankun Hu$^{1}$
   \end{tabular}\par}
  \vspace{0.95em}
  {\small $^{1}$University of New South Wales \qquad $^{2}$Griffith University\par}
  \vspace{0.55em}
  {\footnotesize\ttfamily
   \{xiaoyu.li2,\,zheng.gao1,\,jiaojiao.jiang,\,yulei.sui\}@unsw.edu.au\\[0.3em]
   xiaoyan.feng@griffithuni.edu.au\qquad j.hu@adfa.edu.au\par}
  \vspace{0.7em}
  {\small\itshape\color{inkgray}\today\par}
\end{center}
\vspace{0.6em}

\vspace{-5mm}
\begin{poincarebox}
{\small\linespread{1.0}\selectfont
 \noindent\textbf{\sffamily\color{detectblue}Abstract.}\enspace
A watermark in a generative model's output is normally asked only whether a text is machine-made. The
same mark can do more: \emph{attribute} it to the user who produced it, \emph{extract} a hidden payload,
or \emph{localize} the part that survives editing. These form a forensic ladder, and we ask what each
rung costs in the sample length $n$.

A single object organizes the answers. Let $S$ be the secret the mark carries (a user's identity, a
payload), and let the \emph{information profile} $\nu(t)=I(S;X_t\mid X_{<t})$ measure how much the
$t$-th token reveals about $S$ given the earlier ones. Its \emph{total} pays for attribution and
extraction; how it is \emph{spread} across the text pays for localization; and detection alone is paid
for not by information but by presence, the distance from the marked to the unmarked distribution. The
two quality models in the literature, a mark subtle on every token and one that stamps a few tokens
loudly, are then incomparable ways of capping this one profile.

Our main theorem settles the ladder's entropy column. For schemes that leave the model's output
distribution statistically intact, attributing a text to one of $N$ users costs $\Theta(\log N / h)$
tokens over every stationary-ergodic source of entropy rate $h$, sharp to a $(1+o(1))$ factor; to our
knowledge the first tight entropy-rate law for multi-user attribution (via exact alignment; the
edit-robust case is open). A warning: the natural collision-counting argument overcharges without bound,
and only a decoder that scores each candidate by its own realized surprisal, not against the others,
attains the rate while almost never implicating an innocent user. A matching converse makes the law
two-sided, and the same accounting prices extraction of an $\ell$-bit payload at $\Theta(\ell/h)$. Two
gaps are real rather than modeling artifacts: a window $\Theta(\log N)$ tokens wide in which a text is
provably machine-made yet unattributable, and an uncertainty principle trading a mark's footprint
against its localization resolution. Since emerging regulation presumes an \emph{extractable} mark, these
are limits it inherits rather than escapes.

Across three language models (GPT-2, Pythia-410M, Qwen2.5) under deployed watermarks, the predicted
constants appear: the collision analysis overcharges by about $1.6\times$, and detection saturates its
$\Theta(\log N)$ window before attribution begins.
\par}
\end{poincarebox}

\vspace{-1.4em}
\begin{center}
\begin{tikzpicture}[font=\sffamily,>=Stealth,line cap=round]

\begin{scope}
  \def\W{1.55}                              
  \def\hA{1.58}\def\hB{2.92}\def\hC{3.80}\def\hD{4.68}   
  \fill[black!8,line join=round]
    (0.07,-0.07) -- (0.07,\hA-0.07) -- (\W+0.07,\hA-0.07) -- (\W+0.07,\hB-0.07)
    -- (2*\W+0.07,\hB-0.07) -- (2*\W+0.07,\hC-0.07) -- (3*\W+0.07,\hC-0.07)
    -- (3*\W+0.07,\hD-0.07) -- (4*\W+0.07,\hD-0.07) -- (4*\W+0.07,-0.07) -- cycle;
  \draw[->,inkgray!80,line width=0.7pt] (0,0) -- (0,5.4) node[left=1pt,inkgray,font=\scriptsize]{cost $n$};
  \draw[inkgray!50,line width=0.7pt] (0,0) -- (4.15*\W,0);
  \fill[bluetint,draw=detectblue,line width=1.1pt,line join=round]
    (0,0) -- (0,\hA) -- (\W,\hA) -- (\W,\hB) -- (2*\W,\hB) -- (2*\W,\hC)
    -- (3*\W,\hC) -- (3*\W,\hD) -- (4*\W,\hD) -- (4*\W,0) -- cycle;
  \draw[detectblue!30,line width=0.5pt] (\W,0)--(\W,\hB);
  \draw[detectblue!30,line width=0.5pt] (2*\W,0)--(2*\W,\hC);
  \draw[detectblue!30,line width=0.5pt] (3*\W,0)--(3*\W,\hD);
  \foreach \k/\lab in {0/0,1/1,2/2,3/3}{
    \node[circle,fill=detectblue,text=white,font=\scriptsize\bfseries,inner sep=1.4pt]
      at ({(\k+0.5)*\W},0.30) {\lab};
  }
  \node[font=\footnotesize\bfseries,text=detectblue] at (0.5*\W,\hA-0.30) {Detect};
  \node[font=\scriptsize,text=inkgray,align=center]   at (0.5*\W,0.82) {is it\\marked?};
  \node[font=\footnotesize\bfseries,text=detectblue] at (1.5*\W,\hB-0.30) {Attribute};
  \node[font=\scriptsize,text=inkgray,align=center]   at (1.5*\W,1.05) {which of\\$N$ users?};
  \node[font=\footnotesize\bfseries,text=detectblue] at (2.5*\W,\hC-0.30) {Extract};
  \node[font=\scriptsize,text=inkgray,align=center]   at (2.5*\W,1.05) {what\\payload?};
  \node[font=\footnotesize\bfseries,text=detectblue] at (3.5*\W,\hD-0.30) {Localize};
  \node[font=\scriptsize,text=inkgray,align=center]   at (3.5*\W,1.05) {where?};
  \node[font=\scriptsize,text=inkgray]    at (0.5*\W,\hA+0.40) {$\Theta(1/\Delta)$};
  \node[font=\scriptsize\bfseries,text=detectblue] at (1.5*\W,\hB+0.40) {$\Theta(\tfrac{\log N}{\Delta})$};
  \node[font=\scriptsize,text=inkgray]    at (2.5*\W,\hC+0.40) {$\Theta(\tfrac{\ell}{\Delta})$};
  \node[font=\scriptsize,text=inkgray]    at (3.5*\W,\hD+0.40) {$m\,w{=}\Omega(n)$};
  \node[font=\scriptsize\itshape,text=inkgray] at (0.5*\W,-0.34) {mechanism-blind};
  \node[rounded corners=2.5pt,fill=detectblue,text=white,font=\scriptsize\bfseries,
        align=center,inner sep=3.5pt] (hl) at (1.62*\W,4.92)
        {Detection $\boldsymbol{\neq}$ Attribution};
  \draw[->,detectblue,line width=1.1pt] (hl.south) to[bend right=45] (\W,0.55*\hA+0.45*\hB);
  \node[font=\small\bfseries,text=detectblue] at (2.0*\W,-0.92) {(a)\; the forensic ladder};
\end{scope}

\begin{scope}[shift={(8.5,0)}]
  \def\bw{5.5}                       
  \def\hbias{0.62}                   
  \node[font=\footnotesize\bfseries,text=inkgray,align=center,anchor=north] at (\bw/2,5.42)
     {one measure $\nu(t)=I(S;X_t\mid X_{<t})$};

  \def\ty{3.05}                      
  \node[anchor=west,font=\scriptsize\bfseries,text=detectblue] at (0,\ty+1.55) {biasing $\;\mathcal B(\Delta)$};
  \node[anchor=west,font=\scriptsize,text=inkgray] at (0,\ty+1.22) {subtle on every token, full support};
  \draw[inkgray!50,line width=0.7pt] (0,\ty) -- (\bw,\ty);
  \foreach \i/\hh in {0/0.58,1/0.62,2/0.49,3/0.60,4/0.45,5/0.61,6/0.53,7/0.62,8/0.47,9/0.56,10/0.62,11/0.51,12/0.59,13/0.54}{
    \fill[detectblue!70] ({0.12+\i*0.385},\ty) rectangle ({0.12+\i*0.385+0.24},\ty+\hh);
  }
  \draw[dashed,detectblue,line width=0.9pt] (0,\ty+\hbias+0.06) -- (\bw,\ty+\hbias+0.06)
     node[right,font=\scriptsize,text=detectblue,inner sep=2pt] {$\Delta$: an $L^\infty$ cap};

  \def\by{0.55}                      
  \node[anchor=west,font=\scriptsize\bfseries,text=attribgold!85!black] at (0,\by+1.78) {embedding $\;\mathcal E(m)$};
  \node[anchor=west,font=\scriptsize,text=inkgray] at (0,\by+1.45) {a short, loud stamp};
  \draw[inkgray!50,line width=0.7pt] (0,\by) -- (\bw,\by);
  \foreach \i/\hh in {0/1.32,1/1.07,2/1.22}{
    \fill[attribgold] ({0.12+\i*0.40},\by) rectangle ({0.12+\i*0.40+0.27},\by+\hh);
  }
  \draw[decorate,decoration={brace,amplitude=4pt,mirror},attribgold!85!black,line width=0.8pt]
     (0.08,\by-0.10) -- (1.31,\by-0.10);
  \node[anchor=north,font=\scriptsize,text=attribgold!85!black] at (0.7,\by-0.18) {footprint $m$: an $L^0$ cap};

  \node[anchor=north,align=center,font=\scriptsize,text=inkgray] at (\bw/2,\by-0.60)
     {{\bfseries same shaded mass} $\|\nu\|_1=I(S;X)\ge\log N$, two shapes};

  \node[font=\small\bfseries,text=detectblue] at (\bw/2,-0.92) {(b)\; one profile $\nu(t)$, two caps};
\end{scope}

\draw[inkgray!22,line width=0.6pt] (7.05,-0.4) -- (7.05,5.2);

\end{tikzpicture}
\vspace{-9mm}

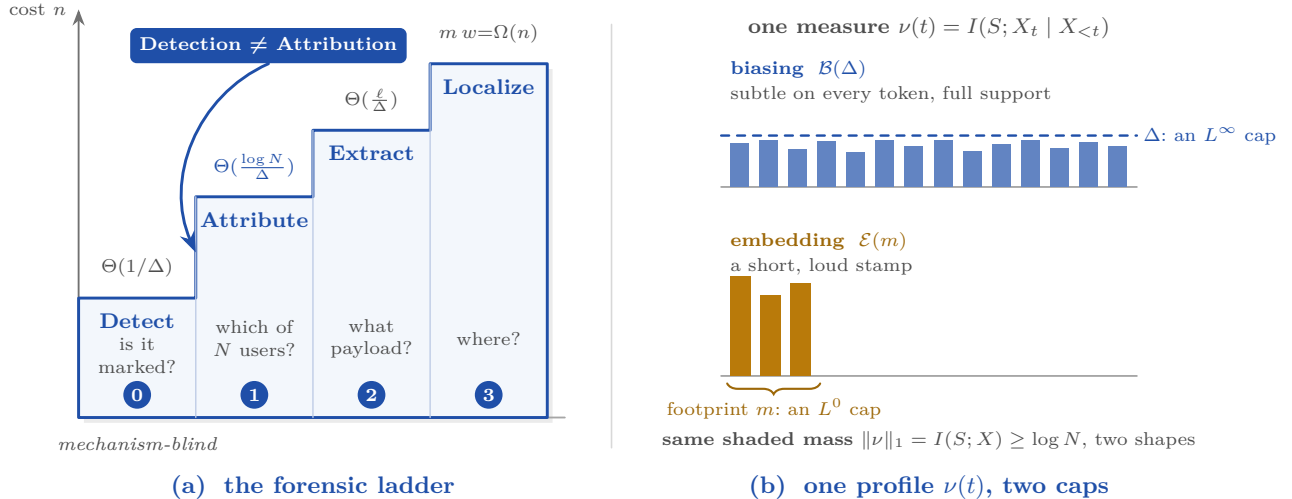
\captionof{figure}{\scriptsize \textbf{The forensic ladder, and the one measure beneath it.}
\textsf{(a)} The four forensic questions form a \emph{sample-complexity ladder}. Detection is
cheap and mechanism-blind; but attribution is not free (naming one of $N$ users costs
$\Theta(\log N/\Delta)$ or $\Theta(\log N/h)$ in the appropriate per-token currency, the jump from detection to attribution), and extraction and crop-robust
localization sit higher still. \textsf{(b)} The rungs above detection are read off a single object, the
information profile $\nu(t)=I(S;X_t\mid X_{<t})$. The literature's two quality models are two
\emph{incomparable} caps on the same $\nu$: a \emph{biasing} mark spends the mass thinly under
an $L^\infty$ height cap $\Delta$ (subtle on every token, full support), an \emph{embedding}
mark concentrates it under an $L^0$ width cap $m$ (a short, loud stamp). Same shaded mass
$\|\nu\|_1=I(S;X)$, two shapes. \vspace{-10mm}}
\label{fig:hero}
\end{center}

\clearpage
\tableofcontents
\newpage

\section{Introduction}
\label{sec:intro}

Consider two deployed watermarks for generative models. \emph{Tree-Ring} \citep{wen2023treering} embeds a ring
pattern into the Fourier spectrum of a diffusion model's initial noise; to test an image one inverts the
generation and matches the ring. The \emph{green-list} of \citet{kirchenbauer2023watermark} colours half the
vocabulary green at each decoding step and tilts the next token toward it; to test a passage one counts green
tokens. Tuned to the same detection ROC, the two are interchangeable for the question \emph{is this
AI-generated}. Now ask a second question. Given the marked image, can one \emph{point to} where the mark
lives? Yes: invert to the noise and read off the ring. Given the marked text, can one point to which tokens
\emph{are} the mark? Not in any strict sense; every token is plausible under the natural distribution, and the
mark exists only in the statistical aggregate of the whole passage. In the second case the watermark is a
\emph{property} of the carrier, not an \emph{object inside} it, and asking which tokens are the mark is a
category error, much as one cannot ask which molecules of a warm gas carry its temperature. The distinction has
direct consequences for attribution, extraction, and regulation, and the literature treats it only case by
case.

This paper makes that informal gap quantitative. We treat the forensic uses of a watermark as a ladder of
four statistical problems on a length-$n$ carrier: \textbf{Level 0} \emph{detection} (machine-generated or
not), \textbf{Level 1} \emph{attribution} (which of $N$ users produced it), \textbf{Level 2} \emph{extraction}
(recover an $\ell$-bit payload), and \textbf{Level 3} \emph{localization} (which sub-region carries the mark).
We ask for the sample complexity of each. The detection rung is by now well understood: under a per-token
distortion budget $\Delta$ (each marked conditional within KL $\Delta$ of the natural one), detection costs
$\Theta(1/\Delta)$ tokens, with sharp constants \citep{huang2023optimal, li2024statistical, cai2024better}. We
import this rung and never reprove it. Our question is what happens \emph{above} it, and whether the
embedding-versus-biasing distinction that separates our two examples is a real boundary or a matter of taste.

\paragraph{One object.} It is a real boundary, and one measure draws it. Every forensic question asks how much
of the secret $S$ a decoder can recover from the carrier $X$; by data processing this is at most the
\emph{mutual information} $I(S;X)$, the number of nats that observing $X$ reveals about $S$ (for attribution and
distortion-free extraction, conditioned on the decoder's public key registry; \cref{sec:prelim}, and a primer
on the information-theoretic notions is in \cref{sec:itbackground}). The chain rule resolves it position by position,
$I(S;X)=\sum_t I(S;X_t\mid X_{<t})$, each summand the information about $S$ that the $t$-th token adds given its
past. This is the \emph{information profile} of the
scheme's forensic payload $S$,
\[
  \nu(t)\;:=\;I(S;X_t\mid X_{<t}),\qquad t=1,\dots,n.
\]
Attribution and extraction are recovery questions, \emph{how much} of $S$, so each is governed by
the total \emph{mass} $\norm{\nu}_1=I(S;X)$; localization asks not how much but \emph{where}, so it reads the
\emph{shape} $\supp\nu$, the footprint of \cref{def:embedbias}. Detection is the exception: the imported
baseline asks only whether the mark is \emph{present}, priced by how far the marked law sits from natural and
not by the mass --- a scheme can make $S$ recoverable yet be undetectable without the key. Our two opening examples are the extremal
profiles (\cref{fig:hero}): the green-list spreads its mass $\Delta$-thin over every position, while Tree-Ring
concentrates it on a vanishing set of transform coordinates. The profile is written once, at generation time:
by data processing, no reading of the carrier can recover more about $S$ than the sampler put in, so what any
decoder can ever learn is fixed by the $\nu$ the scheme writes. This is the paper's one idea: the forensic
power of a watermark is set by the information it makes recoverable, position by position --- the profile ---
and not by the mechanics of how the mark is embedded. The literature asks four questions of a watermark and
treats them as four problems; detection is the mechanism-blind baseline, and the three rungs above it are
readings of $\nu$. Everything below is a statement about functionals of it. The profile earns its keep as the object that unifies the ladder above detection --- attribution,
extraction, and localization, together with the two quality models, become functionals of one measure ---
more than as a proof device: the scalar mass $\norm{\nu}_1$ already separates attribution and extraction,
while the per-position \emph{shape} does load-bearing work only in the mechanism dichotomy
(\cref{thm:dichotomy}) and in localization (\cref{thm:crop}). We are explicit throughout about which is which.

\paragraph{Contributions.} The contribution is the object; each result below is a way of reading it.
\begin{itemize}
  \item \textbf{One measure, two norms (\cref{prop:profile}).} The chain rule gives the exact identity
  $\norm{\nu}_1=I(S;X)=\Phi(n)$, the quantity the forensic-recovery budget caps. The two quality models of the
  literature are then norm caps on $\nu$: the distributional model $\mathcal B(\Delta)$ caps it in $L^\infty$
  ($\nu(t)\le\Delta$ at every position, via the golden formula), and the footprint model $\mathcal E(m)$ caps
  it in $L^0$ ($\nu$ supported on at most $m$ positions, by conditional screening). We organize what were two
  incomparable quality models and a separately stated footprint axis as one object read two ways; the
  incomparability becomes definitional, two norms on one measure. Only the mass is invariant to the choice of
  dictionary and ordering; the caps, like sparsity, are relative to a fixed basis. Crop-\emph{robustness} is
  the one property not of $\nu$ alone: it additionally needs an edit coupling on the carrier footprint.
  \item \textbf{The entropy column closes (\cref{thm:entropyclose}).} For statistically distortion-free
  schemes with i.i.d.\ per-user keys (the decoder knows $\pnat$ and the deployed key set), attribution among
  $N$ users costs
  \[
    \ntask{\att}\;=\;\Theta\Big(\frac{\log N}{h}\Big),
    \qquad\text{and}\qquad
    \ntask{\att}\;=\;(1+o(1))\,\frac{\log N}{h}
  \]
  as the per-innocent false-positive level $\delta\to0$ with $\log(1/\delta)=o(\log N)$ (the $\log(1/\delta)$
  term is then lower order),
  for \emph{every} stationary-ergodic source of entropy rate $h>0$; the Shannon--McMillan--Breiman theorem is
  the only source hypothesis. To our knowledge this is the first tight $h$-rate for multi-user attribution.
  Achievability is a surprisal-threshold decoder, and its threshold is \emph{forced}, not tuned: an innocent
  key matches the observed text only with probability $\pnat(X)=e^{-\imath(X)}$, where $\imath(X)=-\log\pnat(X)$
  is the realized surprisal, so to implicate no innocent above $\delta/N$ the decoder must answer only when
  $\imath(X)>\tau=\log(N/\delta)$ (then $\pnat(X)<\delta/N$); completeness only asks that the true text, being
  $\pnat$-typical, clear that same bar, which it does once $n\ge\tau/h$. The obvious
  decoder instead reports any \emph{consistent} key; bounding its error means counting the innocents
  consistent by chance, $(N-1)e^{-nr_2}$, and waiting for that list to clear at $\log N/r_2$, the R\'enyi-2
  collision rate. But the forensic question is not list size; it is naming one user with a guarantee, governed
  by the Shannon rate $h$. Since $r_2\le h$, and can be far smaller, the consistency count overcharges by the
  unbounded factor $h/r_2$ (\cref{rem:listsize}). That guarantee is exact
  and non-asymptotic, and for the deterministic witness scheme pointwise, which is the form a forensic
  guarantee has to take: no innocent user is implicated except with probability $\delta/N$ --- at every
  length, for every source, for every realized set of innocent keys. The matching converse is Fano plus a mixture-entropy
  bound on the i.i.d.\ key draw.
  Extraction is the same theorem at $N=2^\ell$, giving $\Theta(\ell/h)$ (\cref{cor:extraction-h}). The
  entropy column of the ladder, previously a converse-only ledger, is thereby two-sided.
  \item \textbf{The rest of the ladder, as readings of $\nu$ (\cref{thm:budget,thm:separation,thm:dichotomy,thm:crop,prop:region}).}
  The mass and the shape of one measure, read five more ways. The
  forensic-recovery budget $\Phi(n)=\norm{\nu}_1\le n\Delta$, respectively $\le n\,h_n$
  ($h_n:=H(\pnat^{(n)})/n\downarrow h$, the entropy rate). The co-located window: a length
  at which a biasing watermark is detectable with error $\to0$ yet every attribution rule errs with
  probability $\to1$, within a single scheme. The dichotomy at the extremal profiles: $\mathcal B(\Delta)$
  couples, $\ntask{\att}=\Omega(\log N/\Delta)$ (tight for regular mixing sources), while $\mathcal E(m)$
  decouples, $\ntask{\att}=\Theta(\log N/\log q)$, $\Delta$-free; a localization map enables the cheaper model
  but does not decide it (\cref{cor:bridge}). The footprint--resolution uncertainty $m\cdot w=\Omega(n)$:
  fine-resolution crop-robustness forces full support. And the rate region: attribution and extraction split
  one mass, $R_{\att}+R_{\ext}\le(1+o(1))\,n\Delta$, order-tight.
\end{itemize}

\paragraph{What is imported and what is ours.} The Level-0 detection rate and the fidelity--detectability
frontier are imported \citep{cai2024better, li2024statistical, huang2023optimal}, as is the distortion-free
detection rate \citep{christ2024undetectable}. The attribution converse in the $\mathcal B(\Delta)$ column is
the single-user Fano specialization of the distortion-constrained fingerprinting converse of
\citet{moulin2008universal, somekhbaruch2005capacity}; we claim no novelty for it. The false-positive
mechanism behind \cref{thm:entropyclose} is the folklore likelihood-ratio Markov bound (a one-line change of
measure); we claim its use here, not the inequality. Ours are: the profile as the organizing object and the
norm-cap reading of the quality models; the closed entropy column, that is, the decoder, its exact
false-positive side, and tightness at full stationary-ergodic generality; the co-located window; the
footprint--resolution uncertainty; the rate
region; and the taxonomy. We flag plainly that the entropy-column witness decodes by exact alignment and is
not edit-robust; edit-robust $\Delta$-free attribution remains open. The closest prior work proves
complementary things: \citet{cohen2024many} construct multi-user watermarks where detection is cheap and
tracing is expensive, with a cryptographic counting argument but no tight information-theoretic rate in $h$;
\citet{jiang2024attribution} argue empirically that detection and attribution coincide under a strict
threshold; the SoK of \citet{sok2024watermarking} taxonomizes schemes on a security axis, with no
sample-complexity content; concurrent work \citep{distribembed2025} studies the scalar $I(S;X)$ without the
positional decomposition on which our questions turn. We position against each in \cref{sec:related}.

\paragraph{Why these questions.} Watermarking is becoming provenance infrastructure: the mechanism by which a
world filling with synthetic content is supposed to answer \emph{who made this}. The four rungs of the ladder
are the operational content of that hope, and they are not one capability but four, with four prices. What
theory can contribute is the price list: which questions are priced by the population (attribution: a
$\log N$ toll in the appropriate currency), which are not (detection), and which are structurally
unavailable (fine crop-robust localization for any small-footprint mark, \cref{thm:crop}). Regulation makes
the list concrete. Watermark statutes tacitly assume the embedding ontology, asking for marks ``in a
machine-readable format'' and the like: language that presumes an extractable
object. Our dichotomy prices that presumption: a biasing watermark cannot meet an attribution or localization
mandate at any fidelity without a sample cost that diverges as the mark is made stealthier, and the closed
entropy column gives an attribution mandate its price when the mark is distortion-free,
$\Theta((\log N+\log(1/\delta))/h)$ tokens with a per-innocent false-positive guarantee of $\delta/N$. A
statute can mandate a capability; it cannot repeal an inequality. We return to this in \cref{sec:discussion}.

\paragraph{Organization.} \Cref{sec:prelim} sets the model, the quality assumptions, and the information
profile (\cref{prop:profile}). \Cref{sec:results} states the results: the forensic-recovery budget, the closed entropy
column (\cref{thm:entropyclose}), the co-located window, the dichotomy, the uncertainty bound, and the rate
region, with a numerical illustration and, on three watermarked language models (GPT-2, Pythia, Qwen2.5), a
measurement of the operative constants (\cref{sec:llm}). \Cref{sec:taxonomy} reads twenty deployed schemes
through the profile. \Cref{sec:related} positions the paper against prior and concurrent work, \cref{sec:discussion}
draws consequences and states the open problems, and \cref{sec:conclusion} concludes. Proofs are in
\cref{sec:appendix}, which also contains the imported-results glossary (\cref{sec:imported}), the
scheme-by-scheme fit of the deployed watermarks (\cref{sec:schemefit}), and a broader survey of adjacent
work (\cref{sec:morerelated}).

\section{Model: the forensic ladder}
\label{sec:prelim}

The model has three ingredients: a carrier, the laws that generate it, and a keyed family that secretly
selects among those laws; every forensic task is a question about the selection.

\paragraph{Carrier, model, keys.} Tokens are drawn from an alphabet $\Sigma$ of size $q:=|\Sigma|$; a
carrier is a sequence $X\in\Xs:=\Sigma^n$. An unmarked model induces the natural law $\pnat$ on $\Xs$. A
\emph{keyed watermarking scheme} draws a secret key $K\sim\mu$ from a key space $\Keys$ and generates
$X\sim\pwm^{(K)}$; write $\pbar:=\E_K\,\pwm^{(K)}$ for the key-averaged output law. The \emph{secret}
$S$ (a user identity in $[N]$, an $\ell$-bit payload, or both) determines the key, $K=K(S)$, so
$S\to K\to X$ is a Markov chain. The primary object throughout is the family of \emph{per-secret kernels}
$\pwm(\cdot\mid s,x_{<t})$ (the conditional next-token law under secret $s$ given the prefix $x_{<t}$)
from which key-averaged quantities are derived, never the reverse. Decoders are informed: an attribution
decoder knows $\pnat$ and the key registry $\kappa:=\{k_1,\dots,k_N\}$, and false-positive statements are
averaged over the i.i.d.\ draw of those keys. Recoverable information is measured against this side
information: throughout, the per-secret quantities --- the kernels $\pwm(\cdot\mid S,x_{<t})$, the information
profile $\nu$, the distortion profile $\bar d$, and the budget $\Phibud$ --- are read conditionally on the
registry $\kappa$, written $I(S;X)$, $\bar d(t)$, and so on for brevity. Given $\kappa$ the secret determines
its key, $K=k_S$, so each per-secret object is the \emph{realized} per-key one and every average over $S$
ranges over the enrolled keys. This conditioning is vacuous for a fixed registry (a single enrolled key,
detection, or deterministic extraction). It matters for any randomized registry with exchangeable keys ---
bounded schemes included, since marginally over the key draw the carrier is independent of the secret --- and
is sharpest for the distortion-free witness of \cref{thm:entropyclose}, whose output is exactly $\pnat$: there
the unconditional $I(S;X)=0$ while the registry-conditional budget is positive, and it is that conditional
quantity the converse bounds by $nh_n$. All logarithms are natural unless a base is shown;
$\KL{\cdot}{\cdot}$ is relative entropy and $I(\cdot;\cdot)$ mutual information, both in nats.

\paragraph{Quality as a budget.} Fidelity to the base model is what makes a watermark usable, and it is what
constrains forensic power. We consider the two quality models that occur in practice, both stated on the
per-secret kernels.

Write $\bar d(t):=\E_{X_{<t}}\E_S\,\KL{\pwm(\cdot\mid S,X_{<t})}{\pnat(\cdot\mid X_{<t})}$ for the
\emph{per-secret distortion profile} (expectations under the joint marked law), so by the chain rule
$\sum_{t=1}^n\bar d(t)=\E_S\,\KL{\pwm(\cdot\mid S)}{\pnat}$, and
$T:=\{t:\pwm(\cdot\mid s,x_{<t})\ne\pnat(\cdot\mid x_{<t})\text{ for some }s,x_{<t}\}$ for the
\emph{per-secret footprint}, the set of marked positions. A natural alternative is the key-averaged marginal
profile $\delta_{\mathrm{marg}}(t):=\E_{X_{<t}}\,\KL{\E_S\,\pwm(\cdot\mid S,X_{<t})}{\pnat(\cdot\mid X_{<t})}$;
joint convexity of relative entropy gives the bridge $\delta_{\mathrm{marg}}(t)\le\bar d(t)$, and the
inequality is strict in general (\cref{app:witnesses}). We retire the marginal profile because it cannot see
information hidden in correlations across positions: the two-token parity scheme (\cref{app:witnesses}) has
zero marginal distortion at every position, yet a decoder reads one full bit from it ($\nu(2)=\log2$ below).
Every cap below is therefore stated on the per-secret kernels.

\begin{assumption}[Bounded-distortion quality; the class $\mathcal B(\Delta)$]\label{asm:bounded}
There is a per-token budget $\Delta\ge0$, imposed in one of two forms. \emph{Sup-form:}
$\KL{\pwm(\cdot\mid s,x_{<t})}{\pnat(\cdot\mid x_{<t})}\le\Delta$ for every secret $s$, prefix $x_{<t}$, and
position $t$. \emph{Avg-form:} $\max_t\bar d(t)\le\Delta$. The sup-form implies the avg-form, and we tag each
claim with the form it needs: every converse below needs only the avg-form, while the exponential-tilting
achievability witness satisfies the sup-form. By the chain rule the global divergence inherits the cap:
$\KL{\pwm^{(k)}}{\pnat}\le n\Delta$ for every key $k$ under the sup-form, and
$\E_S\,\KL{\pwm(\cdot\mid S)}{\pnat}=\sum_t\bar d(t)\le n\Delta$ under the avg-form. Each token stays
$\Delta$-close to natural on average; the mark is \emph{delocalized}.
\end{assumption}

\begin{assumption}[Distortion-free quality]\label{asm:free}
The mixture is undetectable without the key, at one of two levels. \emph{Sequence level:} $\pbar=\pnat$
(statistically) or $\pbar\approx\pnat$ up to a negligible distinguishing advantage (computationally) --- the
operational statement: no key-less observer detects the mark. \emph{Kernel level:}
$\E_K\,\pwm(\cdot\mid K,x_{<t})=\pnat(\cdot\mid x_{<t})$ for every prefix $x_{<t}$ --- the structural
statement, satisfied by the exact-sampler family (inverse-CDF / exponential-minimum sampling;
\citealp{kuditipudi2023robust}; \cref{sec:taxonomy}), and implying the sequence level
whenever the key coordinates consumed at distinct positions are position-independent. The per-key divergence
$\KL{\pwm^{(k)}}{\pnat}$ is unconstrained; the binding resource is the generation entropy
$h_n:=H(\pnat^{(n)})/n$, where $\pnat^{(n)}$ is the length-$n$ law; for stationary sources $h_n$ decreases to
the entropy rate $h:=\lim_n h_n$, and finite-$n$ statements use $h_n$, asymptotic ones $h$.
We state for each result which level it uses.
\end{assumption}

\Cref{asm:bounded} is the regime of green-list and tilting schemes \citep{kirchenbauer2023watermark};
\cref{asm:free} is the regime of distortion-free and undetectable schemes \citep{kuditipudi2023robust,
christ2024undetectable, christ2024pseudorandom, dathathri2024synthid}.

For the localization results (\cref{sec:results}, Level 3) we use one structural hypothesis, satisfied by the
deployed schemes.

\begin{assumption}[$r$-local / edit-based]\label{asm:local}
There is a set $R\subseteq[n]$ (the \emph{carrier footprint}), a constant $r$, and a coupling of the marked
carrier $X$ with an unmarked carrier $X'\sim\pnat$ such that $X$ and $X'$ agree on every coordinate outside the
dilated footprint $R^{+}:=\{i:\mathrm{dist}(i,R)\le r\}$; for a scheme in the footprint-capped class
$\mathcal E(m)$ of \cref{def:embedbias} one may take $R=T$ when the natural law factorizes across positions
(off-$T$ kernels then depend on no marked prefix; a perfectly correlated source admits no such coupling).
Equivalently, the watermark is a local edit
confined to a neighbourhood of $R$ and does not propagate elsewhere. This is the post-hoc regime in which
Level-3 localization is actually studied: image, audio, and patch watermarks edit only the marked region and
leave the rest of the carrier as generated (a biasing scheme such as green-list instead has $R=[n]$, so
$R^{+}=[n]$ and the assumption is vacuous). The consequence we use is immediate: any contiguous window
disjoint from $R^{+}$ has $x_W=x'_W$, distributed exactly as $\pnat$.
\end{assumption}

\paragraph{The two mechanisms.} The distinction the introduction drew between an \emph{object inside} the
carrier and a \emph{property} of its law has a clean operational form, which we take from the ontology of
generative watermarking \citep[after the steganography trichotomy of][]{fridrich2009stego} and state
informally, since no result below invokes it directly; the operative cut is \cref{def:embedbias}. An
\emph{embedding} watermark is a triple $(\mathsf{Gen},\mathsf{Embed},\mathsf{Ext})$ with a carrier-level
object space $\mathcal W$, maps $\mathsf{Embed}:\Xs\times\Keys\times\mathcal W\to\Xs$ and
$\mathsf{Ext}:\Xs\times\Keys\to\mathcal W\cup\{\bot\}$ with $\mathsf{Ext}(\mathsf{Embed}(x,k,w),k)=w$ for
almost all $(k,w)$, and a \emph{localization map} $\Lmap_k:\Xs\to 2^{[n]}$ through which alone
$\mathsf{Ext}$ reads $x$: the watermark is an object on $\Lmap_k(x)$. A \emph{biasing} watermark is a pair
$(\mathsf{Gen}_{\mathrm{wm}},\mathsf{Det})$, where $\mathsf{Gen}_{\mathrm{wm}}$ induces $\pwm^{(k)}$ with
$\KL{\pwm^{(k)}}{\pnat}>0$ (or $=0$, indistinguishable only against a bounded distinguisher) and
$\mathsf{Det}:\Xs\times\Keys\to\{0,1\}$ is a hypothesis test realising a Type-I/Type-II trade-off: the
watermark is a property of the law $\pwm^{(k)}$. The two mechanisms come equipped with different native
operations ($\mathsf{Ext}$ recovers an object at Level 2, $\mathsf{Det}$ runs a test at Level 0), which is
the seed of the ladder below. What this object-versus-property distinction says about what a watermark
\emph{is} we take up, as a closing reflection, in \cref{sec:philosophy}.

\paragraph{The operative refinement: footprint.} For the sample-complexity analysis, the object-level
reading leaves one thing under-determined: its localization map constrains where the secret is \emph{read from}, but
not the \emph{size} of $\Lmap_k(x)$. A green-list secret is readable from a long enough window, so the
read-from criterion alone would miscount it as embedding. The ``object inside'' intuition needs the region to
be \emph{small}. We therefore read the distinction through the per-secret footprint $T$.

\begin{definition}[Footprint; embedding vs.\ biasing]\label{def:embedbias}
A scheme is \emph{embedding} if $|T|=o(n)$, so the per-secret kernels are marked on a vanishing fraction of
coordinates, and $\Lmap_k:=T$ is then a localization map (off $T$ the next-token law is exactly natural given
the past, for every secret); it is \emph{biasing} if $|T|=\Omega(n)$: a constant fraction of coordinates carry
distortion for some secret, so no $o(n)$ region contains the mark and it lives only in the aggregate. The
footprint-capped class is $\mathcal E(m):=\{|T|\le m\}$; together with $\mathcal B(\Delta)$ of
\cref{asm:bounded} it is read in a fixed dictionary and coordinate ordering (see the remarks after
\cref{prop:profile}).
\end{definition}

\begin{remark}[Model change: per-secret, not key-averaged]
Both classes are cut on the per-secret kernels, not on the key-averaged kernel $\E_S\,\pwm(\cdot\mid S,x_{<t})$:
by the convexity bridge they are strict subclasses of their marginal counterparts, and the marginal cut
misclassifies --- every kernel-level distortion-free scheme (\cref{asm:free}) has key-averaged kernels exactly
natural, so the key-averaged support would place all of them in $\mathcal E(0)$ while they carry $\Phibud(n)>0$.
The two cuts separate in both directions (\cref{app:witnesses}): writing $\nu$ for the information profile of
\cref{prop:profile} below, the two-token parity scheme has $\nu(2)=\log2$ while $\delta_{\mathrm{marg}}\equiv0$,
and the key-independent kernel-bend has $\delta_{\mathrm{marg}}>0$ while $\nu\equiv0$.
\end{remark}

This refinement does real work, and it splits the embedding ontology's localization map into \emph{two}
orthogonal axes. The \emph{footprint} $T$ governs whether forensic cost decouples from the
distortion budget (\cref{thm:dichotomy}); the \emph{readout resolution} of $\mathsf{Ext}$, namely whether the
secret is decoded globally or per-region, governs localization (\cref{thm:crop}), and is independent of the
footprint. The two come apart on real schemes: green-list \citep{kirchenbauer2023watermark} is biasing
($T=[n]$) although its secret is recoverable from a window; and SEAL or Gaussian Shading, which an
object-level reading would file as embedding, are full-support hence biasing, localizing (when they do) by
fine readout rather than by a small footprint (\cref{sec:taxonomy}). The
distortion-free family (Kuditipudi, Christ--Gunn) is likewise full-support biasing, living in the entropy
regime of \cref{asm:free} rather than the $\Delta$ regime.

\subsection{The information profile}

The per-secret kernels carry one further object, and it organizes everything that follows: the
\emph{information profile} of a scheme, the measure on positions
\[
\nu(t):=I(S;X_t\mid X_{<t}),\qquad \nu(A):=\sum_{t\in A}\nu(t),\qquad \norm{\nu}_1:=\nu([n]).
\]

\begin{restatable}[Information profile]{proposition}{PropProfile}\label{prop:profile}
For every keyed scheme:
\begin{enumerate}[label=(\roman*)]
\item \emph{(Mass identity.)} $\norm{\nu}_1=I(S;X)=\Phibud(n)$.
\item \emph{($\mathcal B(\Delta)$ is an $L^\infty$ cap.)} $\nu(t)\le\bar d(t)$ for every $t$; hence under
\cref{asm:bounded}, in either form, $\max_t\nu(t)\le\Delta$.
\item \emph{($\mathcal E(m)$ is an $L^0$ cap.)} If the scheme lies in $\mathcal E(m)$ (\cref{def:embedbias}:
per-secret kernels unmarked off $T$, $|T|\le m$, in the fixed dictionary and ordering), then $\nu(t)=0$ for
$t\notin T$; hence $\supp(\nu)\subseteq T$ and $\Phibud(n)=\sum_{t\in T}\nu(t)\le m\log q$.
\end{enumerate}
\end{restatable}

\begin{proof}
(i) is the chain rule for mutual information. (ii): condition on $X_{<t}=x_{<t}$ and apply the golden formula
(\cref{fact:golden}) with reference $\pnat(\cdot\mid x_{<t})$, giving
$I(S;X_t\mid X_{<t}=x_{<t})\le\E_S\,\KL{\pwm(\cdot\mid S,x_{<t})}{\pnat(\cdot\mid x_{<t})}$; averaging over
$x_{<t}$ gives $\nu(t)\le\bar d(t)$. (iii): for $t\notin T$ the kernel
$\pwm(\cdot\mid s,x_{<t})=\pnat(\cdot\mid x_{<t})$ does not depend on $s$, so $X_t\perp S\mid X_{<t}$
(given the past, the kernel at $t$ is the natural law for every secret, so $X_t$ carries no
secret-information there; whatever earlier marked positions did to shape $X_{<t}$ was already counted at those
positions) and $\nu(t)=0$; each surviving term
obeys $\nu(t)\le H(X_t\mid X_{<t})\le\log q$.
\end{proof}

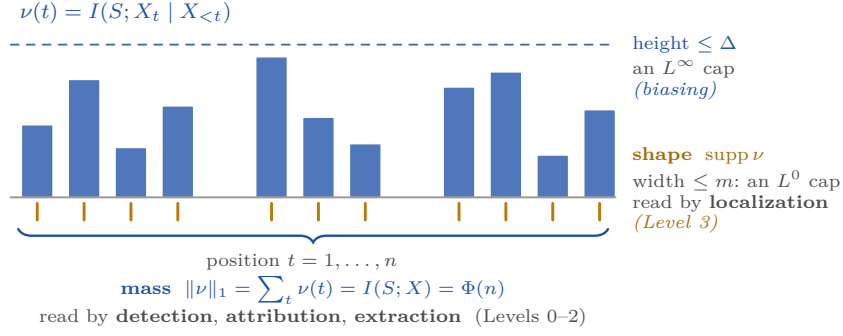
\begin{figure}[t]
\centering
\begin{tikzpicture}[font=\sffamily,>=Stealth,line cap=round]
  \def\dx{0.62}
  \foreach \i/\h in {0/0.95,1/1.55,2/0.65,3/1.20,5/1.85,6/1.05,7/0.70,9/1.45,10/1.65,11/0.55,12/1.15}{
    \fill[detectblue!78,rounded corners=0.3pt] ({\i*\dx},0) rectangle ({\i*\dx+0.40},\h);
  }
  \draw[inkgray!55,line width=0.8pt] (-0.15,0) -- (7.85,0);
  \node[anchor=south west,font=\footnotesize\bfseries,text=detectblue] at (-0.15,2.18)
     {$\nu(t)=I(S;X_t\mid X_{<t})$};
  \node[anchor=north,font=\scriptsize,text=inkgray] at (3.7,-0.62) {position $t=1,\dots,n$};
  \foreach \i in {0,1,2,3,5,6,7,9,10,11,12}{
    \draw[attribgold,line width=1.1pt] ({\i*\dx+0.20},-0.08) -- ({\i*\dx+0.20},-0.30);
  }
  \draw[dashed,detectblue!85,line width=0.7pt] (-0.15,2.02) -- (7.85,2.02);
  \node[anchor=west,font=\scriptsize,text=detectblue] at (7.95,2.02) {height $\le\Delta$};
  \node[anchor=west,font=\scriptsize,text=inkgray]   at (7.95,1.70) {an $L^\infty$ cap};
  \node[anchor=west,font=\scriptsize\itshape,text=detectblue] at (7.95,1.42) {(biasing)};

  \draw[decorate,decoration={brace,amplitude=5pt,mirror},detectblue,line width=0.9pt]
     (-0.05,-0.42) -- (7.75,-0.42);
  \node[anchor=north,align=center,font=\scriptsize,text=detectblue] at (3.85,-0.92)
     {{\bfseries mass} $\;\|\nu\|_1=\textstyle\sum_t\nu(t)=I(S;X)=\Phi(n)$};
  \node[anchor=north,align=center,font=\scriptsize,text=inkgray] at (3.85,-1.34)
     {read by \textbf{detection}, \textbf{attribution}, \textbf{extraction} \;(Levels 0--2)};

  \node[anchor=west,align=left,font=\scriptsize,text=attribgold!88!black] at (7.95,0.55)
     {{\bfseries shape} $\;\operatorname{supp}\nu$};
  \node[anchor=west,align=left,font=\scriptsize,text=inkgray] at (7.95,0.22)
     {width $\le m$: an $L^0$ cap};
  \node[anchor=west,align=left,font=\scriptsize,text=inkgray] at (7.95,-0.08)
     {read by \textbf{localization}};
  \node[anchor=west,align=left,font=\scriptsize\itshape,text=attribgold!88!black] at (7.95,-0.34)
     {(Level 3)};
\end{tikzpicture}
\caption{\textbf{One profile, read two ways} (\cref{prop:profile}). The information profile
$\nu(t)=I(S;X_t\mid X_{<t})$ is a measure on positions, and a measure has a \emph{total} and a
\emph{support}. The forensic tasks above detection are exactly these two functionals: attribution and
extraction read the \textbf{mass} $\|\nu\|_1=I(S;X)=\Phibud(n)$, which the forensic-recovery budget caps
(\cref{thm:budget}); localization reads the \textbf{shape} $\operatorname{supp}\nu$; detection is the
mechanism-blind baseline, priced by presence. The two quality
models of the literature are then two caps on the \emph{same} $\nu$: a height cap $\nu(t)\le\Delta$
($L^\infty$, \cref{prop:profile}(ii), biasing) and a support cap $|\operatorname{supp}\nu|\le m$
($L^0$, \cref{prop:profile}(iii), embedding).}
\label{fig:profile}
\end{figure}

\begin{remark}[Two caps, one object]
\Cref{prop:profile} reads the two quality models as norm caps on the same measure: $\mathcal B(\Delta)$ caps
$\norm{\nu}_\infty$ and $\mathcal E(m)$ caps $\norm{\nu}_0$, the support size; their incomparability is now
definitional (different norms on one object) rather than a contingency of two ontologies (\cref{fig:profile}).
How the forensic
tasks read $\nu$ (mass for Levels 1--2, shape plus the edit coupling for Level 3; detection, Level 0, reads
neither) is spelled out under \emph{Mass, shape, robustness} below.
\end{remark}

\begin{remark}[Basis and ordering relativity]
$\nu$ is computed in a coordinate system: a dictionary in which $X$ is written and an order in which its coordinates
are revealed. Only the mass $\norm{\nu}_1=I(S;X)$ is invariant; the $L^\infty$ and $L^0$ caps (and with them the
embedding/biasing cut of \cref{def:embedbias}) are stated relative to a fixed dictionary and ordering, declared
once per scheme (for Tree-Ring, the latent-Fourier dictionary; \cref{sec:taxonomy}). This is the same relativity as
sparsity, and we treat it the same way: a property of the representation, fixed in advance.
\end{remark}

\paragraph{Mass, shape, robustness.} Attribution and extraction read the \emph{mass}
$\norm{\nu}_1$ (\cref{def:budget}); localization reads the \emph{shape} $\supp(\nu)$, together with the edit
coupling of \cref{asm:local} (\cref{thm:crop}); detection, the mechanism-blind baseline, reads neither --- a
key-blind test is priced by the mixture divergence $\KL{\pbar}{\pnat}$, which a scheme can drive to zero while
keeping $\norm{\nu}_1>0$ (the parity witness), and a keyed test by the per-key divergence or soundness
($\Theta(1/\Delta)$, $\Theta(\lambda/h)$). Mass and shape are properties of the per-secret kernels alone.
Crop-\emph{robustness} is not: it lives on the coupling footprint $R$ of \cref{asm:local}, not on $\supp(\nu)$.
The two witnesses of \cref{app:witnesses} separate the notions in both directions: the kernel-bend has
$\nu\equiv0$ (nothing recoverable anywhere) yet a window covering the bent coordinate is visibly
non-natural; the parity scheme has $\nu(2)=\log2$ yet every window marginal is exactly natural. \Cref{thm:crop}
is accordingly stated on $R$: mass and shape are read from $\nu$, robustness from the edit coupling.

\paragraph{The forensic levels.} The forensic uses of a watermark form a ladder of statistical problems.

\begin{definition}[Forensic tasks]\label{def:tasks}
Given a length-$n$ sample (and, for keyed detectors, the key), the forensic tasks are:
\begin{itemize}
  \item \textbf{Level 0 (Detection).} Test $H_0:X\sim\pnat$ against $H_1:X\sim\pwm$, with Type-I/Type-II
  errors $(\alpha,\beta)$.
  \item \textbf{Level 1 (Attribution).} With $N$ users, keys $k_1,\dots,k_N$, $U\sim\Unif[N]$ and
  $X\sim\pwm^{(k_U)}$, output $\hat U(X)\in[N]$, with error $P_e:=\Pr[\hat U\ne U]$.
  \item \textbf{Level 2 (Extraction).} A payload $W\sim\Unif\{0,1\}^\ell$ is fixed at generation; recover
  $\hat W(X)$, with error $\Pr[\hat W\ne W]$.
  \item \textbf{Level 3 (Localization).} Identify the marked window, robustly to cropping to any window of
  size $\ge w$ (\cref{sec:results}).
\end{itemize}
We write $\ntask{\det},\ntask{\att},\ntask{\ext}$ for the least $n$ achieving a fixed target error.
\end{definition}

\paragraph{The forensic-recovery budget.} The single quantity that drives every lower bound is the
information a scheme makes recoverable about its secret.

\begin{definition}[Forensic-recovery budget]\label{def:budget}
The \emph{forensic-recovery budget} of a scheme is $\Phibud(n):=I(S;X)=\norm{\nu}_1$
(\cref{prop:profile}(i); the mutual information is registry-conditional, per the convention above): the total
mass of the information profile. By data processing, any informed extractor $\hat S=g(X,\kappa)$ satisfies
$I(S;\hat S\mid\kappa)\le\Phibud(n)$, so $\Phibud(n)$ upper-bounds the recoverable information
of \emph{every} Level-1 / Level-2 procedure simultaneously.
\end{definition}

\noindent We use $\Unif$ for the uniform distribution and write $a=\Theta(b)$ in the usual two-sided sense,
hiding universal constants. The standing convention is that the detection rung of \cref{def:tasks} is the
imported baseline $\ntask{\det}=\Theta(1/\Delta)$ \citep{cai2024better,li2024statistical,huang2023optimal}
(\cref{fact:det}); all of our results concern the rungs above it. Throughout, $\lambda$ denotes the soundness
(security) parameter of distortion-free detection: the keyed detector's false-positive probability is at most
$2^{-\lambda}$; it is distinct from the per-innocent attribution level $\delta$ of \cref{thm:entropyclose},
which parametrizes a different task.

\subsection{Useful tools}
\label{sec:tools}
The whole paper rests on three textbook facts, which we use as machinery and record once.

\begin{fact}[Data processing and the golden formula; \citealp{cover2006elements}]\label{fact:golden}
If $S\to K\to X$ is a Markov chain then $I(S;X)\le I(K;X)$. Moreover, for a key $K$ with $X\sim P_{X\mid K}$
and marginal $P_X$, $I(K;X)=\min_Q \E_K\,\KL{P_{X\mid K}}{Q}$, with the minimum attained at $Q=P_X$; hence
$I(K;X)\le\E_K\,\KL{P_{X\mid K}}{Q}$ for every reference $Q$.
\end{fact}

\begin{fact}[Fano's inequality; \citealp{cover2006elements}]\label{fact:fano}
For $U$ on a set of size $N$ and any estimator $\hat U=g(X)$ with error $P_e=\Pr[\hat U\ne U]$,
$H(U\mid X)\le \Hb(P_e)+P_e\log(N-1)$, where $\Hb$ is binary entropy. Consequently, if $U\sim\Unif[N]$,
$P_e\ge 1-(I(U;X)+\log 2)/\log N$.
\end{fact}

\begin{fact}[Stein's lemma and its ergodic extension; \citealp{cover2006elements, barron1985density}]\label{fact:stein}
(i) For testing $P$ against $Q$ from $n$ i.i.d.\ samples at fixed Type-I level $\alpha\in(0,1)$, the smallest
Type-II error $\beta_n$ obeys $-\tfrac1n\log\beta_n\to\KL{P}{Q}$ \citep{cover2006elements}. (ii) More
generally, if $P,Q$ are the length-$n$ marginals of stationary ergodic laws with positive KL-rate
$\bar D:=\lim_n\tfrac1n\KL{P}{Q}>0$ and an $L^1$-domination condition on the log-density, then the generalized
(relative-entropy) Shannon--McMillan--Breiman theorem \citep{barron1985density} gives
$\tfrac1n\log\tfrac{dP}{dQ}\to\bar D$ almost surely under $P$, so the optimal test of $P$ against $Q$ on a
single length-$n$ sample has error $\to0$. A merely diverging \emph{total} KL on non-i.i.d.\ laws does not
suffice: it is the per-symbol rate $\bar D$, under that domination condition, that must be bounded below.
\end{fact}

Three further results that we \emph{import} rather than use as machinery (the Level-0 detection rate, the
distortion-constrained fingerprinting converse, and the general-channel coding theorem) are stated, with the
non-standard terminology they involve, in \cref{sec:imported}.

\section{Results}
\label{sec:results}

One object organizes the results: the information profile $\nu(t):=I(S;X_t\mid X_{<t})$ of
\cref{prop:profile} (\cref{sec:prelim}). Attribution and extraction read its \emph{mass}
$\|\nu\|_1=I(S;X)=\Phibud(n)$; localization reads its \emph{shape} $\supp(\nu)$ together with the edit
coupling; detection is the mechanism-blind baseline, priced by presence rather than the mass. Everything below
is a statement about $\nu$: its \emph{mass} funds attribution and extraction and sets them against the
detection baseline (\cref{thm:budget,thm:separation,thm:dichotomy,thm:entropyclose}); its \emph{shape} governs
localization (\cref{thm:crop}); and the levels draw on one shared budget (\cref{prop:region}). Full proofs are in
\cref{sec:appendix}, and we are explicit throughout about which step is imported and which is ours.

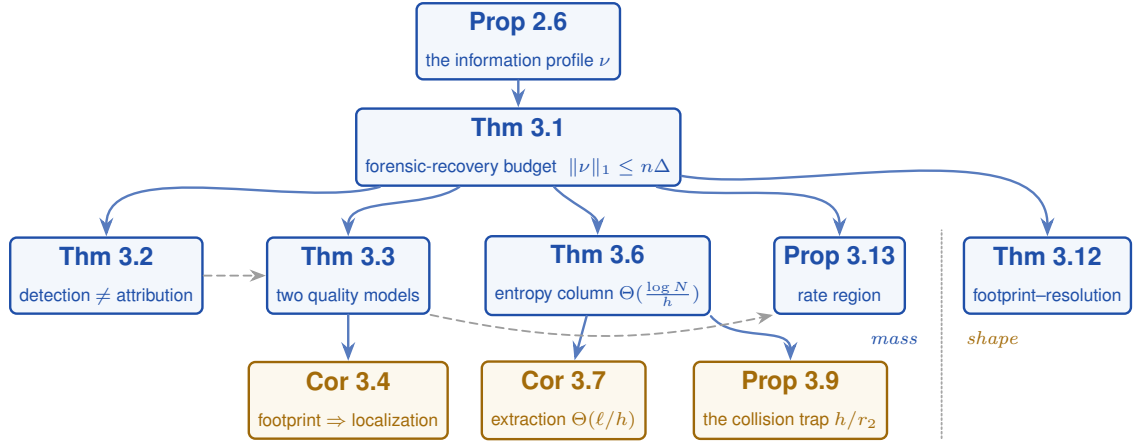
\begin{figure}[t]
\centering
\begin{tikzpicture}[font=\sffamily,>=Stealth,line cap=round,
  rbox/.style={rounded corners=2.5pt,draw=detectblue,line width=0.9pt,fill=bluetint,
               align=center,inner sep=3.5pt,text=detectblue},
  cbox/.style={rounded corners=2.5pt,draw=attribgold!85!black,line width=0.9pt,fill=goldtint,
               align=center,inner sep=3.5pt,text=attribgold!88!black},
  dep/.style={->,detectblue!75,line width=0.9pt},
  sec/.style={->,inkgray!55,line width=0.7pt,densely dashed},
  ]
  \node[rbox] (prof) at (7.0,5.45) {\textbf{Prop~\ref*{prop:profile}}\\[1pt]{\scriptsize the information profile $\nu$}};
  \node[rbox] (bud)  at (7.0,4.05) {\textbf{Thm~\ref*{thm:budget}}\\[1pt]{\scriptsize forensic-recovery budget $\;\|\nu\|_1\le n\Delta$}};
  \node[rbox] (sep)  at (1.55,2.35) {\textbf{Thm~\ref*{thm:separation}}\\[1pt]{\scriptsize detection $\ne$ attribution}};
  \node[rbox] (dic)  at (4.75,2.35) {\textbf{Thm~\ref*{thm:dichotomy}}\\[1pt]{\scriptsize two quality models}};
  \node[rbox] (ent)  at (8.05,2.35) {\textbf{Thm~\ref*{thm:entropyclose}}\\[1pt]{\scriptsize entropy column $\Theta(\tfrac{\log N}{h})$}};
  \node[rbox] (reg)  at (11.25,2.35){\textbf{Prop~\ref*{prop:region}}\\[1pt]{\scriptsize rate region}};
  \node[rbox] (crop) at (14.0,2.35) {\textbf{Thm~\ref*{thm:crop}}\\[1pt]{\scriptsize footprint--resolution}};
  \node[cbox] (brg)  at (4.75,0.7) {\textbf{Cor~\ref*{cor:bridge}}\\[1pt]{\scriptsize footprint $\Rightarrow$ localization}};
  \node[cbox] (ext)  at (7.6,0.7) {\textbf{Cor~\ref*{cor:extraction-h}}\\[1pt]{\scriptsize extraction $\Theta(\ell/h)$}};
  \node[cbox] (trap) at (10.6,0.7) {\textbf{Prop~\ref*{rem:listsize}}\\[1pt]{\scriptsize the collision trap $h/r_2$}};
  \draw[dep] (prof) -- (bud);
  \draw[dep] (bud) to[out=196,in=90,looseness=0.75] (sep.north);
  \draw[dep] (bud) to[out=214,in=90,looseness=0.85] (dic.north);
  \draw[dep] (bud) to[out=312,in=90,looseness=0.85] (ent.north);
  \draw[dep] (bud) to[out=344,in=90,looseness=0.85] (reg.north);
  \draw[dep] (bud) to[out=350,in=90,looseness=0.75] (crop.north);
  \draw[dep] (dic) -- (brg);
  \draw[dep] (ent) -- (ext);
  \draw[dep] (ent.south east) to[out=-50,in=90] (trap.north);
  \draw[sec] (sep) -- (dic);     
  \draw[sec] (dic.south east) to[bend right=12] (reg.south west);  
  \draw[densely dotted,inkgray!50,line width=0.7pt] (12.6,0.2) -- (12.6,2.95);
  \node[font=\scriptsize\itshape,text=detectblue,anchor=east]        at (12.4,1.5) {mass};
  \node[font=\scriptsize\itshape,text=attribgold!88!black,anchor=west] at (12.8,1.5) {shape};
\end{tikzpicture}
\caption{\textbf{The results at a glance.} Every result is a reading of the information profile $\nu$. The
forensic-recovery budget (\cref{thm:budget}) caps its mass $\|\nu\|_1=I(S;X)$; the four results left of the divider
concern that mass budget (the detect/attribute separation; the two quality models; the closed entropy column; the shared rate
region), while \cref{thm:crop} reads its shape. Solid arrows mark a direct dependence, dashed arrows a
secondary use (\cref{thm:separation} supplies the dichotomy's lower bound; \cref{thm:dichotomy}(a) supplies the
rate region's achievability). Corollaries and the collision-trap proposition hang under their parents; the
imported and textbook facts each proof
draws on are catalogued in \cref{sec:imported}.}
\label{fig:resultmap}
\end{figure}

\subsection{The forensic-recovery budget}

Each forensic procedure is funded from a single account, the mass $\norm{\nu}_1=I(S;X)$ of the profile, and
the quality models cap that account in three different norms. The theorem reads off all three.

\begin{restatable}[Forensic-recovery budget]{theorem}{ThmBudget}\label{thm:budget}
Let a keyed scheme generate $X\in\Sigma^n$ from a secret $S$ via a key $K=K(S)$, and set
$\Phibud(n):=I(S;X)=\sup_g I(S;g(X,\kappa)\mid\kappa)$ (informed decoders, registry-conditional per the
convention of \cref{sec:prelim}). Then (a) under \cref{asm:bounded}, $\Phibud(n)\le n\Delta$; (b)
under \cref{asm:free} with statistical undetectability, $\Phibud(n)\le n h_n$ ($h_n\downarrow h$ for
stationary sources; asymptotic statements may use $h$); and (c) under the footprint cap
$\mathcal E(m)$ of \cref{sec:prelim}, with marked set $T$ ($|T|\le m$),
$\Phibud(n)=\sum_{t\in T}\nu(t)\le m\log q$ (\cref{prop:profile}(iii)). In particular every Level-1 and
Level-2 procedure recovers at most $\Phibud(n)$ nats about the secret, irrespective of $N$ or $\ell$.
\end{restatable}

\begin{proof}[Proof sketch]
$S\to K\to X$ is Markov, so $I(S;X)\le I(K;X)\le\E_K\KL{\pwm^{(K)}}{Q}$ for every reference $Q$ (golden
formula). Take $Q=\pnat$ with \cref{asm:bounded} for (a); for (b), when $\pbar=\pnat$ there is no divergence
left to spend, so no reference beats the source itself and $I(K;X)\le H(X)=nh_n$.
For (c), \cref{prop:profile}(iii) gives $\nu(t)=0$ off $T$, and each remaining term is at most
$H(X_t\mid X_{<t})\le\log q$.
\end{proof}

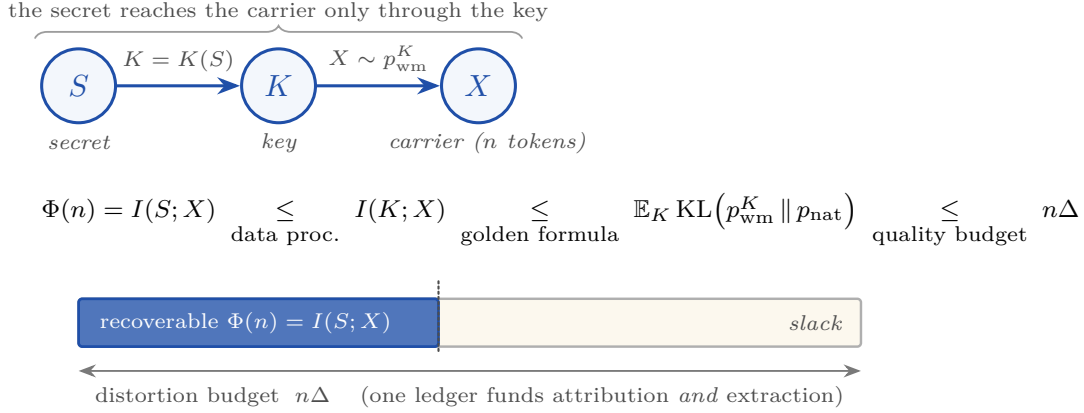
\begin{figure}[t]
\centering
\scalebox{1.15}{
\begin{tikzpicture}[font=\sffamily,>=Stealth,line cap=round,
   nd/.style={circle,draw=detectblue,line width=1pt,fill=bluetint,minimum size=8.5mm,
              font=\normalsize\bfseries,text=detectblue,inner sep=0pt}]
  \node[nd] (S) at (0,3.0)   {$S$};
  \node[nd] (K) at (2.3,3.0) {$K$};
  \node[nd] (X) at (4.6,3.0) {$X$};
  \draw[->,detectblue,line width=1pt] (S)--(K) node[midway,above=1pt,font=\scriptsize,text=inkgray]{$K=K(S)$};
  \draw[->,detectblue,line width=1pt] (K)--(X) node[midway,above=1pt,font=\scriptsize,text=inkgray]{$X\sim\pwm^{K}$};
  \node[font=\scriptsize\itshape,text=inkgray] at (0,2.35)   {secret};
  \node[font=\scriptsize\itshape,text=inkgray] at (2.3,2.35) {key};
  \node[font=\scriptsize\itshape,text=inkgray] at (4.7,2.35) {carrier ($n$ tokens)};
  \draw[decorate,decoration={brace,amplitude=4pt},inkgray!55,line width=0.6pt] (-0.5,3.55) -- (5.1,3.55);
  \node[anchor=south,font=\scriptsize,text=inkgray] at (2.3,3.62)
     {the secret reaches the carrier only through the key};

  \node[anchor=west,font=\footnotesize] at (-0.55,1.45)
    {$\Phibud(n)=I(S;X)\;\underset{\text{\scriptsize data proc.}}{\le}\;I(K;X)
      \;\underset{\text{\scriptsize golden formula}}{\le}\;
      \E_K\KL{\pwm^{K}}{\pnat}\;\underset{\text{\scriptsize quality budget}}{\le}\;n\Delta$};

  \def\BW{9.0}\def\fillf{0.46}
  \draw[inkgray!45,line width=0.8pt,fill=goldtint,rounded corners=1.2pt] (0,0) rectangle (\BW,0.55);
  \fill[detectblue!78,rounded corners=1.2pt] (0,0) rectangle ({\fillf*\BW},0.55);
  \draw[detectblue,line width=0.9pt,rounded corners=1.2pt] (0,0) rectangle ({\fillf*\BW},0.55);
  \node[anchor=west,font=\scriptsize,text=white] at (0.12,0.275) {recoverable $\Phibud(n)=I(S;X)$};
  \node[anchor=east,font=\scriptsize\itshape,text=inkgray] at (\BW-0.1,0.275) {slack};
  \draw[densely dotted,inkgray,line width=0.7pt] ({\fillf*\BW},-0.05)--({\fillf*\BW},0.78);
  \draw[<->,inkgray!75,line width=0.6pt] (0,-0.28)--(\BW,-0.28);
  \node[anchor=north,font=\scriptsize,text=inkgray] at (\BW/2,-0.33)
     {distortion budget $\;n\Delta$ \;\; (one ledger funds attribution \emph{and} extraction)};
\end{tikzpicture}}
\caption{\textbf{The forensic-recovery budget} (\cref{thm:budget}). The secret reaches the carrier only
through the key, so $S\to K\to X$ is a Markov chain, and three forced steps (data processing, then the
golden formula, then the quality budget) bound the recoverable information $\Phibud(n)=I(S;X)$ by the
distortion actually spent, $n\Delta$. Nothing about embedding versus biasing enters: it is one ledger,
and every nat it holds is shared between attribution and extraction (\cref{prop:region}).}
\label{fig:budget}
\end{figure}

This is a pure data-processing bound: it uses only the quality budget, never the embedding/biasing
distinction, so it holds for \emph{every} scheme. In the profile reading, (a), (b), (c) are three faces of
one quantity: the $L^\infty$ cap, the entropy ceiling, and the $L^0$ cap on the mass $\|\nu\|_1$. The
chain-rule reading $\Phibud(n)=\sum_t\nu(t)\le\E_K\sum_i d_i^{(K)}=\sum_i d_i$ (here
$d_i:=\E_K d_i^{(K)}=\bar d(i)$, the key-averaged per-token distortion) is the same budget once more: the
recoverable secret cannot exceed the total per-token distortion spent, and a single ledger funds attribution
and extraction alike (\cref{fig:budget}).

\subsection{Detection is not attribution}

Detection asks whether a mark is present; attribution asks which of $N$ keys placed it. The distinction is
not rhetorical---the two tasks spend the mass $\norm{\nu}_1$ at different rates---and the theorem below makes
the separation quantitative.

\begin{restatable}[Attribution converse and the forensic gap]{theorem}{ThmSeparation}\label{thm:separation}
Under \cref{asm:bounded}, for $N\ge2$ and any attribution rule,
\begin{equation}\label{eq:fano}
  P_e\ \ge\ 1-\frac{n\Delta+\log 2}{\log N},
\end{equation}
so $\ntask{\att}=\Omega(\log N/\Delta)$. Moreover, suppose the scheme is \emph{mixture-detectable}: the
key-unknown test of $\pbar$ against $\pnat$ reaches vanishing error by some length $\ntask{\det}^{\mathrm{mix}}$
(in the stationary-token model a positive per-token mixture divergence $\tfrac1n\KL{\pbar}{\pnat}\ge c>0$
suffices). Then for all sufficiently large $N$, on the window $\ntask{\det}^{\mathrm{mix}}\le n<(1-\eta)\log
N/\Delta$ the watermark is detectable with error $\to0$ while every attribution rule has $P_e\ge\eta-o(1)$.
\end{restatable}

\begin{proof}[Proof sketch]
\Cref{thm:budget}(a) gives $I(U;X)\le n\Delta$; Fano, $H(U\mid X)\le\log2+P_e\log N$, with $H(U\mid X)
=\log N-I(U;X)$, gives \eqref{eq:fano}. On the window $n\ge\ntask{\det}^{\mathrm{mix}}$ the scheme is
mixture-detectable (key-unknown detection error $\to0$; in the stationary-token model a positive per-token
mixture divergence $\tfrac1n\KL{\pbar}{\pnat}\ge c>0$ suffices, by the relative-entropy ergodic theorem),
while $n\Delta<(1-\eta)\log N$ forces $P_e\ge\eta-o(1)$.
\end{proof}

In one sentence: detection is a fixed two-point test, $\pbar$ against $\pnat$, whose cost does not grow with
$N$; attribution must extract $\log N$ nats of identity, and the budget funding it accrues only at rate
$\Delta$ per token, so it cannot begin until $n\Delta\gtrsim\log N$. The window between is where the watermark
is loud but anonymous.

We are careful about credit. The converse \eqref{eq:fano} is an elementary single-user Fano bound on $I(U;X)$
under the budget (derived from scratch in \cref{sec:appendix}); it is the collusion-free special case of the
\emph{distortion-constrained fingerprinting converse} of \citet{moulin2008universal} and
\citet{somekhbaruch2005capacity}, whose hard content (the worst-case collusion channel) we do not use. We
claim no novelty for it. Note also that the converse consumes only the \emph{avg-form} cap of
\cref{sec:prelim} (Fano on $\|\nu\|_1$ needs $\max_t\bar d(t)\le\Delta$, not the sup-form; the sup-form is
what the achievability constructions certify). What we add is the explicit \emph{co-located} regime: a single
biasing scheme, at a single length, detectable yet information-theoretically unattributable. The
mixture-detectability hypothesis is genuine, not free: the parity scheme on $\{0,1\}^2$ (under key $k_0$,
$(x_1,x_2)$ is uniform on $\{00,11\}$; under $k_1$, uniform on $\{01,10\}$; each has $\KLop=\log2$ within
budget, yet the mixture is uniform on $\{0,1\}^2=\pnat$ exactly) is detectable only with the key, which itself
is informative; for distortion-free schemes ``detect but not attribute'' is automatic,
not a separation. The clean takeaway survives at every $N$: attribution costs $\Theta(\log N)$ more tokens
than detection (\cref{sec:numerics}). The razor-thin endpoint, $P_e\to1$ at near-perfect detection, is the
asymptotic-in-$N$ sharpening of that same gap, not a separate phenomenon.

\subsection{Two quality models, and the localization map}

The separation above caps every bounded-distortion scheme, embedding included; it does not by itself separate
the mechanisms. The mechanism enters through which \emph{quality model} a scheme can use: the two norm
caps of \cref{sec:prelim} on the same profile $\nu$. Two orthogonal axes are in play (\cref{fig:twoaxes} in \cref{sec:taxonomy}): the
\emph{footprint}, which this subsection moves along (does the cost decouple from $\Delta$?), and the
\emph{readout resolution} of \cref{thm:crop} (can the mark be named in place?). The \emph{distributional cap}
$\mathcal B(\Delta)$ (the
sup-form per-secret cap: every kernel within $\Delta$ of natural at every prefix, hence $\nu(t)\le\Delta$
everywhere, an $L^\infty$ cap) keeps the mark subtle on every token; the \emph{footprint cap}
$\mathcal E(m)$ (per-secret kernels unmarked off a set $T$ with $|T|\le m$, hence $\nu=0$ off $T$, an
$L^0$ cap) confines it to $\le m$ coordinates, with no per-token bound. The two are definitionally
incomparable: different norms on one object.

Attribution succeeds once the cumulative profile reaches $\nu([n])\ge(1-o(1))\log N$ (Fano on the mass), so
$\ntask{\att}$ is a \emph{first-passage time} of $\nu$. The two caps force two extremal shapes, and the
dichotomy is nothing but their two first-passage lengths: a thin-and-wide $\nu\equiv\Delta$ reaches $\log N$
only at $n\sim\log N/\Delta$, while a tall-and-narrow $\nu=(\log q)\,\mathbf 1_T$ reaches it already at
$n\sim\log N/\log q$, with no $\Delta$ in sight. The theorem makes this precise.

\begin{restatable}[Two quality models]{theorem}{ThmDichotomy}\label{thm:dichotomy}
For attribution among $N$ users and $\ell$-bit extraction:
\begin{enumerate}[label=(\alph*)]
  \item \textbf{(Distributional cap couples.)} Every scheme in $\mathcal B(\Delta)$ has
  $\ntask{\att}=\Omega(\log N/\Delta)$ and $\ntask{\ext}=\Omega(\ell/\Delta)$. These are achieved, with
  $\ntask{\att}=\Theta(\log N/\Delta)$ and $\ntask{\ext}=\Theta(\ell/\Delta)$, when $\pnat$ is memoryless,
  or regular ($\pnat(\cdot\mid\text{past})\ge p_0>0$) and strongly mixing; for general $q$-ary $\pnat$ the
  matching upper bound is conjectural.
  \item \textbf{(Footprint cap decouples.)} There is a scheme in $\mathcal E(m)$ achieving
  $\ntask{\att}=\Theta(\log N/\log q)$, $\ntask{\ext}=\Theta(\ell/\log q)$ at $m=\Theta(\log N/\log q)$, with
  no dependence on any per-token cap. In a carrier of length $n\gg\ntask{\att}$ the footprint fraction
  $m/n\to0$ (and the aggregate distortion $n^{-1}\sum_t\bar d(t)\to0$ when $\pnat$ is regular,
  $\pnat(\cdot\mid\text{past})\ge p_0>0$): imperceptibility comes from the vanishing footprint, not a
  per-token bound.
\end{enumerate}
The localization map makes $\mathcal E(m)$ available: a scheme in $\mathcal E(m)$ has localization map $T$
(\cref{cor:bridge}), and a biasing scheme (no localization map) cannot lie in $\mathcal E(m)$ for $m=o(n)$,
so under a quality constraint it is confined to $\mathcal B(\Delta)$ and pays the coupled cost~(a).
\end{restatable}

\begin{proof}[Proof sketch]
(a) Lower bound: $\mathcal B(\Delta)\Rightarrow\KL{\pwm^{(k)}}{\pnat}\le n\Delta$, so \cref{thm:separation}
applies (unconditionally). Upper bound by a \emph{conditional tilting code}: tilt the natural conditional
$\pnat(\cdot\mid X_{<i})$ by $\Delta$ when $c_u[i]=1$; the KL chain rule makes the total distortion exactly
$(\#\text{ones})\Delta\le n\Delta$ for any autoregressive $\pnat$ (the construction controls KL). Fidelity is spent in KL, but the bits a
decoder reads back are measured by the symmetric \emph{Jensen--Shannon} information
$\E\,\JSD=\Theta(\Delta)$, the capacity of the per-position channel ($C\approx\Delta/4$ in the binary
witness); under the regularity/mixing hypothesis the per-position information stays $\ge c(p_0)\Delta$
uniformly and the information density concentrates (Verd\'u--Han), giving reliable
decoding at $n=\Theta(\log N/\Delta)$ above a finite-blocklength floor $n=\Omega(1/\Delta)$.
(b) The \emph{stamp}: write the user's base-$q$ digits into the first $m=\lceil\log N/\log q\rceil$ tokens,
$T=\{1,\dots,m\}$; readback is exact, the region is $\Theta(\log N/\log q)$, and footprint
$m/n\to0$ gives the fidelity.
\end{proof}

Three clarifications keep the statement honest. First, the asymmetry is not definitional: (a)'s lower bound
rides on the forensic-recovery budget (\cref{thm:budget}), not on ``no localization map.'' Second, this is \emph{not}
an iff phase boundary. A scheme admitting a localization map that \emph{also} obeys $\mathcal B(\Delta)$ still
needs $m\ge\log N/\Delta$ readable coordinates and remains $\Delta$-coupled; $\Delta$-freeness comes from the
footprint model, which the localization map \emph{enables} but does not force. Third, the witness in (b) is a
deliberately trivial cleartext digit-stamp: it is not robust to a single edit and is imperceptible only
through the $m/n\to0$ footprint, nothing like the semantic, edit-robust localization of deployed embedding
watermarks \citep{seal2025, editguard2024}. A $\Delta$-free achievability for any operational (edit-robust,
imperceptible) watermark remains open; (b) establishes only that the information-theoretic obstruction is
absent. The honest statement is a one-directional bridge (\cref{fig:dichotomy}):

\begin{restatable}[Footprint forces localization]{corollary}{CorBridge}\label{cor:bridge}
Let the scheme lie in $\mathcal E(m)$ with marked set $T$, $|T|\le m$. Then $\nu=0$ off $T$, so
$I(S;X)=\nu(T)$; the likelihood's dependence on $s$ factors through the pairs $((x_{<t},x_t))_{t\in T}$, so
$X_{\le\max T}$ is a sufficient statistic for $S$; and when $T$ is a prefix (as for the stamp) this is the
exact Markov chain $S\to X_T\to X$. Thus $\Lmap_k:=T$ localizes the mark. The converse fails: a scheme in
$\mathcal E(m)$ that also obeys the avg-form cap $\max_t\bar d(t)\le\Delta$ (which suffices, since
$\nu(t)\le\bar d(t)$; the sup-form implies it) still needs $m\ge(1-o(1))\log N/\Delta$ marked coordinates
--- localization does not buy $\Delta$-freeness.
\end{restatable}

\begin{figure}[t]
\centering
\begin{tikzpicture}[font=\sffamily,line cap=round,
  cell/.style={rounded corners=3pt,minimum width=4.3cm,minimum height=1.72cm,align=center,inner sep=3pt},
  ]
  \def\cA{2.85}\def\cB{7.35}\def\rT{2.0}\def\rB{0.1}\def\cw{4.3}\def\ch{1.72}
  \node[font=\small\bfseries,text=detectblue,align=center] at ({\cA+\cw/2},4.18)
    {pays the $\mathcal B(\Delta)$ cost\\[-1pt]{\scriptsize\textnormal{$\Delta$-coupled, \,$\Omega(\log N/\Delta)$}}};
  \node[font=\small\bfseries,text=attribgold!88!black,align=center] at ({\cB+\cw/2},4.18)
    {pays the $\mathcal E(m)$ cost\\[-1pt]{\scriptsize\textnormal{$\Delta$-free, \,$\Theta(\log N/\log q)$}}};
  \node[font=\small\bfseries,text=inkgray,align=right,anchor=east] at (\cA-0.2,{\rT+\ch/2})
    {small footprint\\$|T|=o(n)$\\{\scriptsize\textnormal{(has a map)}}};
  \node[font=\small\bfseries,text=inkgray,align=right,anchor=east] at (\cA-0.2,{\rB+\ch/2})
    {full support\\$|T|=\Omega(n)$\\{\scriptsize\textnormal{(no map)}}};
  \node[cell,fill=bluetint,draw=detectblue!55,line width=0.8pt,anchor=south west] at (\cA,\rT)
    {possible\\[1pt]{\scriptsize a small mark that still obeys $\mathcal B(\Delta)$}\\[-1pt]
     {\scriptsize pays $\Omega(\log N/\Delta)$ \;(\cref{cor:bridge})}};
  \node[cell,fill=goldtint,draw=attribgold!85!black,line width=1.5pt,anchor=south west] at (\cB,\rT)
    {{\bfseries the stamp}~$\checkmark$\\[1pt]{\scriptsize footprint \emph{decouples} the cost}\\[-1pt]
     {\scriptsize $\Theta(\log N/\log q)$, \,$\Delta$-free}};
  \node[cell,fill=bluetint,draw=detectblue!55,line width=0.8pt,anchor=south west] at (\cA,\rB)
    {{\bfseries biasing}, forced\\[1pt]{\scriptsize no footprint $\Rightarrow$ confined to $\mathcal B(\Delta)$}\\[-1pt]
     {\scriptsize pays $\Omega(\log N/\Delta)$}};
  \node[cell,fill=inkgray!8,draw=inkgray!40,line width=0.8pt,anchor=south west] at (\cB,\rB)
    {\textcolor{inkgray}{\itshape impossible}\\[1pt]{\scriptsize no small footprint}\\[-1pt]
     {\scriptsize $\Rightarrow$ no $\mathcal E(m)$ to use}};
\end{tikzpicture}
\caption{\textbf{The bridge is one-directional} (\cref{thm:dichotomy}, \cref{cor:bridge}). Crossing
``footprint'' with ``quality model used'' leaves exactly one cheap, $\Delta$-free cell, the stamp, and it
needs \emph{both} a small footprint \emph{and} actually using the footprint model $\mathcal E(m)$. A small
footprint (a localization map) \emph{enables} the cheap model but does not \emph{decide} the cost: a small
mark that still obeys the per-token cap $\mathcal B(\Delta)$ stays $\Delta$-coupled (top-left). A full-support
biasing scheme has no small footprint, so it is confined to $\mathcal B(\Delta)$ (bottom-left); the
$\Delta$-free corner is simply unreachable for it (bottom-right). This is a one-way bridge, not an
``embedding $\Leftrightarrow\Delta$-free'' equivalence.}
\label{fig:dichotomy}
\end{figure}
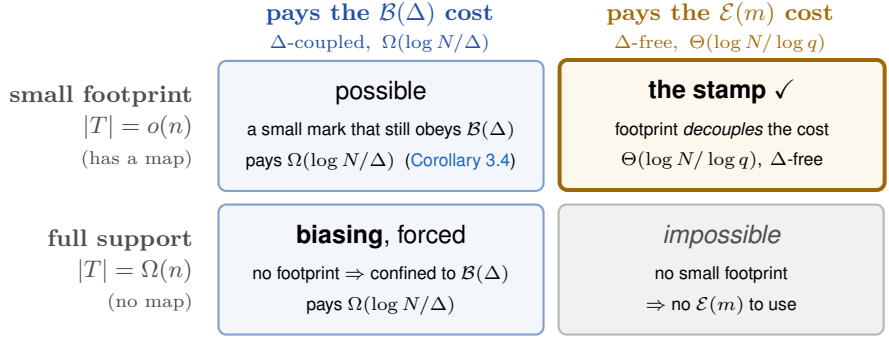

\subsection{Closing the entropy column}

In the distortion-free regime (\cref{asm:free}) the binding resource is the entropy rate $h$, and
\cref{thm:budget}(b) gives the Fano converse $\ntask{\att}=\Omega(\log N/h)$. Whether the rate is attainable
(whether a $\pbar=\pnat$ scheme can actually attribute at $n=O(\log N/h)$) is the achievability question left
open by that converse-only ledger. It closes, at full stationary-ergodic generality, by a decoder
that thresholds the realized surprisal (\cref{fig:surprisal}), giving, to our knowledge, the first tight
$h$-rate for multi-user attribution.

\begin{definition}[Witness scheme; threshold decoder]\label{def:witness}
Give the $N$ users i.i.d.\ keys, each an i.i.d.\ $\Unif[0,1]$ sequence, and mark by inverse-CDF sampling,
$X_t=F^{-1}_{\pnat(\cdot\mid X_{<t})}(k_u(t))$, so that $X\sim\pnat$ exactly (kernel-level distortion-free,
\cref{asm:free}; exponential-minimum sampling after \citealp{kuditipudi2023robust}). The decoder knows
$\pnat$ and the realized key set $\{k_1,\dots,k_N\}$. With $\tau:=\log(N/\delta)$ and
$\imath(X):=-\log\pnat(X)$, the realized surprisal, the \emph{threshold decoder} returns a key consistent
with $X$ when $\imath(X)>\tau$ (any one, if several are consistent) and declares failure when
$\imath(X)\le\tau$.
\end{definition}

The theorem states the rate, then gives in three parts the exact per-innocent guarantee, the completeness
threshold, and the matching converse that together make it two-sided.

\begin{restatable}[Distortion-free attribution is $\Theta(\log N/h)$]{theorem}{ThmEntropyClose}\label{thm:entropyclose}
Let $\pnat$ be stationary ergodic with entropy rate $h>0$ (no mixing hypothesis), marked and decoded as in
\cref{def:witness}. Fix a per-innocent false-positive level $\delta\in(0,\tfrac12)$ and let $N\to\infty$ (optionally
$\delta\to0$ with $\log(1/\delta)=o(\log N)$ for the $(1+o(1))$ form). Then attribution among $N$ users costs
$\ntask{\att}=\Theta(\log N/h)$ (achievability alone holding for all $\delta\in(0,1)$), sharpening to
$(1+o(1))\log N/h$, and:
\begin{enumerate}[label=(\roman*)]
  \item \textbf{(False positives; exact, non-asymptotic.)} For every $n$ and arbitrary $\pnat$ (no
  ergodicity used): for the deterministic sampler of \cref{def:witness}, a fixed innocent $v$ is the key
  returned only with probability $\Pr[\hat U=v]\le e^{-\tau}=\delta/N$, \emph{pointwise in the realized
  innocent keys} (the probability is over the true text $X\sim\pnat$); for randomized distortion-free
  kernels the same bound holds averaged over the key draw.
  \item \textbf{(Completeness.)} $\Pr[\imath(X)>\tau]\to1$ whenever $n\ge(1+\alpha)\tau/h$ for any fixed
  $\alpha>0$ (Shannon--McMillan--Breiman; here $h$ is the entropy \emph{rate}).
  \item \textbf{(The rate, two-sided.)} Achievability: the average error is $\le\delta+o(1)$ once
  $n\ge(1+\alpha)(\log N+\log(1/\delta))/h$, for every fixed $\alpha>0$ ($N$ large). Converse: for
  i.i.d.\ keys, every attribution rule with error $\le2\delta$ needs $n\ge(1-2\delta)\log N/h_n-O(1/h)$,
  with $h_n:=H(\pnat^{(n)})/n$ ($=(1+o(1))h$) --- conditional Fano discharged by a mixture-entropy step
  (\cref{app:entropyclose}). Thus
  \begin{align*}
    \ntask{\att}&\;=\;\Theta(\log N/h)\quad\text{for fixed }\delta\in(0,\tfrac12),\\
    \ntask{\att}&\;=\;(1+o(1))\,\frac{\log N}{h}\quad\text{as }\delta\to0,\ \log(1/\delta)=o(\log N).
  \end{align*}
  For a \emph{fixed} deployed key set the converse instead holds under a mixture-entropy hypothesis on the
  realized keys (as $N\to\infty$; automatic when $\pnat$ is the maximum-entropy law), detailed in
  \cref{app:entropyclose}.
\end{enumerate}
\end{restatable}

\begin{restatable}[Payload extraction]{corollary}{CorExtraction}\label{cor:extraction-h}
Index $N=2^\ell$ keys by the payload $W\sim\Unif\{0,1\}^\ell$. Under the hypotheses of
\cref{thm:entropyclose}, $\ntask{\ext}=\Theta(\ell/h)$ for fixed $\delta$, and
$\ntask{\ext}=(1+o(1))\,\ell\log2/h$ as $\delta\to0$ with $\log(1/\delta)=o(\ell)$.
\end{restatable}

\begin{proof}[Proof sketch]
(i) is a two-line computation, the deterministic case first (matching the proof order in the appendix): fix
the realized key set; on $\{\imath(X)>\tau\}$ the single string an innocent key determines has $\pnat$-mass
$<e^{-\tau}$, so $\Pr[\hat U=v]\le e^{-\tau}$ pointwise. For randomized distortion-free kernels, average
over the key draw: given $X$, an independent wrong key is consistent with probability
$\prod_t\pnat(X_t\mid X_{<t})=e^{-\imath(X)}$, so $\Pr[\hat U=v]\le\E[e^{-\imath}\,\mathbf 1\{\imath>\tau\}]
=\sum_{X:\,\pnat(X)<e^{-\tau}}\pnat(X)^2\le e^{-\tau}$.
(ii) is Shannon--McMillan--Breiman: $\imath(X)/n\to h$ a.s.\ and $\tau/n\le h/(1+\alpha)<h$. (iii) combines
the two with a union bound over the $N-1$ innocents; the converse is conditional Fano plus the Jensen step
$\E_\kappa[H(\hat p_\kappa)]\le H(\E_\kappa[\hat p_\kappa])=nh_n$, where $\kappa=\{k_1,\dots,k_N\}$ is the
realized key set and $\hat p_\kappa:=\tfrac1N\sum_v\pwm^{(k_v)}$ the induced key mixture. Full proof in
\cref{app:entropyclose}.
\end{proof}

\begin{remark}[Tightness and necessity of the hypotheses]\label{rem:necessity}
Both structural hypotheses are necessary, not technical. \emph{Ergodicity:} take the stationary
non-ergodic mixture that emits, with probability $\tfrac12$ each, the all-zero stream or i.i.d.\ fair
bits. Its process entropy rate is $\tfrac12\log2>0$, yet on the constant component the realized surprisal
never exceeds $\log2<\tau$, so every attribution rule (informed or not) errs with probability
$\ge\tfrac12-o(1)$ at every length: the rate $\log N/h$ fails when $h$ is read as the process rate. What
survives for stationary sources is the componentwise statement: part (i) holds verbatim, completeness
holds within each ergodic component, and as $\delta\to0$ the rate becomes $(1+o(1))\log N/h^{*}$ with
$h^{*}$ the essential infimum of the components' entropy rates (elementary for finite mixtures; for fixed
$\delta$ the constant improves to a $\delta$-quantile of the component rates, whose sharp form we leave
open). Ergodicity is exactly the condition collapsing $h^{*}$, the quantile, and the process rate to one
number. \emph{The key registry:} without the indexed assignment $v\mapsto k_v$ the users are exchangeable
given $X$, indeed $X$ is independent of $U$, so every rule succeeds with probability exactly $1/N$;
attribution without the registry is attribution without the identities. The converse, by contrast, already
bounds the informed decoder.
\end{remark}

\begin{figure}[t]
\centering
\begin{tikzpicture}[font=\sffamily,>=Stealth,line cap=round]
  \begin{scope}
    \def\W{6.2}\def\H{2.45}\def\hh{1.08}\def\xstar{2.0}
    \fill[bluetint]
      (0.5,1.78) -- (1.5,1.52) -- (2.5,1.34) -- (3.5,1.24) -- (4.5,1.17) -- (\W,1.14) --
      (\W,1.02) -- (4.5,0.99) -- (3.5,0.92) -- (2.5,0.80) -- (1.5,0.62) -- (0.5,0.38) -- cycle;
    \draw[->,inkgray!80,line width=0.8pt] (0,0)--(\W+0.15,0) node[below=1pt,font=\scriptsize,text=inkgray]{$n$};
    \draw[->,inkgray!80,line width=0.8pt] (0,0)--(0,\H);
    \draw[detectblue,line width=1.2pt] (0,\hh)--(\W,\hh) node[right,font=\scriptsize\bfseries,text=detectblue]{$h$};
    \draw[detectblue!85,line width=1pt]
      (0.5,1.92)--(0.9,1.55)--(1.3,1.74)--(1.8,1.42)--(2.3,1.26)--(2.9,1.33)--(3.5,1.19)--(4.2,1.13)--(5,1.10)--(\W,1.07);
    \node[font=\scriptsize,text=detectblue,anchor=west] at (2.5,1.74) {$\imath(X)/n\to h$};
    \draw[attribgold,line width=1.2pt,smooth,samples=60,domain=1.0:\W] plot (\x,{2.1/\x});
    \node[font=\scriptsize,text=attribgold!88!black,anchor=west] at (3.0,0.86) {$\tau/n=\log(N/\delta)/n$};
    \draw[densely dotted,inkgray,line width=0.8pt] (\xstar,0)--(\xstar,\hh);
    \fill[inkgray] (\xstar,\hh) circle (1.5pt);
    \node[font=\scriptsize\bfseries,text=inkgray,anchor=north] at (\xstar,-0.04) {$n^\star=\tau/h$};
    \draw[->,detectblue!70,line width=0.7pt] (\xstar+0.12,0.30) -- (\W-0.2,0.30);
    \node[font=\scriptsize,text=detectblue,anchor=west] at (\xstar+0.18,0.50) {$\imath(X)>\tau$};
    \node[font=\small\bfseries,text=detectblue] at (\W/2,-0.62) {(ii) completeness \;\textnormal{\small(rate $=h$)}};
  \end{scope}

  \draw[inkgray!22,line width=0.6pt] (6.95,-1.25)--(6.95,2.45);

  \begin{scope}[shift={(7.5,0)}]
    \node[anchor=north west,align=left,font=\scriptsize,text=inkgray] at (0,2.4)
      {An innocent key $k_v$ generates \emph{one} fixed string $X(k_v)$.\\[2pt]
       It frames $v$ only if $X(k_v)=X$ \emph{and} $\imath(X)>\tau$:\\[5pt]
       \hspace{0.6em}$\displaystyle \Pr[\hat U{=}v]\;=\;\pnat\big(X(k_v)\big)\;=\;e^{-\imath}\;<\;e^{-\tau}\;=\;\frac{\delta}{N}.$};
    \def\mx{0.2}\def\mw{4.6}
    \draw[densely dashed,attribgold!85!black,line width=0.9pt] (\mx,0.45)--(\mx+\mw,0.45)
       node[right,font=\scriptsize,text=attribgold!85!black,inner sep=2pt] {$\delta/N=e^{-\tau}$};
    \fill[detectblue!75] (\mx+0.1,0) rectangle (\mx+0.32,0.30);
    \node[anchor=west,font=\scriptsize,text=detectblue] at (\mx+0.42,0.16) {innocent's mass $e^{-\imath}$ (a sliver)};
    \node[anchor=north west,font=\scriptsize\itshape,text=inkgray] at (0,-0.18)
      {pointwise, exact, every $n$, every source};
    \node[font=\small\bfseries,text=attribgold!85!black] at (2.45,-1.12) {(i) soundness \;\textnormal{\small(FP $=\delta/N$)}};
  \end{scope}
\end{tikzpicture}
\caption{\textbf{The surprisal threshold does both jobs} (\cref{thm:entropyclose}). The decoder accepts a
consistent key only when the realized surprisal $\imath(X)=-\log\pnat(X)$ exceeds $\tau=\log(N/\delta)$, and
that single threshold is \emph{forced}, not tuned. \textbf{(ii)} The marked text is itself $\pnat$-typical, so
$\imath(X)/n\to h$ (Shannon--McMillan--Breiman) and clears $\tau$ once $n\ge\tau/h$: this is why the rate is
the entropy $h$. \textbf{(i)} An innocent key produces \emph{one} fixed string, which coincides with $X$ only
with $\pnat$-mass $e^{-\imath}<e^{-\tau}=\delta/N$: the per-innocent guarantee is exact and pointwise. To frame
no innocent above $\delta/N$ one \emph{must} reject every candidate below $\tau$, and completeness then only
asks that the true text clear the same bar.}
\label{fig:surprisal}
\end{figure}

\begin{proposition}[The collision trap]\label{rem:listsize}
Consider the witness scheme and threshold decoder of \cref{def:witness} on an i.i.d.\ source, and say a key
is \emph{consistent} with $X$ if it generates $X$ token for token. Then:
\begin{enumerate}[label=(\roman*)]
  \item the expected number of falsely consistent keys is exactly $(N-1)e^{-nr_2}$, where
  $r_2:=-\log\sum_x\pnat(x)^2$ is the R\'enyi-2 rate, so the expected list clears (falls below one) only at
  $n\ge\log(N-1)/r_2$;
  \item yet the surprisal-threshold decoder keeps per-innocent false positives $\le\delta/N$ at every
  length, and its average error is $\le\delta+o(1)$ once $n\ge(1+\alpha)\log(N/\delta)/h$, for any fixed
  $\alpha>0$;
  \item $r_2\le h$ always, with equality iff $\pnat$ is uniform on its support, so the consistency count
  prescribes $h/r_2$ times more tokens than the decision requires, a factor unbounded as the source skews.
\end{enumerate}
\end{proposition}

\begin{proof}
(i) is the collision identity of \cref{lem:overlap} specialized to the deterministic witness and an
i.i.d.\ source (\cref{app:overlap}); (ii) is \cref{thm:entropyclose}(i)--(iii); (iii) is monotonicity of
R\'enyi entropy in its order, with equality exactly on support-uniform sources (\cref{sec:itbackground});
for $\pnat=\mathrm{Ber}(p)$, $h/r_2\to\infty$ as $p\to0$.
\end{proof}

The trap, in words. Consistency (alignment) is the rule deployed detectors actually use: answer with any
key consistent with the text. Its natural analysis drives the expected list $(N-1)e^{-nr_2}$ below a
constant and so predicts the rate $\log N/r_2$ --- the \emph{wrong} rate. Collision entropy prices the
candidate \emph{list}; Shannon entropy prices the \emph{decision}: in the window
$\log N/h\le n\ll\log N/r_2$ the list is exponentially large, yet the surprisal threshold has already
separated the true key \emph{within} it (\cref{fig:collision}). Only the second rate is a sample
complexity. The seeded numerics certify the trap: at the operating point the consistency-only decoder's
false-positive rate overshoots the $\delta/N$ target by orders of magnitude (\cref{sec:numerics}). A
tie-abstaining variant of the consistency rule has zero false positives in the closed-world model
$U\in[N]$ but loses the guarantee on out-of-set text (\cref{app:overlap}); it is not used by the headline.
Whether a genuine $r_2$ cost reappears for edit-robust decoders is taken up in \cref{sec:discussion}.

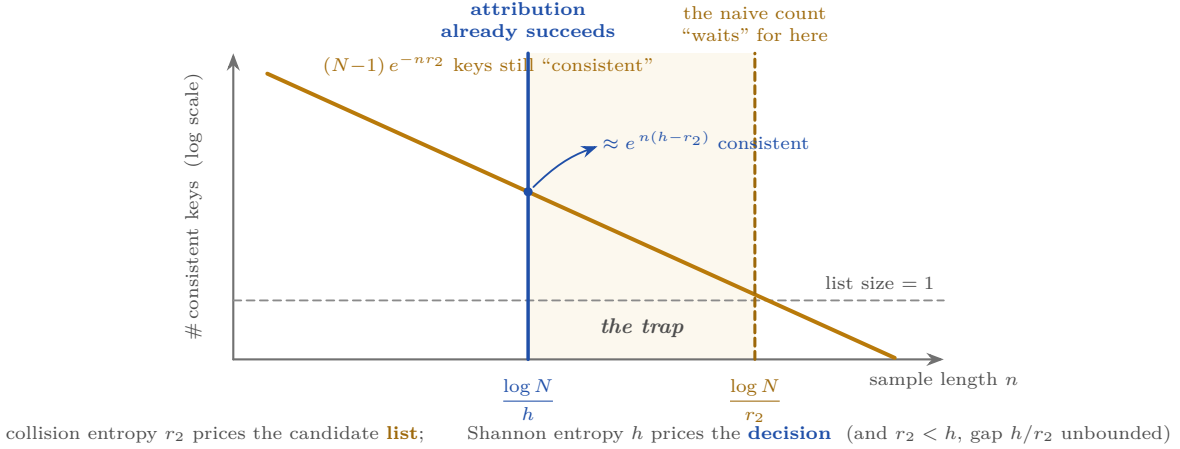
\begin{figure}[t]
\centering
\begin{tikzpicture}[font=\sffamily,>=Stealth,line cap=round]
  \def\xd{3.9}\def\xl{6.9}\def\yone{0.78}\def\W{9.4}\def\H{4.05}
  \fill[goldtint] (\xd,0) rectangle (\xl,\H);
  \draw[->,inkgray!80,line width=0.8pt] (0,0)--(\W,0)
     node[below=1pt,font=\scriptsize,text=inkgray]{sample length $n$};
  \draw[->,inkgray!80,line width=0.8pt] (0,0)--(0,\H);
  \node[rotate=90,anchor=south,font=\scriptsize,text=inkgray] at (-0.28,\H/2)
     {\#\,consistent keys \;(log scale)};
  \draw[densely dashed,inkgray!65,line width=0.7pt] (0,\yone)--(\W,\yone);
  \node[anchor=south east,font=\scriptsize,text=inkgray] at (\W,\yone+0.03) {list size $=1$};
  \draw[attribgold,line width=1.5pt] (0.45,3.78) -- (8.75,0.02);
  \node[anchor=west,font=\scriptsize,text=attribgold!88!black] at (1.06,3.9)
     {$(N{-}1)\,e^{-n r_2}$ keys still ``consistent''};
  \draw[detectblue,line width=1.3pt] (\xd,0)--(\xd,\H);
  \fill[detectblue] (\xd,2.217) circle (1.7pt);   
  \node[anchor=south,align=center,font=\scriptsize\bfseries,text=detectblue] at (\xd,\H+0.04)
     {attribution\\already succeeds};
  \node[anchor=north,font=\scriptsize\bfseries,text=detectblue] at (\xd,-0.14) {$\dfrac{\log N}{h}$};
  \draw[densely dashed,attribgold!85!black,line width=1.15pt] (\xl,0)--(\xl,\H);
  \node[anchor=north,font=\scriptsize\bfseries,text=attribgold!85!black] at (\xl,-0.14) {$\dfrac{\log N}{r_2}$};
  \node[anchor=south,align=center,font=\scriptsize,text=attribgold!85!black] at (\xl,\H+0.04)
     {the naive count\\``waits'' for here};
  \draw[->,detectblue,line width=0.7pt] (\xd+0.07,2.28) to[bend left=12] (\xd+0.92,2.82);
  \node[anchor=west,font=\scriptsize,text=detectblue] at (\xd+0.86,2.9) {$\approx e^{\,n(h-r_2)}$ consistent};
  \node[font=\footnotesize\itshape\bfseries,text=inkgray] at ({(\xd+\xl)/2},0.42) {the trap};
  \node[anchor=north,align=center,font=\scriptsize,text=inkgray] at (\W/2,-0.74)
     {collision entropy $r_2$ prices the candidate \textbf{\textcolor{attribgold!85!black}{list}};\qquad
      Shannon entropy $h$ prices the \textbf{\textcolor{detectblue}{decision}} \ (and $r_2<h$, gap $h/r_2$ unbounded)};
\end{tikzpicture}
\caption{\textbf{The collision trap} (\cref{rem:listsize}). The obvious decoder answers with any key
\emph{consistent} with the text; its list of falsely consistent innocents, $(N-1)e^{-nr_2}$, clears only
at $n=\log N/r_2$. So a count-based analysis predicts that rate. But the surprisal threshold decides
already at $n=\log N/h$, deep inside the window where the list is still exponentially large
($\approx e^{\,n(h-r_2)}$): it singles out the true key \emph{within} the list rather than waiting for the
list to shrink. The R\'enyi-2 rate $r_2$ prices the candidate list; the Shannon rate $h$ prices the
decision, and $r_2<h$ makes the gap $h/r_2$ unbounded.}
\label{fig:collision}
\end{figure}

\begin{example}[A skewed coin makes the gap concrete]\label{ex:coin}
Take $\pnat=\mathrm{Ber}(0.1)$ i.i.d.\ (the source of \cref{sec:numerics}). Its entropy rate is
$h=\Hb(0.1)=0.325$ nats, while its R\'enyi-2 rate is $r_2=-\log(0.1^2+0.9^2)=0.198$ nats, so
$h/r_2\approx1.64$. To attribute among $N=2^{32}$ users, the surprisal threshold succeeds at
$n\approx\log N/h\approx68$ tokens (the $\log(1/\delta)$ term, about $9$ tokens here, is lower-order in $N$), while the consistency list
clears only at $n\approx\log N/r_2\approx112$; their ratio is exactly $h/r_2\approx1.64$. The gap grows
without bound as the coin skews ($h/r_2\to\infty$ as $p\to0$): the trap is not a constant-factor nuisance.
\end{example}

Credit, again. The false-positive mechanism in (i) is the folklore likelihood-ratio Markov bound, for which
we claim no novelty (\cref{app:entropyclose}). What is new is the closed column: the matching
achievability at full stationary-ergodic generality (SMB only), with an exact non-asymptotic false-positive
side, and its forensic reading --- the entropy mass of \cref{thm:budget}(b) is operational, not merely a
ceiling. The distortion-free sub-column of \cref{tab:ladder}, a one-sided converse ledger in the first
version of this paper, is now two-sided at Levels 1 and 2. The closest cryptographic tracing of many
adaptive users gives no tight information-theoretic $h$-rate \citep{cohen2024many}; Christ--Gunn-style
undetectable schemes are computational and governed by seed entropy (scope note
after \cref{thm:budget} in \cref{sec:appendix}). The witness decoder is exact-alignment, hence not
edit-robust; a $\Delta$-free, edit-robust attribution rate remains open (\cref{sec:discussion}).

\begin{remark}[Where the columns meet]\label{rem:columnsmeet}
The two extremal columns of \cref{tab:ladder} touch at exactly one source. On a full-alphabet-uniform
i.i.d.\ source, $h=\log q$, so the closed entropy column gives $(1+o(1))\log N/\log q$: the invisible
distortion-free witness then costs no more than the cleartext digit-stamp of \cref{thm:dichotomy}(b), and
only there, since $\log N/h\ge\log N/\log q$ always, with equality iff $\pnat=\Unif(\Sigma)$; the witness
moreover saturates \cref{thm:budget}(b) with equality ($I(K;X)=H(X)=nh_n$ for the deterministic sampler),
which is why the column closes. Separately, the collision trap vanishes exactly on \emph{support}-uniform
sources ($r_2=h$ iff $\pnat$ is uniform on its support, \cref{rem:listsize}(iii)): the overcharge $h/r_2$
is a price of skew. The two degeneracies are distinct: a source uniform on a strict sub-alphabet kills the
trap but not the gap between the columns.
\end{remark}

\begin{table}[t]
\centering
\caption{The forensic ladder. \emph{Symbols:} $\Delta$ per-token KL budget; $h$ per-token entropy rate;
$\lambda$ detection soundness ($\mathrm{FP}\le2^{-\lambda}$); $q=|\Sigma|$ alphabet size; $m$ footprint;
$w$ crop window; $N$ users; $\ell$ payload bits. In the
distributional column the lower bounds ($\Omega$) are unconditional; the matching upper bounds ($\Theta$,
shown in brackets) hold for regular, strongly-mixing $\pnat$ (\cref{thm:dichotomy}) and are conjectural for
general $q$-ary $\pnat$. The embedding column is two-sided ($\Theta$, exact construction); the distortion-free
($h$) sub-column is now two-sided at Levels 1--2 (\cref{thm:entropyclose}: stationary-ergodic source,
i.i.d.\ keys, decoder knowing $\pnat$ and the key set). The bounded Level-0 rate $\Theta(1/\Delta)$ is the
imported baseline \citep{cai2024better,li2024statistical,huang2023optimal}; the distortion-free Level-0 rate
$\Theta(\lambda/h)$ (with $\lambda$ the entropy needed for $2^{-\lambda}$ false-positive security) is from
\citet{christ2024undetectable}; the bounded-regime attribution converse is the single-user Fano specialization
of \citet{moulin2008universal,somekhbaruch2005capacity}. Level 3 (crop-robust
localization to resolution $w$, \cref{thm:crop}) turns on a different axis, readout resolution: it forces full
support, so the small-footprint embedding column cannot reach it.}
\label{tab:ladder}
\small
\begin{tabular}{@{}llll@{}}
\toprule
Level & Task & Distributional $\mathcal B(\Delta)$ \ ($h$: distortion-free) & Footprint $\mathcal E(m)$ (embedding) \\
\midrule
0 & Detection & $\Theta(1/\Delta)$ \ / \ $\Theta(\lambda/h)$ & same (baseline) \\
1 & Attribution ($N$) & $\Omega(\log N/\Delta)\,[\Theta\text{ if mixing}]$ \ / \ $\Theta(\log N/h)$ & $\Theta(\log N/\log q)$, \emph{$\Delta$-free} \\
2 & Extraction ($\ell$) & $\Omega(\ell/\Delta)\,[\Theta\text{ if mixing}]$ \ / \ $\Theta(\ell/h)$ & $\Theta(\ell/\log q)$, \emph{$\Delta$-free} \\
3 & Localization (res.\ $w$) & needs full support, $w\Delta\ge\log(n/w)$ & coarse $w=\omega(1)$ only \\
\bottomrule
\end{tabular}
\end{table}

\subsection{Level 3, and a shared budget}

The top rung, localization, is where the framework does work the budget alone does not predict. A watermark is
\emph{localizing} if it can name which sub-region carries the mark, and the demanding version asks this to
survive cropping. Model a crop adversary that deletes all but a contiguous window $W$ of length $w$; a scheme
is \emph{$w$-crop-robust} if the detector succeeds on $x_W$ for every offset. We assume the scheme is
\emph{$r$-local / edit-based} (\cref{asm:local}): the marked carrier agrees with an unmarked one outside an
$r$-neighbourhood $R^{+}$ of the \emph{carrier footprint} $R$ (the support of the edit coupling; for a scheme
in $\mathcal E(m)$ over a factorizing source one may take $R=T$), the post-hoc regime of image, audio, and
patch watermarks, where
Level-3 localization is actually studied. Robustness is the one property that is \emph{not} a functional of
$\nu$ alone: mass and shape live in $\nu$; crop-robustness additionally needs the edit coupling on $R$, which
is why \cref{asm:local} is stated on the carrier.

\begin{restatable}[Footprint--resolution uncertainty]{theorem}{ThmCrop}\label{thm:crop}
Under \cref{asm:local}, any $w$-crop-robust scheme has dilated support $R^{+}$ meeting every length-$w$
window; since $|R^{+}|\le(2r+1)|R|$, its carrier footprint $|R|$ obeys
\[
  |R|\,w\ \ge\ \frac{n-w+1}{2r+1},\qquad\text{i.e.}\qquad |R|\,w=\Omega(n)\ \ (r=O(1)),
\]
the $\Omega(n)$ reading trivial for $w>n/2$ (a nonempty $R$ already gives $|R|\,w>n/2$). For a scheme in
$\mathcal E(m)$ whose natural law factorizes across positions (in particular i.i.d.; the post-hoc regime of
\cref{asm:local}) one may take $R=T$, so $|R|\le m$ and the bound reads $m\,w=\Omega(n)$.
Consequently (i) a small-footprint embedding scheme ($m=o(n)$, hence $|R|=o(n)$) is crop-robust only at
coarse resolution $w=\omega(1)$, and (ii) constant-resolution crop-robustness forces $|R|=\Omega(n)$,
i.e.\ full support, that is, biasing. Moreover (iii) under \cref{asm:bounded}, for a natural law that
factorizes across the length-$w$ blocks (in particular i.i.d., and the whitened regime of the post-hoc
schemes \cref{asm:local} targets), $w$-crop-robust \emph{localization} (naming the window among
the $\Theta(n/w)$ blocks from $x_W$ alone) requires $w\Delta\ge(1-o(1))\log(n/w)$.
\end{restatable}

\begin{proof}[Proof sketch]
If a length-$w$ window $W$ misses $R^{+}$, then by the coupling of \cref{asm:local} $x_W=x'_W\sim\pnat$, so a
detector seeing only $x_W$ tests $\pnat$ against $\pnat$ and cannot beat chance, a contradiction. Hence $R^{+}$
meets every length-$w$ window; there are $n-w+1$ such windows and one element pierces at most $w$ of them,
so $|R^{+}|\ge(n-w+1)/w$; with
$|R^{+}|\le(2r+1)|R|$ this gives $|R|\,w=\Omega(n)$, and $|R|\le m$ for $\mathcal E(m)$. For (iii),
partition $[n]$ into $N_w=\lfloor n/w\rfloor$
disjoint length-$w$ blocks; naming the surviving block is attribution among the blocks (\cref{thm:separation}
with the $N_w$ blocks as users), so by Fano $I(J;x_{W_J})\ge(1-o(1))\log N_w$ for
$J\sim\Unif[N_w]$. When $\pnat$ factorizes across the blocks, the golden formula bounds $I(J;x_{W_J})\le
\E_J\sum_{i\in W_J}d_i=\tfrac1{N_w}\sum_{i}d_i\le n\Delta/N_w=w\Delta$, using only the \emph{global}
budget $\sum_i d_i\le n\Delta$ (no per-token cap); hence $w\Delta\ge(1-o(1))\log(n/w)$. For correlated
$\pnat$ the cropped block marginalizes its prefix and this accounting can fail; the relation persists in our
numerical search, but a tight general proof is left open.
\end{proof}

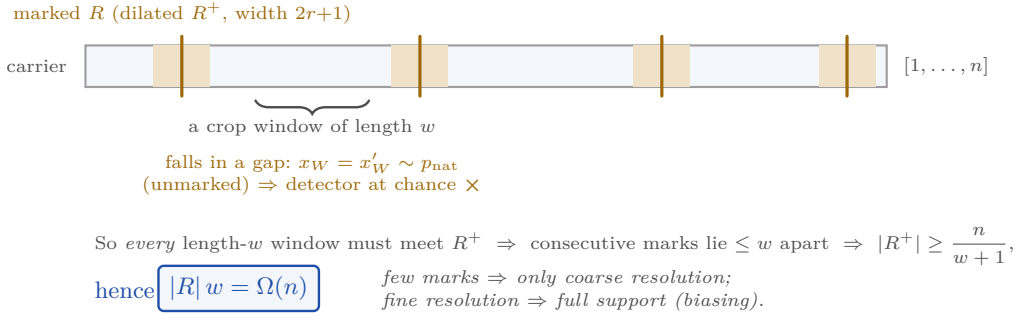
\begin{figure}[t]
\centering
\begin{tikzpicture}[font=\sffamily,>=Stealth,line cap=round]
  \def\L{10.6}\def\sy{2.75}\def\sh{0.55}
  \fill[bluetint] (0,\sy) rectangle (\L,\sy+\sh);
  \draw[inkgray!55,line width=0.8pt] (0,\sy) rectangle (\L,\sy+\sh);
  \node[anchor=east,font=\scriptsize,text=inkgray] at (-0.12,\sy+\sh/2) {carrier};
  \node[anchor=west,font=\scriptsize,text=inkgray] at (\L+0.1,\sy+\sh/2) {$[1,\dots,n]$};
  \foreach \a/\b in {0.9/1.65, 4.05/4.8, 7.25/8.0, 9.7/10.45}{
    \fill[attribgold!22] (\a,\sy) rectangle (\b,\sy+\sh);
    \draw[attribgold!85!black,line width=1.2pt] ({(\a+\b)/2},\sy-0.12) -- ({(\a+\b)/2},\sy+\sh+0.12);
  }
  \node[anchor=south,font=\scriptsize,text=attribgold!85!black] at (1.28,\sy+\sh+0.16)
     {marked $R$ (dilated $R^{+}$, width $2r{+}1$)};
  \def\wa{2.25}\def\wb{3.75}
  \draw[decorate,decoration={brace,amplitude=5pt,mirror},inkgray,line width=0.9pt] (\wa,\sy-0.18)--(\wb,\sy-0.18);
  \node[anchor=north,font=\scriptsize,text=inkgray] at ({(\wa+\wb)/2},\sy-0.30) {a crop window of length $w$};
  \node[anchor=north,align=center,font=\scriptsize,text=attribgold!85!black] at ({(\wa+\wb)/2},\sy-0.74)
     {falls in a gap: $x_W=x'_W\sim\pnat$\\(unmarked) $\Rightarrow$ detector at chance $\bm\times$};
  \node[anchor=north west,align=left,font=\scriptsize,text=inkgray] at (0,1.02)
     {So \emph{every} length-$w$ window must meet $R^{+}$ $\;\Rightarrow\;$ consecutive marks lie $\le w$ apart
      $\;\Rightarrow\;$ $|R^{+}|\ge \dfrac{n}{w+1}$,};
  \node[anchor=north west,font=\small,text=detectblue] at (0.0,0.30) {hence};
  \node[anchor=north west,rounded corners=2pt,draw=detectblue,line width=0.9pt,fill=bluetint,
        inner sep=4pt,font=\small\bfseries,text=detectblue] at (0.95,0.40) {$|R|\,w=\Omega(n)$};
  \node[anchor=west,align=left,font=\scriptsize\itshape,text=inkgray] at (3.8,0.06)
     {few marks $\Rightarrow$ only coarse resolution;\\ fine resolution $\Rightarrow$ full support (biasing).};
\end{tikzpicture}
\caption{\textbf{Why a small footprint cannot localize at fine resolution} (\cref{thm:crop}). Under the edit
coupling (\cref{asm:local}), a crop window that misses the dilated footprint $R^{+}$ sees text distributed
exactly as $\pnat$, so a detector reading only that window tests $\pnat$ against $\pnat$ and cannot beat
chance. Every length-$w$ window must therefore meet $R^{+}$, which forces the marks no more than $w$ apart:
$|R|\,w=\Omega(n)$. Few marks buy only coarse resolution; fine resolution forces full support, that is,
biasing.}
\label{fig:crop}
\end{figure}

So Level 3 turns on a different axis than the footprint (\cref{fig:crop}): a small-footprint embedding scheme
cannot be crop-robust at fine resolution at all (a single splice removes its mark), whereas crop-robust fine
localization \emph{forces} full support $|R|=\Omega(n/w)$. A scheme localizes when it reads a full-support mark
at fine per-region resolution and self-locates at window budget $w\ge\log(n/w)/\Delta$. This is the structural
form of the empirical split (fine-readout schemes localize; global-readout schemes do not), made precise in
\cref{sec:taxonomy}.

Finally, the levels are not independent obligations but draws on one account.

\begin{restatable}[Forensic rate region]{proposition}{PropRegion}\label{prop:region}
Under \cref{asm:bounded} and in the regime $n\Delta\to\infty$, with $R_{\att}=\log N$ and $R_{\ext}=\ell\log2$:
\begin{enumerate}[label=(\alph*)]
  \item \emph{(Converse.)} vanishing-error attribution \emph{and} extraction force $R_{\att}+R_{\ext}\le
  (1+o(1))\,n\Delta$;
  \item \emph{(Achievability.)} for regular, strongly-mixing $\pnat$ (general $q$-ary conjectural, inheriting
  \cref{thm:dichotomy}(a)), every pair with $R_{\att}+R_{\ext}\le(1-o(1))\,nC$ is achieved by one multiplexed
  conditional tilting code, where $C=\Theta(\Delta)$ is the per-token channel capacity.
\end{enumerate}
Thus the achievable region is a slope-$-1$ simplex whose boundary is \emph{order-tight}, $\Theta(n\Delta)$: it
lies between $nC$ and $n\Delta$, which agree up to the capacity constant ($C\approx\Delta/4$ in the binary
witness of \cref{sec:numerics}). The corners carry no extra loss.
\end{restatable}

\begin{proof}[Proof sketch]
With $S=(U,W)$, the chain rule gives $I(U,W;X)=I(U;X)+I(W;X\mid U)\le n\Delta$ (\cref{thm:budget}); Fano on
$U$ and conditional Fano on $W\mid U$ make the two terms $\ge(1-o(1))R_{\att}$ and $\ge(1-o(1))R_{\ext}$,
giving (a). For (b), concatenate the two messages into one codeword and apply \cref{thm:dichotomy}(a) at rate
$nC$.
\end{proof}

So it is not just that attribution and extraction are each $\Omega(\cdot/\Delta)$, but that a biasing watermark
cannot be cheap at both at once: every nat spent naming the user is a nat denied to the payload
(\cref{fig:rateregion}). The embedding
analogue is the $\Delta$-free simplex $R_{\att}+R_{\ext}\le m\log q$, now a theorem
(\cref{thm:budget}(c), the $L^0$ face of the mass cap), and there, uniquely, the bound is constant-exact
(the zero-error stamp readback meets it).

\begin{figure}[t]
\centering
\scalebox{1.3}{
\begin{tikzpicture}[font=\sffamily,>=Stealth,line cap=round]
  \begin{scope}
    \def\A{3.45}
    \fill[bluetint] (0,0)--(2.3,0)--(0,2.3)--cycle;
    \fill[detectblue!12] (2.3,0)--(3.1,0)--(0,3.1)--(0,2.3)--cycle;
    \draw[detectblue,line width=1.1pt] (2.3,0)--(0,2.3);
    \draw[detectblue,densely dashed,line width=1.0pt] (3.1,0)--(0,3.1);
    \draw[->,inkgray!80,line width=0.8pt] (0,0)--(\A,0) node[right,font=\scriptsize,text=inkgray]{$R_{\att}$};
    \draw[->,inkgray!80,line width=0.8pt] (0,0)--(0,\A) node[above,font=\scriptsize,text=inkgray]{$R_{\ext}$};
    \node[font=\scriptsize,text=detectblue,anchor=west,rotate=-45] at (0.95,1.55) {$nC$};
    \node[font=\scriptsize,text=detectblue,anchor=west,rotate=-45] at (1.20,2.26) {$n\Delta$};
    \node[align=center,font=\small\bfseries,text=detectblue] at (1.7,-0.65)
       {(a) biasing $\mathcal B(\Delta)$\\[-1pt]{\scriptsize\textnormal{soft frontier, order-tight $\Theta(n\Delta)$}}};
  \end{scope}
  \begin{scope}[shift={(5.3,0)}]
    \def\A{3.45}
    \fill[goldtint] (0,0)--(3.0,0)--(0,3.0)--cycle;
    \draw[attribgold!88!black,line width=1.4pt] (3.0,0)--(0,3.0);
    \draw[->,inkgray!80,line width=0.8pt] (0,0)--(\A,0) node[right,font=\scriptsize,text=inkgray]{$R_{\att}$};
    \draw[->,inkgray!80,line width=0.8pt] (0,0)--(0,\A) node[above,font=\scriptsize,text=inkgray]{$R_{\ext}$};
    \node[font=\scriptsize,text=attribgold!88!black,anchor=west,rotate=-45] at (0.66,2.05) {$=m\log q$};
    \fill[attribgold!88!black] (3.0,0) circle (1.7pt);
    \node[align=center,font=\scriptsize,text=attribgold!88!black] at (1.15,0.52)
       {\textbf{the stamp}\\zero error};
    \draw[->,attribgold!85!black,line width=0.6pt] (1.78,0.42) to[bend right=10] (2.96,0.06);
    \node[align=center,font=\small\bfseries,text=attribgold!88!black] at (1.7,-0.65)
       {(b) embedding $\mathcal E(m)$\\[-1pt]{\scriptsize\textnormal{sharp frontier, constant-exact}}};
  \end{scope}
  \node[align=center,font=\scriptsize\itshape,text=inkgray] at (4.15,-1.45)
     {one mass split two ways: every nat naming the user is a nat denied the payload};
\end{tikzpicture}}
\caption{\textbf{Two rate regions, one shape} (\cref{prop:region}). Attribution and extraction draw on the
single budget of \cref{thm:budget}, so their rates trade off along a slope-$(-1)$ frontier. \textbf{(a)} Under
the per-token cap $\mathcal B(\Delta)$ the frontier is soft, order-tight $\Theta(n\Delta)$, lying between the
achievable $nC$ and the converse $n\Delta$ ($C\approx\Delta/4$ in the binary witness of \cref{sec:numerics}).
\textbf{(b)} Under the footprint cap $\mathcal E(m)$ it is sharp and constant-exact, $R_{\att}+R_{\ext}=m\log
q$, met with equality and zero error by the stamp. The same simplex, with a soft boundary on one side and a
hard one on the other.}
\label{fig:rateregion}
\end{figure}

\subsection{Numerical illustration}
\label{sec:numerics}

\begin{figure}[t]
\centering
\includegraphics[width=\textwidth]{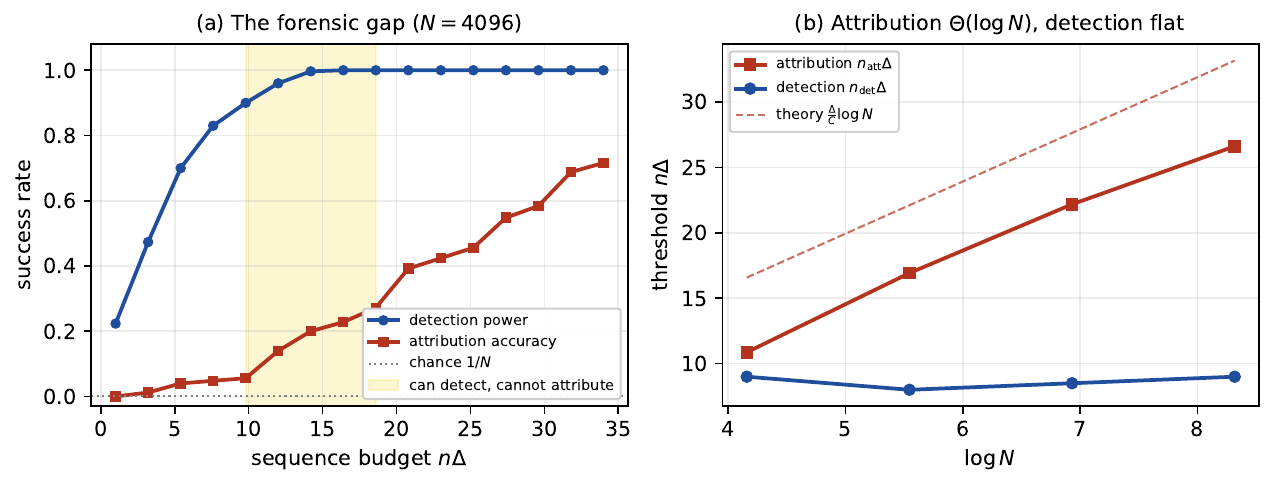}
\caption{Seeded Monte-Carlo confirmation of the separation (binary tilting code, $s=0.08$, so the per-token
budget is $\Delta=\KL{\pwm}{\pnat}/n$ and the channel capacity $C=I(b;X)\approx\Delta/4$). \textbf{(a)} At $N=4096$,
detection power saturates to $1$ while attribution accuracy is still near chance over a wide window (shaded):
one can detect long before one can attribute. \textbf{(b)} The attribution threshold grows linearly in
$\log N$ (slope matching $\Delta/C$, dashed), while the detection threshold is flat in $N$; their ratio is
$\Theta(\log N)$.}
\label{fig:numerics}
\end{figure}

We instantiate the tilting code over a binary alphabet ($p_+=\mathrm{Ber}(\tfrac12+s)$, $\pnat=
\mathrm{Ber}(\tfrac12)$) and measure detection power and attribution accuracy by seeded Monte Carlo
(\cref{fig:numerics}). Three findings confirm the theory. (i) The
converse \eqref{eq:fano} is never violated. (ii) The attribution threshold scales as $\log N/C$, where
$C=I(b;X)=\Theta(\Delta)$ is the channel capacity (empirically $C\approx\Delta/4$ at $s=0.08$, so the operative
constant is the capacity, not $\Delta$ itself), and grows by exactly the factor $\log4096/\log64=2$ as $N$ goes
from $64$ to $4096$, while the detection threshold is flat in $N$, so the ratio is $\Theta(\log N)$
(\cref{fig:numerics}b). (iii) At $N=4096$, detection power is already $0.99$ while attribution accuracy is
still $0.28$ at a common length ($n=1400$; \cref{fig:numerics}a), and
detection saturates to $1.00$ several hundred tokens before attribution leaves chance (it reaches $0.55$ only
at $n=2200$): the finite-$N$ form of the gap. The embedding stamp, by contrast, attributes among all
$4096$ users from $n=12$ tokens, where the biasing scheme is at chance. The first two findings quantify
\cref{thm:separation}; the third quantifies \cref{thm:dichotomy}.

Two further seeded experiments exercise \cref{thm:entropyclose} directly. The first runs the threshold
decoder on an i.i.d.\ binary source
($h\approx0.325$, $r_2\approx0.198$) at $\delta=0.05$ and $N$ up to $2^{32}$: the per-innocent false-positive
rate satisfies $\mathrm{FP}\cdot N/\delta\le0.11$ at every tested length, \emph{including} the boundary
$n=\tau/h$, while the consistency-only decoder (no surprisal threshold) overshoots the $\delta/N$ target by
one to four orders of magnitude --- the collision trap of \cref{rem:listsize}, priced by
$e^{-nr_2}$, not by the decision rate. A companion run tracks
the consistency decoder's threshold drift (seeded) and supports the abstain-on-tie completeness observation
of \cref{rem:listsize}. The second checks the Jensen step,
$\E_\kappa[H(\hat p_\kappa)]\le nh_n$, at all tested $(n,N)$ pairs, including $N>e^{nh}$ (all pass). All
runs are seeded and reproducible.

\subsection{Confirmation on deployed language models}
\label{sec:llm}

The theory makes a falsifiable numerical prediction: the consistency rule's overcharge is the entropy ratio
$h/r_2$, a quantity fixed by the source and by nothing about the watermark. On natural language it should be
a stable constant near $1.6$. The binary source of \cref{sec:numerics} isolates the mechanism on a
distribution we control; we now close the loop, measuring $h/r_2$ and the detection/attribution thresholds
where they are meant to bite: on real autoregressive language models, and on the watermarks practitioners
deploy. These are illustrations
of the theory's operational content, not new theory --- the lower bounds remain the theorems, unfalsifiable
by simulation, and what follows measures their \emph{inputs} (the rates $h,r_2$ and the
detection/attribution thresholds) on real systems; the exact false-positive guarantee $\delta/N$ itself was
exercised in the seeded numerics of \cref{sec:numerics}. All experiments are seeded and reproducible.

\paragraph{The collision trap on a shipped distortion-free scheme.}
\Cref{thm:entropyclose} and its trap (\cref{rem:listsize}) live in the distortion-free regime, whose
deployable instance is the exponential-minimum-sampling (Gumbel-max) watermark of
\citet{aaronson2022watermark} and \citet{kuditipudi2023robust}: a key-seeded Gumbel perturbs the
logits, $x_t=\arg\max_i\bigl(\log\pnat(i\mid x_{<t})+g^{(\kappa)}_{t,i}\bigr)$, so marginally
$x_t\sim\pnat$. We run it on GPT-2 \citep{radford2019gpt2}. Two properties make the entropy column
apply, and we verify rather than assume both. \textsf{(i)}~\emph{Distortion-free}: the marked text has
the same Shannon and R\'enyi-2 rates as unmarked sampling ($h=3.81$ vs.\ $3.76$, $r_2=2.36$ vs.\ $2.33$
nats/token), so the mark spends no distortion. \textsf{(ii)}~\emph{The collision is real}: at a fixed
context, $8000$ independent keys reproduce $\pnat$ (its $20$ most likely tokens to within $0.006$), so
an innocent key emits token $v$ with probability exactly $\pnat(v)$ --- so the per-token collision law
of \cref{rem:listsize} is a property of the construction at a fixed context, not a modelling assumption
($r_2$ here is the empirical R\'enyi-2 \emph{rate}, and the $(N-1)e^{-nr_2}$ count is its per-token
extrapolation). The guarantee
lengths then separate as predicted (\cref{fig:exp-ems}): implicating no innocent above $\delta/N$ costs
$\log(N/\delta)/r_2$ tokens under the consistency rule but only $\log(N/\delta)/h$ under surprisal
thresholding, a ratio $h/r_2=1.61$ on GPT-2. This is the trap on a scheme one can ship, and the
deterministic witness of \cref{sec:numerics} recovers the very same constant.

\begin{figure}[t]
\centering
\includegraphics[width=\textwidth]{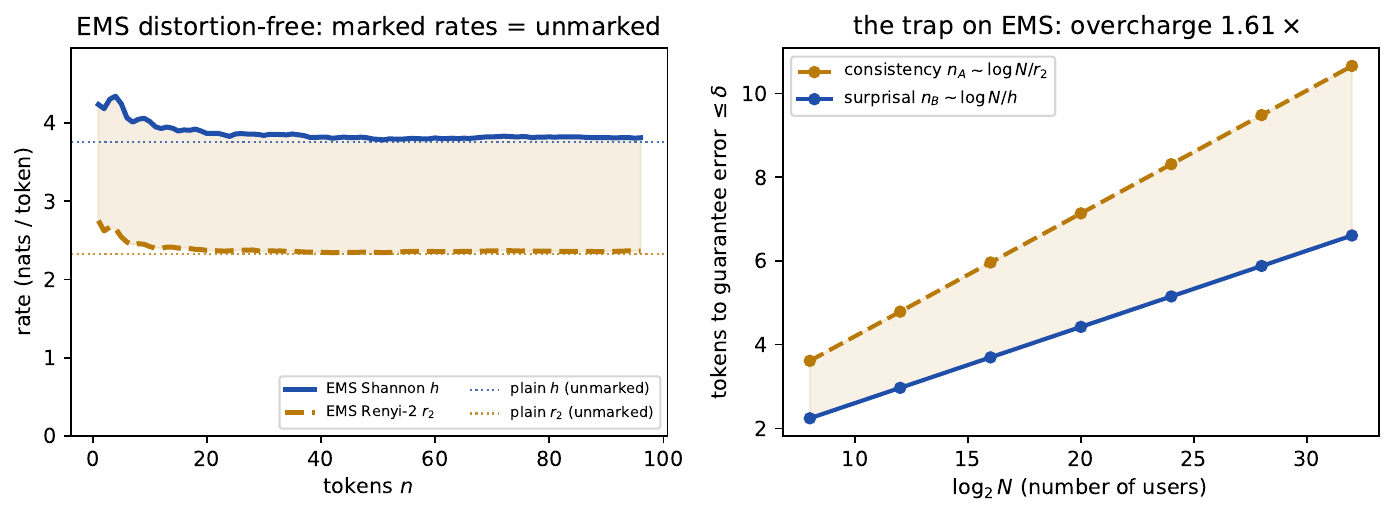}
\caption{\textbf{The collision trap on the deployed EMS (Gumbel-max) distortion-free watermark, GPT-2.}
\textbf{(a)}~The marked text's Shannon ($h$) and R\'enyi-2 ($r_2$) rates coincide with unmarked sampling
(dotted): the watermark is distortion-free, so the entropy column of \cref{thm:entropyclose} applies.
\textbf{(b)}~To implicate no innocent user above $\delta/N$, the consistency-list rule needs
$\log(N/\delta)/r_2$ tokens (amber) while surprisal thresholding needs $\log(N/\delta)/h$ (blue); the
shaded overcharge is the factor $h/r_2=1.61$, widening as $\Theta(\log N)$.}
\label{fig:exp-ems}
\end{figure}

\paragraph{Detection is not attribution, on the deployed green-list.}
The green-list of \citet{kirchenbauer2023watermark} is the extremal \emph{biasing} profile (the
$\mathcal B(\Delta)$ column of \cref{thm:dichotomy}): a context-keyed split greens a fraction
$\gamma=\tfrac14$ of the vocabulary and adds a fixed logit bias (here $2.0$). On GPT-2 the realized green
fraction is $0.62$ under the true key against $0.25$ under any innocent key. \Cref{fig:exp-greenlist}
exhibits \cref{thm:separation} in finite-$N$ form: detection power saturates by $n\approx28$ tokens
\emph{independently of $N$}, whereas the attribution length (here the $50\%$-accuracy crossing) grows
linearly in $\log N$ ($+2.3$ tokens per doubling, from $16$ to $61$ tokens as $N$ runs $2^{8}$ to
$2^{28}$). Between the two lies a
window in which the passage is detectable yet unattributable. Unlike the EMS scheme above, the
green-list \emph{shifts} the entropy: the two quality models of \cref{def:embedbias}, the same
separation read off either column.

\begin{figure}[t]
\centering
\includegraphics[width=\textwidth]{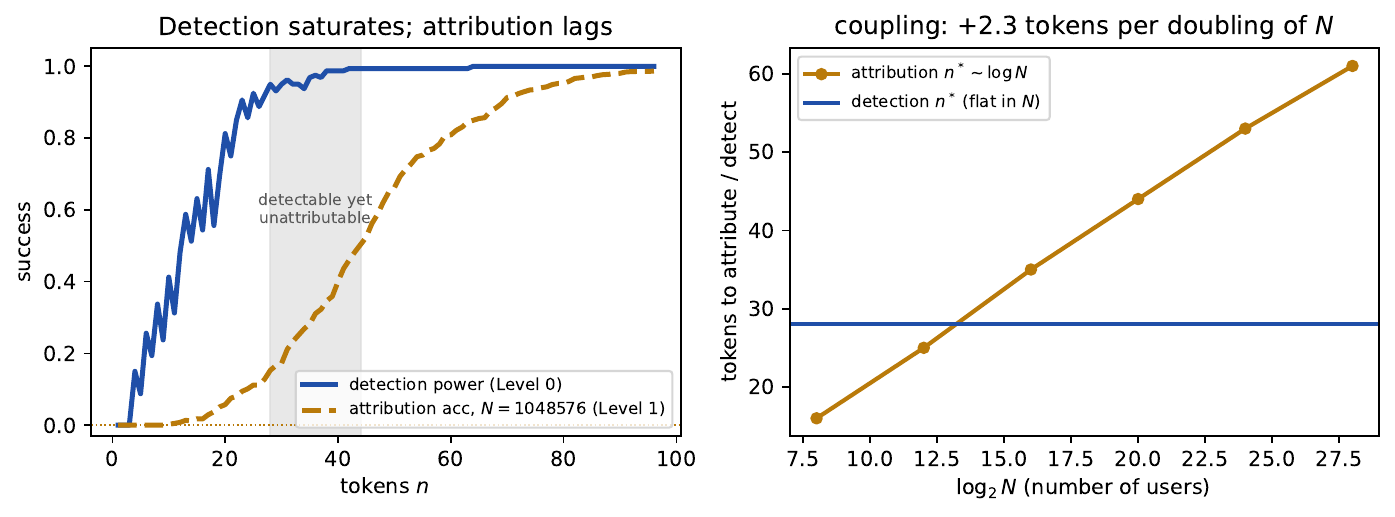}
\caption{\textbf{Detection is not attribution on the deployed green-list watermark, GPT-2.}
\textbf{(a)}~Detection power saturates by $n\approx28$ tokens while attribution accuracy (shown for
$N=2^{20}$) lags; in the shaded window the text is detectable yet unattributable. \textbf{(b)}~The
attribution length (the $50\%$-accuracy crossing) grows linearly in $\log N$ ($+2.3$ tokens per doubling
of $N$); the detection threshold is flat in $N$. The finite-$N$ form of the $\Theta(\log N)$ gap of
\cref{thm:separation}.}
\label{fig:exp-greenlist}
\end{figure}

\paragraph{The overcharge is a property of language, not of one model.}
The constant $h/r_2$ is the only model-dependent quantity in the trap, so we measure it across families.
On GPT-2, Pythia-410M \citep{biderman2023pythia}, and Qwen2.5-0.5B \citep{qwen2025qwen25} --- three
pretraining corpora and tokenizers, vocabularies from $50$k to $152$k --- the absolute rates fall as the
models sharpen ($h$ from $3.80$ to $2.34$ nats/token), but the overcharge is stable: $h/r_2=1.61,\,1.60,
\,1.57$ (\cref{fig:exp-crossmodel}). The collision trap reflects the gap between the Shannon and
R\'enyi-2 rates of natural language, not an artifact of any one model.

\begin{figure}[t]
\centering
\includegraphics[width=0.82\textwidth]{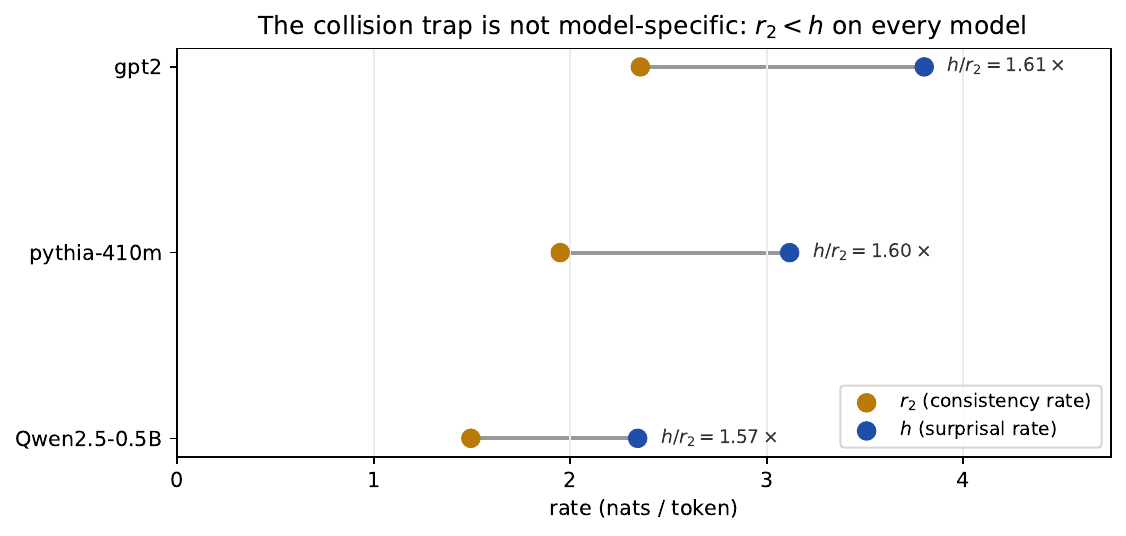}
\caption{\textbf{The overcharge $h/r_2$ across model families.} On GPT-2, Pythia-410M, and Qwen2.5-0.5B
the absolute rates differ, but $r_2<h$ with a stable ratio $\approx1.6$: the consistency-list rule costs
$\sim 60\%$ more tokens than surprisal thresholding on every model. The collision trap is a property of
natural-language entropy.}
\label{fig:exp-crossmodel}
\end{figure}

\section{The framework applied: a taxonomy of deployed watermarks}
\label{sec:taxonomy}

The classification of \cref{def:embedbias} is decided by one measurable quantity: the per-secret footprint
$T$ (the reading set: the positions on which the \emph{per-secret} kernels are marked, \cref{cor:bridge}),
read in the basis the extractor works in. The per-secret reading is not optional: the key-averaged support would misclassify every
distortion-free scheme as $\mathcal E(0)$, since averaging over the key can erase a mark present under every
individual key (parity witness, \cref{app:witnesses}). \Cref{tab:taxonomy} evaluates the criterion on
twenty-one representative schemes (the twenty deployed watermarks and the theoretical digit-stamp witness),
holding it fixed (so a scheme is embedding only if its per-secret kernels truly
mark $o(n)$ coordinates, not merely if its secret is readable from a sub-carrier); \cref{sec:schemefit}
works through each scheme individually. Two axes turn out to
matter, and they are orthogonal: the \emph{footprint} $T$, which governs the embedding/biasing
split and the $\Delta$-free decoupling of \cref{thm:dichotomy}; and the \emph{readout resolution}, whether
the mark is decoded globally or per-region, which governs Level 3. Three patterns follow.

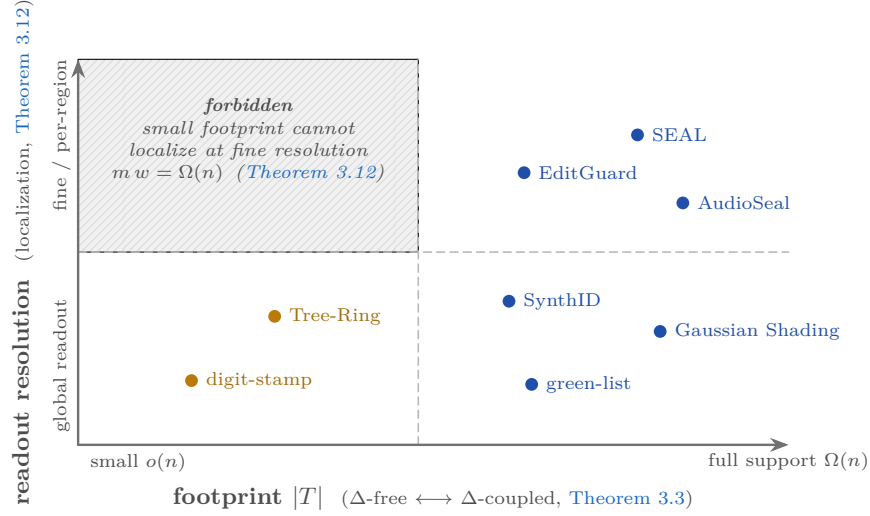
\begin{figure}[t]
\centering
\begin{tikzpicture}[font=\sffamily,>=Stealth,line cap=round,
   dotE/.style={circle,fill=attribgold,inner sep=1.7pt},      
   dotB/.style={circle,fill=detectblue,inner sep=1.7pt},      
   lbl/.style={font=\scriptsize,inner sep=1pt},
  ]
  \def\W{9.4}\def\H{5.1}\def\mx{4.5}\def\my{2.55}
  \fill[inkgray!8] (0,\my) rectangle (\mx,\H);
  \draw[pattern=north east lines,pattern color=inkgray!20] (0,\my) rectangle (\mx,\H);
  \node[align=center,font=\scriptsize\itshape,text=inkgray] at (\mx/2,{(\my+\H)/2+0.2})
     {\textbf{forbidden}\\ small footprint cannot\\ localize at fine resolution\\ $m\,w=\Omega(n)$ \,(\cref{thm:crop})};
  \draw[densely dashed,inkgray!40,line width=0.6pt] (\mx,0)--(\mx,\H);
  \draw[densely dashed,inkgray!40,line width=0.6pt] (0,\my)--(\W,\my);
  \draw[->,inkgray!80,line width=0.9pt] (0,0)--(\W,0);
  \draw[->,inkgray!80,line width=0.9pt] (0,0)--(0,\H);
  \node[lbl,text=inkgray,anchor=north] at (\W,-0.05) {full support $\Omega(n)$};
  \node[lbl,text=inkgray,anchor=north west] at (0.12,-0.06) {small $o(n)$};
  \node[lbl,text=inkgray,anchor=north,font=\small\bfseries] at (\W/2,-0.5)
     {footprint $|T|$ \;\textnormal{\scriptsize($\Delta$-free $\leftarrow\!\rightarrow$ $\Delta$-coupled, \cref{thm:dichotomy})}};
  \node[lbl,text=inkgray,rotate=90] at (-0.22,1.05) {global readout};
  \node[lbl,text=inkgray,rotate=90] at (-0.22,\H-1.0) {fine / per-region};
  \node[lbl,text=inkgray,anchor=south,rotate=90,font=\small\bfseries] at (-0.55,\H/2)
     {readout resolution \;\textnormal{\scriptsize(localization, \cref{thm:crop})}};
  \node[dotE] (stamp) at (1.5,0.85) {}; \node[lbl,anchor=west,text=attribgold!88!black] at (1.65,0.85) {digit-stamp};
  \node[dotE] (tree)  at (2.6,1.7) {};  \node[lbl,anchor=west,text=attribgold!88!black] at (2.75,1.7) {Tree-Ring};
  \node[dotB] (green) at (6.0,0.8) {};  \node[lbl,anchor=west,text=detectblue] at (6.15,0.8) {green-list};
  \node[dotB] (gshad) at (7.7,1.5) {};  \node[lbl,anchor=west,text=detectblue] at (7.85,1.5) {Gaussian Shading};
  \node[dotB] (synth) at (5.7,1.9) {};  \node[lbl,anchor=west,text=detectblue] at (5.85,1.9) {SynthID};
  \node[dotB] (seal)  at (7.4,4.1) {};  \node[lbl,anchor=west,text=detectblue] at (7.55,4.1) {SEAL};
  \node[dotB] (edit)  at (5.9,3.6) {};  \node[lbl,anchor=west,text=detectblue] at (6.05,3.6) {EditGuard};
  \node[dotB] (audio) at (8.0,3.2) {};  \node[lbl,anchor=west,text=detectblue] at (8.15,3.2) {AudioSeal};
\end{tikzpicture}
\caption{\textbf{The two orthogonal axes} of the taxonomy. The footprint $|T|$ (horizontal; small gold dots
$=$ embedding, full-support blue dots $=$ biasing) governs the $\Delta$-free decoupling of
\cref{thm:dichotomy}; the readout resolution (vertical) governs localization (\cref{thm:crop}). Being
independent, the two axes let deployed schemes spread across three quadrants, and the localizers (SEAL,
EditGuard, AudioSeal) reach Level 3 by \emph{fine readout}, not a small footprint. The fourth quadrant, small
footprint with fine resolution, is empty for a reason: \cref{thm:crop} forbids it, $m\,w=\Omega(n)$.}
\label{fig:twoaxes}
\end{figure}

\begin{table}[t]
\centering\footnotesize
\caption{Where deployed watermarks fall, with the footprint criterion applied uniformly. ``Basis'' is the
representation the extractor reads; ``Footprint'' is the per-secret footprint $T$ (\cref{cor:bridge}), full
($\Omega(n)$) or small ($o(n)$) (footprint read on per-secret kernels; \cref{app:witnesses}); ``Model'' is the
operative quality model; ``Readout'' is global (whole-carrier statistic) or fine (per-region); ``Levels''
lists the forensic rungs the base construction reaches, with parenthesized rungs reached only by descendant
or extension variants (multi-bit, many-user, or message-channel versions), not the base scheme. Almost every
deployed scheme is full-support (biasing); the genuine $o(n)$-footprint cases are Tree-Ring (a ring in the
Fourier basis) and the theoretical digit-stamp. Level-3 localizers are full-support with \emph{fine} readout,
exactly as \cref{thm:crop} requires.}
\label{tab:taxonomy}
\setlength{\tabcolsep}{4pt}
\resizebox{\textwidth}{!}{%
\begin{tabular}{@{}lllllll@{}}
\toprule
Scheme & Basis & Footprint & Model & Readout & Levels & Class \\
\midrule
\multicolumn{7}{@{}l}{\emph{Language models (token carrier): all biasing}}\\
Green-list \citep{kirchenbauer2023watermark}        & tokens & full   & $\mathcal B(\Delta)$ & global & 0,(1,2) & biasing \\
Unigram \citep{zhao2024unigram}                      & tokens & full   & $\mathcal B(\Delta)$ & global & 0       & biasing \\
SynthID-Text \citep{dathathri2024synthid}            & tokens & full   & d-free               & global & 0       & biasing (d-free) \\
Reweighting \citep{wu2024dipmark,hu2024unbiased}     & tokens & full   & d-free               & global & 0       & biasing (d-free) \\
Multi-bit \citep{yoo2023multibit,arcmark2026}        & tokens & full   & $\mathcal B(\Delta)$/d-free & bucketed & 0,2 & biasing \\
BiMark \citep{feng2025bimark}                        & tokens & full   & d-free               & bucketed & 0,2 & biasing (d-free) \\
Many-user \citep{cohen2024many}                      & tokens & full   & d-free               & global & 0--2    & biasing (d-free) \\
Distortion-free \citep{kuditipudi2023robust}         & tokens & full   & d-free               & global & 0       & biasing (d-free) \\
Undetectable \citep{christ2024undetectable}          & tokens & full   & d-free               & global & 0       & biasing (d-free) \\
PRC \citep{christ2024pseudorandom}                   & tokens & full   & d-free               & global & 0       & biasing (d-free) \\
SemStamp \citep{hou2024semstamp}                     & sentence-emb. & full & $\mathcal B(\Delta)$ & global & 0 & biasing \\
\midrule
\multicolumn{7}{@{}l}{\emph{Images / audio}}\\
Tree-Ring \citep{wen2023treering}                    & Fourier & \textbf{small} & d-free & global & 0,(1) & \textbf{embedding} \\
RingID \citep{ci2024ringid}                          & Fourier & full  & d-free (part.) & global & 0,1 & biasing \\
Gaussian Shading \citep{yang2024gaussianshading}     & latent  & full  & d-free  & global & 0,2 & biasing (d-free) \\
ROBIN \citep{huang2024robin}                         & Fourier (interm.) & full & $\mathcal B(\Delta)$ & global & 0 & biasing \\
HiDDeN \citep{zhu2018hidden}                         & pixels  & full  & $\mathcal B(\Delta)$ & global & 0,2 & biasing \\
Stable Signature \citep{fernandez2023stablesig}      & feature & full  & $\mathcal B(\Delta)$ & global & 0,2 & biasing \\
SEAL \citep{seal2025}                                & latent  & full  & d-free  & \textbf{fine} & 0--3 & biasing (d-free) \\
EditGuard \citep{editguard2024}                      & pixels  & full  & $\mathcal B(\Delta)$ & \textbf{fine} & 0,(2),3 & biasing \\
AudioSeal \citep{sanroman2024audioseal}              & samples & full  & $\mathcal B(\Delta)$ & \textbf{fine} & 0,(2),3 & biasing \\
\midrule
\emph{(Digit-stamp, \cref{thm:dichotomy}(b))}        & tokens  & \textbf{small} & $\mathcal E(m)$ & global & 0--2 & \textbf{embedding} \\
\bottomrule
\end{tabular}%
}
\end{table}

\paragraph{1. Almost every deployed watermark is biasing, and the $\Delta$-free class is nearly empty.}
Applied consistently, the footprint criterion puts not only the LLM schemes but also the neural image and audio
watermarks on the biasing side: green-list and its descendants tilt every token; the distortion-free family
keeps $\pbar=\pnat$ but spreads its signal over the whole sequence; and SEAL, EditGuard, AudioSeal, and Stable
Signature perturb the entire latent, image, or waveform. For the distortion-free rows the per-secret reading
shows this is forced, not incidental: the per-secret kernels of an exact sampler are near-deterministic in the
key (each conditional collapses onto the token the key selects), so every position is marked (full support), and the per-token conditional divergence exceeds any finite cap, placing these schemes outside
$\mathcal B(\Delta)$ for every finite $\Delta$. They are governed by the entropy column, and by
\cref{thm:entropyclose} that column is now closed: attribution among $N$ users at $n=(1+o(1))\log N/h$ is
tight for the exact-sampler schemes with i.i.d.\ keys that the theorem covers (the computationally
undetectable constructions, governed by seed entropy, sit outside its hypotheses). The only genuine $o(n)$-footprint schemes in
\cref{tab:taxonomy} are Tree-Ring (a ring on a few Fourier shells) and the theoretical digit-stamp. Even
RingID \citep{ci2024ringid}, the Tree-Ring successor engineered for multi-key identification, does not join
them: its key capacity comes from a key-specific noise pattern occupying an \emph{entire} latent channel
alongside the discretized rings, so its per-secret footprint is a constant fraction of the latent
(\cref{sec:schemefit}) --- attribution capacity was bought by leaving the small-footprint class. The
$\Delta$-free decoupling of \cref{thm:dichotomy}(b) is therefore realized by essentially no deployed scheme,
which is precisely why deployable, edit-robust $\Delta$-free attribution is open (\cref{sec:discussion}), not
a solved case we are merely cataloguing.

\paragraph{2. Localization is a second, orthogonal axis (readout resolution), and the framework predicts
its cost.} Whether a scheme reaches Level 3 is decided not by its footprint but by whether its (full-support)
mark is decoded per-region: SEAL reads a per-patch hash, EditGuard a per-region mask, AudioSeal a per-sample
residual, while green-list reads one whole-passage count. This is not in tension with \cref{thm:crop}; it is
its content. Crop-robust localization at resolution $w$ \emph{requires} a carrier footprint of size
$\Omega(n/w)$, so a scheme that localizes finely and survives cropping \emph{must} be full-support. The deployed localizers are
indeed observed to be full-support (\cref{tab:taxonomy}), consistent with this. The framework thus predicts both halves of the empirical
benchmark on which fine-readout embedding-style schemes localize at high AUC while global-readout latent
schemes (Gaussian Shading, Tree-Ring in pixels) score near zero \citep{seal2025}.

\paragraph{3. The classification is basis-relative.} Tree-Ring is the one deployed embedding scheme, and only
in the right basis: in the latent-noise spatial basis the inverse-DFT of the ring spreads over all noise
coordinates (and the noise-to-image map is the full diffusion sampler, not a transform), so there it is
biasing; in the Fourier basis of the initial noise the mark is a small object on a few frequency shells,
and the detector reads it there by inverting the generation.\footnote{Like sparsity, the footprint is
relative to a fixed dictionary \emph{and} a fixed coordinate ordering: for Tree-Ring the dictionary is the
latent-Fourier one and the ordering is its radial indexing of shells. Only the mass $\norm{\nu}_1$ is
invariant under a change of either (\cref{prop:profile}).} The localization map can be a change of basis;
the framework asks in which representation the mark becomes an $o(n)$ object, and answers that for almost all
of today's schemes there is none. SemStamp \citep{hou2024semstamp} makes the relativity concrete a second
time, in the opposite direction: its extractor reads a sentence-embedding dictionary, where the carrier's
coordinates are sentences rather than tokens, and the basis change re-grains the footprint without shrinking
it: every sentence is steered into a keyed LSH region, so the scheme is full-support at either granularity.
A change of basis can relocate the mark (Tree-Ring) or merely change the units in which it is read
(SemStamp); only the former creates an $o(n)$ object.

\paragraph{Consistency with the definition.} A green-list key \emph{is} recoverable from a long window and a
SEAL patch-hash \emph{is} readable from a crop, yet both schemes mark a constant fraction of coordinates
under every key, so both are biasing under the per-secret footprint (\cref{def:embedbias}); they differ in
readout resolution, not in footprint. Keying the embedding/biasing line on ``readable from a sub-carrier''
would split them; keying it on the footprint does not, and \cref{tab:taxonomy} is the consistent application.

\section{Related work}
\label{sec:related}

\paragraph{Detection statistics (Level 0), imported.} A line of recent work pins down the sample complexity
of watermark detection. \citet{huang2023optimal} give uniformly most powerful tests and matching minimax
Type-II bounds; \citet{li2024statistical} derive closed-form detection efficiency and the pivotal statistic;
\citet{cai2024better} establish the per-step KL detectability floor $1-\sqrt{1-e^{-D_{\mathrm{KL}}}}$ and the
fidelity--detectability frontier (sharpened by \citealp{he2025optimal} to a tight $f$-divergence
bound). These fix our Level-0 rung at $\ntask{\det}=\Theta(1/\Delta)$; we treat the rate and the
$\Delta\!\leftrightarrow\!(\alpha,\beta)$ knob as given and build the ladder above them.

\paragraph{Multi-user watermarking and fingerprinting (Level 1).} \citet{cohen2024many} construct
language-model watermarks for many adaptive users with a detect-short / trace-long gap, built from Tardos and
Boneh--Shaw fingerprinting codes \citep{tardos2008optimal, boneh1998collusion}. This is the closest prior
art, and it proves the direction \emph{complementary} to ours: achievability of cheap detection with
expensive tracing in a cryptographic model. Their sole optimality remark is a counting argument (``$\log n$
bits are needed to identify a user''); there is no statistical converse on length, no distortion budget, and
no biasing side. The matching converse we use is not new: the \emph{distortion-constrained fingerprinting
converse} of \citet{moulin2008universal} and \citet{somekhbaruch2005capacity} already proves that
fingerprinting capacity under a fidelity constraint $\to0$ as the constraint tightens, exactly the
``couples to the budget'' phenomenon (with $L_2$-distortion in place of KL). We import it and instantiate
it for the per-token KL channel rather than claim it. Classical fingerprinting \citep{tardos2008optimal}
supplies the $\Theta(c^2\log N)$ scaling under the combinatorial marking assumption (the same code machinery
drives tight lower bounds well beyond watermarking, e.g.\ in differential privacy,
\citealp{bun2014fingerprinting}), but with no detection task and hence no separation. Our addition over this line is the contrast with the footprint model
(\cref{thm:dichotomy}), not the converse. \citet{jiang2024attribution} study detection versus attribution
empirically at fixed length and argue, under a strict threshold, that the two essentially coincide. This is
the opposite phenomenology to our impossibility regime, which their model (embedding bitstrings, no biasing,
no length law) cannot express. Concurrent work of \citet{song2025ideal} studies attribution faithfulness without
a sample-complexity ladder. On the deployed-attribution axis, \citet{sander2024radioactive} recover a
membership/identity signal from watermarked text (``radioactivity''); this is a Level-1-flavored capability in
practice, but it asks whether a model was trained on marked text rather than which of $N$ users emitted a
fixed sample, and it carries no length law. Multi-bit payload capacity (Level 2) is studied by
\citet{covert2026multibit} via Gelfand--Pinsker coding and by \citet{yoo2023multibit, arcmark2026}; our
embedding-side payload term is consistent with these, and the $\log N$ attribution surcharge is the axis they
lack.

\paragraph{The embedding/biasing taxonomy (Level 3 and the dichotomy).} The recent SoK
\citep{sok2024watermarking} taxonomizes ``statistical'' (token-probability) versus ``embedding''
(content-feature) schemes and argues detection and attribution should be separate functions, but on the
robustness-versus-unforgeability \emph{security} axis, with no sample-complexity content; the
statistical/embedding split there is an engineering taxonomy, never a theorem or a phase boundary. Empirical
localization results confirm the mechanism without explaining it: SEAL \citep{seal2025} localizes at
$98.2\%$ AUC where distributional schemes score $0\%$ on the same task, and EditGuard \citep{editguard2024}
gives tamper localization for image watermarks. Our \cref{thm:dichotomy} is, to our knowledge, the first
theorem stating the existence of a localization map as the boundary between $\Delta$-coupled and $\Delta$-free
forensic cost.

\paragraph{Concurrent information-theoretic frameworks.} Distributional Information Embedding
\citep{distribembed2025} is the closest concurrent framework: it analyzes the watermarking trade-off at the
level of the scalar $\I(S;X)$. Our profile $\nu(t)$ is the positional refinement of that scalar (its mass
recovers it), and the forensic ladder is the use the refinement is put to: tasks become functionals of $\nu$
that a single number cannot separate. The unified-optimization framework of \citet{gloaguen2026unified}
treats scheme \emph{design} as constrained optimization, the design axis, orthogonal to sample complexity.
Proportion estimation \citep{li2025proportion} targets a different Level-3-adjacent estimand: what
\emph{fraction} of a text is marked, not which positions; WISER \citep{bonnerjee2025wiser} localizes
multi-segment watermarks via change-point methods, an empirical/finite-sample method with no
information-theoretic lower bound.

\paragraph{Steganography.} The information-theoretic model of steganography
\citep{cachin2004stego, fridrich2009stego} and the square-root law of steganographic capacity
\citep{ker2008sqrtlaw} concern the \emph{bounded-total-KL} regime, where keeping the total divergence
constant forces an embeddable payload of order $\sqrt n$. Our regime is the \emph{linear-total-KL} one: a
per-token budget $\Delta$ gives total budget $n\Delta$ and a payload linear in $n$. The two should not be
conflated; the forensic-recovery budget (\cref{thm:budget}) is the linear-regime analogue of the stego capacity
bound, now resolved by forensic level and by mechanism. The structurally closest predecessor of a
per-position information quantity is \citet{fillerfridrich2009fisher}: under mutually independent embedding,
a per-element \emph{Fisher} information, summed across the cover, determines the capacity of
$\varepsilon$-secure steganography --- the engine of the square-root law. The profile $\nu$ differs on three
axes: it is a mutual-information quantity, not a second-order (Fisher) one; it feeds per-task functionals
(mass \emph{and} shape), not a single aggregate capacity; and it is defined on autoregressive conditionals,
not i.i.d.\ cover models.

\paragraph{Covert communication and perfect steganography.} The square-root law of covert capacity
\citep{bash2012srl, ker2008srl, bloch2016covert} lives in the budget regime $\KLop\le\eps$: a fixed total
divergence forces an $O(\sqrt n)$ payload. The statistically distortion-free column of \cref{tab:ladder} is
the $\KLop=0$ \emph{keyed} regime (Cachin-style perfect steganography \citep{cachin2004stego}) where
the binding resource is not a divergence budget but the generation entropy $h$. \Cref{thm:entropyclose} is
the multi-user attribution rate for that regime; to our knowledge no prior work gives a tight $h$-rate for
$N$-user attribution. The closest, \citet{cohen2024many}, is cryptographic (security-parameter counting,
with adaptive robustness as the focus), complementary, not rate-tight.

\paragraph{Identification via channels.} \citet{steinberg2001identification} cast watermark identification
as identification via channels in the sense of \citet{ahlswede1989identification}: the blocklength is held
fixed and the number of \emph{verifiable} signatures grows doubly exponentially in it. This is the dual of
our question, which fixes the user population $N$ and asks for the sample cost $n(N)$. Their split between
verification (``is it signature $s$?'') and decoding (``which signature?'') has, to our knowledge, no
sample-complexity counterpart; we pose it as an open question in \cref{sec:discussion}.

\paragraph{Impossibility and robustness.} A separate thread shows strong watermarks can be removed under mild
assumptions \citep{zhang2023sand, francati2026codinglimits}, the latter giving a sharp coding-theoretic
ceiling on tamper robustness, and that reusing a detector across keys inflates false positives
\citep{furussell2025falsedetection}; on the attack side, \citet{abdalla2025mixtures} exhibit
statistical-to-computational gaps for removal. The false-detection result is a $\Delta$-dependent upper bound
on one parasitic detector, not a floor over all detectors on natural content; we discuss the corresponding
open question (a false-positive floor) in \cref{sec:discussion}. Recent design, detection, attack, and
localization work adjacent to the framework is surveyed further in \cref{sec:morerelated}.

\section{Discussion}
\label{sec:discussion}

\paragraph{What the dichotomy says for regulation.} Current statutes ask for marks ``in a machine-readable
format'' (EU AI Act, Art.~50), ``noticeable labels'' (China's Deep Synthesis provisions), or ``latent
disclosure'' (California SB~942). Read operationally, each presumes an extractable object: the embedding
side of \cref{def:embedbias}. \Cref{thm:dichotomy} makes the cost of that presumption precise: a biasing
watermark can satisfy a Level-1 (attribution) or Level-2 (extraction) mandate only at a sample cost
$\Omega(\log N/\Delta)$ that \emph{diverges} as the mark is made stealthier ($\Delta\to0$), whereas an
embedding watermark pays a fixed $\Theta(\log N/\log q)$ independent of fidelity. A regulator who wants
attribution is, whether or not the text says so, mandating the embedding ontology. We state this as an
informed structural consequence, not a legal interpretation.

\paragraph{Which rate is a sample complexity.} The closed column also settles a question the naive analysis
gets wrong: the R\'enyi-2 rate $r_2$ governs the size of the falsely-consistent list, but the Shannon rate
$h$ governs the \emph{decision}, and only the decision is what a forensic guarantee names
(\cref{rem:listsize}). Whether an edit-robust decoder, which can no longer check consistency token-for-token,
is forced back onto the larger $r_2$ rate is the open question we turn to next.

\paragraph{Localization (Level 3) and the witness gap.} \Cref{thm:crop} makes the redundancy--resolution
tradeoff quantitative: crop-robust localization at resolution $w$ forces $m\,w=\Omega(n)$ (precisely
$|R^{+}|\ge(n-w+1)/w$ for the dilated support and $|R|\ge|R^{+}|/(2r+1)$, with $|R|\le m$ in the
factorizing regime). This explains the empirical split (\cref{sec:taxonomy}) and exposes the honest limit of our
decoupling witness. The digit-stamp of \cref{thm:dichotomy}(b) has footprint $m=\Theta(\log N/\log q)=o(n)$, so
by \cref{thm:crop}(ii) it cannot be crop-robust at any fine resolution: a single splice destroys it. The
\emph{deployed} schemes that do localize (SEAL at $98.2\%$ AUC, EditGuard, AudioSeal) are, by \cref{thm:crop},
forced to be \emph{full-support}; they are biasing and localize only because they read their delocalized
mark at fine per-region resolution. So the clean $\Delta$-free decoupling and
operational edit-robust localization pull in opposite directions, and the question our witness leaves open
(\emph{is there a deployable (imperceptible, edit-robust) scheme whose attribution cost is $\Delta$-free?})
is, by \cref{thm:crop}, exactly the question of whether one can be $\Delta$-free on the payload while paying
the $m\,w=\Omega(n)$ tax only on the locating layer. We conjecture this is achievable with a layered
PRC-style construction but do not prove it. The same gap is now the live one on the distortion-free side:
\cref{thm:entropyclose} closes the attribution rate there, but its witness decoder is exact-alignment
(consistency is checked token-for-token against the key) and a single edit destroys it. Under a soft
(edit-robust) notion of consistency the true key is no longer deterministically consistent, the
tie-abstention device behind \cref{rem:listsize} loses its zero-false-positive guarantee, and a genuine
$r_2$-type cost may resurface. Whether edit-robust distortion-free attribution still costs $\log N/h$, or
genuinely pays the larger $\log N/r_2$, is the sharpest open problem the framework raises.

\paragraph{Verification versus search (open).} Our attribution task is a \emph{search}: which of $N$ users
emitted $X$? The verification variant (is it user $u$?) is a composite binary test, and the
identification-via-channels literature \citep{steinberg2001identification, ahlswede1989identification}, with a
covert (distortion-free) analogue in \citet{houkramer2021covertid}, suggests the two scale very differently:
at fixed blocklength the number of verifiable signatures is doubly exponential while the decodable ones are
merely exponential. The sample-complexity dual would make
verification $N$-free, an exponential separation from search's $\Theta(\log N/\Delta)$; the proof obstacle
is testing the composite null (the $N-1$ other keys) uniformly. The sample-complexity version of the
verification-vs-decoding split is open; we have not pursued it here.

\paragraph{A false-positive floor (open).} The seed of this work conjectured an information-theoretic
false-positive floor for biasing watermarks on natural content: that no detector can drive the
false-positive rate to zero at fixed power when $\pwm$ and $\pnat$ overlap. Against a computationally
unbounded distinguisher this is delicate, since undetectable constructions \citep{christ2024undetectable} make
the marginal gap vanish, so any floor must be stated in the key-known or information-theoretic-adversary
regime. \citet{furussell2025falsedetection} give a related upper bound (false positives inflate when a
detector is reused across keys), but a genuine floor over \emph{all} detectors on a natural-content null
remains open. We leave it as a conjecture rather than fold a weaker statement into the main results.

\paragraph{Tightness and generality of the achievability (open).} Three gaps remain between our converses and
the constructions that meet them. The $\mathcal B(\Delta)$ achievability of \cref{thm:dichotomy}(a) matches
its converse only up to the per-position capacity constant, and only for memoryless or regular strongly-mixing
sources; the exact $q$-ary tilting constant, and a general-ergodic achievability that removes the mixing
hypothesis the conditional tilting code needs, are both open. On the distortion-free side the entropy column
assumes i.i.d.\ per-user keys; correlated or coded assignments (Tardos-style,
\citealp{tardos2008optimal}) might trade candidate-list size against length, a different point on the same
rate curve. None of these changes a rate order; each sharpens a constant or widens a hypothesis.

\paragraph{Existence, identity, and the ledger.} Beneath the rates, the ladder separates two kinds of
question. Detection asks whether a mark \emph{exists}: a two-point hypothesis test, whose cost is set by a
fixed divergence and is indifferent to the size of the user population. Attribution asks \emph{who}: an
estimation among $N$ alternatives, which cannot conclude before $\log N$ nats of identity have arrived.
Existence is cheap; identity is priced; and the window of \cref{thm:separation} is what lies between, a text
that can raise suspicion but cannot yet support a name. The form of a forensic guarantee matters as much as
its rate: \cref{thm:entropyclose}(i) prices false accusation per innocent rather than on average (no innocent
is implicated except with probability $\delta/N$, exactly and non-asymptotically), which is the quantitative
shape of a presumption of innocence. And the account every level draws on is settled before any question is
asked: by data processing (\cref{thm:budget}), reading the carrier can only spend information that generation
deposited, never mint more. What can ever be detected, attributed, or extracted from a scheme's output is
fixed, token by token, by the profile $\nu$ it writes, and where it can be localized is fixed by the shape of
$\nu$ together with the edit coupling on the carrier; the quality models are covenants on the deposits
($\Delta$ at a time, or only what the source's entropy affords); attribution and extraction are withdrawals
against one balance (\cref{prop:region}). Design is the choice of a deposit schedule, and a forensic mandate
is, in the end, a constraint on $\nu$.

\paragraph{Limitations.} (i) Our bounds are information-theoretic: $\Phibud(n)=I(S;X)$ caps every
\emph{computationally unbounded} extractor, so the lower bounds are unconditional, but the distortion-free
budget (\cref{asm:free}) uses statistical undetectability, with a negligible slack under merely computational
undetectability. (ii) The Level-0 detection rate we import assumes the standard
independent/stationary-token model \citep{li2024statistical, huang2023optimal}; our co-located gap needs only
that the key-unknown detection error vanishes on the window, which in that model follows from a positive
per-token mixture divergence via the relative-entropy ergodic theorem (a large \emph{total} KL alone would not
suffice for a single joint sample). (iii) The within-scheme impossibility
(\cref{thm:separation}, ``detectable yet unattributable at one length'') is an asymptotic statement: it lives
on the window $\omega(1/\Delta)\le n<(1-\eta)\log N/\Delta$, which is wide as $N\to\infty$ but narrow at
moderate $N$, so our simulations (\cref{sec:numerics}) accordingly exhibit the robust form of the separation
(attribution threshold $\Theta(\log N)$ above the detection threshold) rather than the razor-thin point.
(iv) The embedding ``stamp'' of \cref{thm:dichotomy} is the simplest witness of decoupling, not a deployable
scheme; robust embedding watermarks (commitment- or code-based; SEAL, EditGuard empirically) realize the same
decoupling with imperceptibility and robustness, which we cite rather than reprove. (v) The distortion-free
achievability of \cref{thm:entropyclose,cor:extraction-h} is information-theoretic in its decoder as well: the
witness extractor searches the key set, so it runs in time $\Theta(N)=\Theta(2^\ell)$ in the payload; an
efficient $\Theta(\ell/h)$-sample extractor is a separate requirement, which the ``machine-readable mark''
reading of a regulatory mandate implicitly wants and which we leave open.

\paragraph{Summary.} A watermark is an object in its carrier or a property of its law, and the footprint of
\cref{def:embedbias} is where that informal distinction becomes measurable. Detection is the cheapest
question one can ask of either kind, and the only one whose cost is mechanism-blind. Every question above it
(who, what, where) is metered by a single distortion budget for schemes confined to the distributional
quality model, and bought at a $\Delta$-free price by schemes that can use a bounded footprint; the
localization map is the structural feature that admits the cheaper model, and it enables without deciding.
The three questions above detection are readings of one measure, two of its mass and one of its shape with
the edit coupling on the carrier, and their prices differ because the readings do.

\section{Conclusion}
\label{sec:conclusion}

A watermark decides, token by token as it generates, how much of its secret the world will ever recover; the
information profile $\nu(t)=I(S;X_t\mid X_{<t})$ is the record of that decision, and the forensic
questions above detection are readings of it. The prices follow. The mass $\norm{\nu}_1$ funds attribution
and extraction in one of two currencies, $\Delta$ per token against a fidelity budget or $h$ per
token against the entropy the source affords; the shape of $\nu$, with the edit coupling on the carrier,
governs localization under the uncertainty $m\,w=\Omega(n)$; detection is the mechanism-blind baseline
beneath them. The question that would have led our list of
open problems, a distortion-free achievability to match the entropy converse, is closed by
\cref{thm:entropyclose}: attribution costs $(1+o(1))(\log N+\log(1/\delta))/h$ at full stationary-ergodic
generality, with an exact per-innocent false-positive guarantee, and extraction follows at $\Theta(\ell/h)$
(\cref{cor:extraction-h}).

With this the ladder's currencies are identified at every rung: $\Delta$, $h$, and the geometry of the
footprint. What the framework leaves open --- and the sharpest of it is a matter of robustness, not of
information --- we set out in \cref{sec:discussion}.

Nothing in the profile is special to watermarking. $\nu(t)=I(S;X_t\mid X_{<t})$ prices what any secret
threaded through a sequential generator makes recoverable, and its mass-versus-shape reading should apply
wherever information is embedded position by position --- steganography, fingerprinting, covert
communication. We have used it to say one thing about one domain: the forensic power of a watermark is fixed
by the information it makes recoverable, not by the mechanics of how the mark is embedded.

\bibliographystyle{plainnat}   
\bibliography{refs}

\appendix
\newpage
\section{Notation}
\label{sec:notation}

\Cref{tab:notation} collects the symbols used throughout. The dependency structure of the results is mapped in
\cref{fig:resultmap} (\cref{sec:results}); the imported and textbook \emph{facts} each proof draws on are
catalogued in \cref{sec:imported} and cited inline at each use. All logarithms and information quantities are
in nats unless a base is shown; per the convention of \cref{sec:prelim}, the profile $\nu$ and the budget
$\Phibud$ are read conditional on the registry $\kappa$ (written $I(S;X)$ for brevity).

{\small
\begin{longtable}{@{}l @{\hspace{1.4em}} l @{\hspace{1.4em}} l@{}}
\caption{Notation used throughout the paper.}\label{tab:notation}\\
\toprule
Symbol & Meaning & Defined in \\
\midrule
\endfirsthead
\toprule Symbol & Meaning & Defined in \\ \midrule
\endhead
\multicolumn{3}{@{}l}{\textit{\textcolor{detectblue}{Model and keys}}}\\[1pt]
$n$ & carrier length (number of tokens) & \cref{sec:prelim} \\
$\Sigma,\ q=|\Sigma|$ & token alphabet and its size & \cref{sec:prelim} \\
$X\in\Sigma^n$ & the carrier (the generated text/sample) & \cref{sec:prelim} \\
$\pnat$ & natural (unmarked) law on $\Sigma^n$ & \cref{sec:prelim} \\
$\pwm^{(k)}$ & marked law under key $k$ & \cref{sec:prelim} \\
$\pbar=\E_K\,\pwm^{(K)}$ & key-averaged (mixture) law & \cref{sec:prelim} \\
$K\in\Keys$ & secret key; key space & \cref{sec:prelim} \\
$\kappa=\{k_1,\dots,k_N\}$ & key registry (enrolled keys; the informed decoder's side information) & \cref{sec:prelim} \\
$\hat p_\kappa=\tfrac1N\sum_v\pwm^{(k_v)}$ & realized key mixture given the registry $\kappa$ & \cref{app:entropyclose} \\
$S$ & secret: user $U\in[N]$, payload $W\in\{0,1\}^\ell$, or both & \cref{sec:prelim} \\
$N,\ \ell$ & number of users; payload length (bits) & \cref{def:tasks} \\
\addlinespace
\multicolumn{3}{@{}l}{\textit{\textcolor{detectblue}{The information profile}}}\\[1pt]
$\nu(t)=I(S;X_t\mid X_{<t})$ & information profile (secret information added by token $t$) & \cref{prop:profile} \\
$\|\nu\|_1=I(S;X)=\Phibud(n)$ & its mass $=$ the forensic-recovery budget & \cref{def:budget} \\
$\supp\nu$ & its shape (support) & \cref{prop:profile} \\
$T$ & per-secret footprint (the marked positions) & \cref{def:embedbias} \\
$\bar d(t)$ & per-secret distortion profile & \cref{sec:prelim} \\
\addlinespace
\multicolumn{3}{@{}l}{\textit{\textcolor{detectblue}{Quality models}}}\\[1pt]
$\Delta$ & per-token distortion (KL) budget & \cref{asm:bounded} \\
$\mathcal B(\Delta)$ & bounded-distortion / biasing class; $L^\infty$ cap $\nu(t)\le\Delta$ & \cref{asm:bounded} \\
$\mathcal E(m),\ m$ & footprint / embedding class; $L^0$ cap $|\supp\nu|\le m$ & \cref{def:embedbias} \\
$\Lmap_k$ & localization map & \cref{def:embedbias} \\
$R,\ R^{+}$ & carrier footprint and its $r$-dilation (Level 3) & \cref{asm:local} \\
$w$ & crop resolution (window length) & \cref{def:tasks} \\
\addlinespace
\multicolumn{3}{@{}l}{\textit{\textcolor{detectblue}{Forensic tasks and rates}}}\\[1pt]
$\ntask{\det},\ntask{\att},\ntask{\ext}$ & sample complexity of detection / attribution / extraction & \cref{def:tasks} \\
$P_e$ & attribution error probability & \cref{def:tasks} \\
$\delta$ & per-innocent false-positive level (attribution) & \cref{thm:entropyclose} \\
$\lambda$ & detection soundness parameter ($\mathrm{FP}\le 2^{-\lambda}$) & \cref{sec:prelim} \\
\addlinespace
\multicolumn{3}{@{}l}{\textit{\textcolor{detectblue}{Information-theoretic quantities}}}\\[1pt]
$h,\ h_n$ & entropy rate; finite-$n$ entropy $H(\pnat^{(n)})/n$ & \cref{asm:free} \\
$r_2=-\log\sum_x\pnat(x)^2$ & R\'enyi-2 rate & \cref{rem:listsize} \\
$\imath(X)=-\log\pnat(X)$ & realized surprisal & \cref{def:witness} \\
$\tau=\log(N/\delta)$ & surprisal threshold & \cref{def:witness} \\
$I(\cdot;\cdot),\ H(\cdot),\ \KL{\cdot}{\cdot}$ & mutual information, entropy, relative entropy (nats) & \cref{sec:itbackground} \\
$\JSD,\ \Unif$ & Jensen--Shannon divergence; uniform distribution & \cref{sec:itbackground} \\
\bottomrule
\end{longtable}
}

\section{Information-theoretic background}
\label{sec:itbackground}

This section is a self-contained primer on the information-theoretic notions the paper uses; readers fluent in
information theory can skip it. All quantities are in nats (natural logarithms), and $\pnat,\pwm$ are laws on
the length-$n$ carrier $\Sigma^n$.

\paragraph{Surprisal and entropy.} The \emph{surprisal} of an outcome $x$ under a law $p$ is
$\imath_p(x)=-\log p(x)$: a rare outcome is surprising. Its average is the \emph{entropy}
$H(p)=\E_{x\sim p}\,\imath_p(x)=-\sum_x p(x)\log p(x)$, the mean uncertainty in a draw from $p$. For a sequence
the per-symbol entropy $h_n=H(\pnat^{(n)})/n$ measures how much fresh randomness each token carries; for a
stationary source it decreases to the \emph{entropy rate} $h=\lim_n h_n$.

\paragraph{Mutual information and the chain rule.} The \emph{mutual information}
$I(S;X)=H(S)-H(S\mid X)$ is how much observing $X$ lowers the uncertainty about $S$. It is the right currency
for a forensic decoder: by data processing (below), no function of $X$ recovers more than $I(S;X)$ nats about
$S$. The \emph{chain rule} resolves it position by position, $I(S;X)=\sum_t I(S;X_t\mid X_{<t})$, and the
summand $\nu(t)=I(S;X_t\mid X_{<t})$ is the information profile that organizes the paper.

\paragraph{Relative entropy, the golden formula, and data processing.} The \emph{relative entropy} (or
Kullback--Leibler divergence) $\KL{p}{q}=\sum_x p(x)\log\frac{p(x)}{q(x)}\ge0$ measures how distinguishable $p$
is from $q$; it is the exponent at which a likelihood-ratio test separates them. It is our fidelity currency:
a mark that moves a token law by KL $\Delta$ spends $\Delta$ of distortion there. The \emph{golden formula}
(\cref{fact:golden}) writes a mutual information as a divergence against \emph{any} reference $Q$,
$I(K;X)=\min_Q\E_K\KL{P_{X\mid K}}{Q}$, so $I(K;X)\le\E_K\KL{P_{X\mid K}}{Q}$ for every $Q$; taking $Q=\pnat$
turns the bound into the distortion budget. The \emph{data-processing inequality}, that a Markov chain
$S\to K\to X$ forces $I(S;X)\le I(K;X)$, says post-processing cannot manufacture information. Together they are
the forensic-recovery budget (\cref{thm:budget}).

\paragraph{Fano's inequality.} To convert an information bound into an error bound we use \emph{Fano's
inequality} (\cref{fact:fano}): naming one of $N$ equiprobable users with error $P_e$ needs
$I(U;X)\ge(1-P_e)\log N-\log2$, i.e.\ $P_e\ge 1-\frac{I(U;X)+\log2}{\log N}$. Recovering $\log N$ nats of
identity therefore demands a budget of $\approx\log N$ nats; this is the engine of every lower bound here.

\paragraph{Typicality and Stein's lemma.} The \emph{asymptotic equipartition property} says the surprisal of
a long sample concentrates: $\imath_{\pnat}(X)/n\to h$ almost surely (the Shannon--McMillan--Breiman theorem),
so a typical sample has probability $\approx e^{-nh}$. This is exactly why the surprisal threshold of
\cref{thm:entropyclose} works: the true (typical) text clears any sub-$h$ bar $\tau$ once $n\gtrsim\tau/h$. The
dual statement for testing is \emph{Stein's lemma} and its ergodic extension (\cref{fact:stein}): the optimal
test of $\pbar$ against $\pnat$ has error decaying at the KL-rate, so a positive per-symbol divergence makes
the two laws asymptotically separable.

\paragraph{R\'enyi-2 entropy and collisions.} The \emph{R\'enyi-2 entropy rate} $r_2=-\log\sum_x\pnat(x)^2$
governs \emph{collisions}: two independent draws from $\pnat$ coincide with probability
$\sum_x\pnat(x)^2=e^{-r_2}$ per symbol. R\'enyi entropy is non-increasing in its order, so $r_2\le h$, with
equality only for a source uniform on its support; the gap $h/r_2$ is precisely the factor by which the naive
consistency analysis overcharges (\cref{rem:listsize}).

\paragraph{Jensen--Shannon divergence.} The \emph{Jensen--Shannon divergence}
$\JSD(p,q)=\tfrac12\KL{p}{m}+\tfrac12\KL{q}{m}$, with $m=\tfrac12(p+q)$, is a symmetric, bounded divergence; it
equals the mutual information between a fair coin choosing $p$ or $q$ and the resulting sample. It is the
\emph{yield} of the tilting code (\cref{thm:dichotomy}): fidelity is spent in KL, but the bits a decoder reads
back per token are measured by the JSD between the two tilted kernels.

\section{Deferred proofs}
\label{sec:appendix}

The starred macros reprint each statement with its original number, so the appendix copy cannot drift from the
body. The three textbook tools the proofs use as machinery are in \cref{sec:tools}; the prior results we
\emph{import}, together with the non-standard terminology they involve, are restated next.

\subsection{Prior results we import, and their terminology}
\label{sec:imported}
The proofs lean on three results from prior work that are not ours, and on a few information-theoretic terms
that may be unfamiliar; we explain both informally here before invoking them.

\paragraph{Fingerprinting capacity and collusion.} A \emph{fingerprinting code} gives each of $N$ users a
distinct codeword and asks, from one observed copy, which user it came from. A \emph{collusion} attack lets
several users pool their copies and average them into a forgery that no single codeword matches; the
\emph{fingerprinting capacity} is the largest number of identifiable bits per output symbol that survives such
attacks, under a fidelity (distortion) budget on how far the marked copy may sit from the original. The result
we use \citep{moulin2008universal, somekhbaruch2005capacity} is that this capacity tends to $0$ as the budget
tightens, which is exactly the ``couples to the budget'' phenomenon. We invoke only its degenerate,
\emph{collusion-free} ($c=1$) case, where the worst-case collusion channel disappears and the bound is the
plain Fano inequality (\cref{fact:fano}) applied to $I(U;X)$ under the budget.

\paragraph{Channels with memory, information density, and mixing.} A channel ``with memory'' is one whose
output at position $i$ may depend on the entire past, not just the current input (as an autoregressive
language model does). The \emph{information density} $i(x;y)=\log\frac{p(x,y)}{p(x)\,p(y)}$ is the
log-likelihood ratio whose expectation is the mutual information; a channel is \emph{information-stable} when
the normalised density $\frac1n\sum_i i(X_i;Y_i\mid\text{past})$ concentrates around a constant (a law of large
numbers for information), and a sufficient condition is that the source is \emph{strongly mixing}, i.e.\ its
correlations decay with distance. For information-stable channels the Shannon coding theorem still holds:
roughly $nC$ message bits decode reliably at blocklength $n$, where $C$ is the capacity \citep{verdu1994general}.

\paragraph{Exponential tilt and Jensen--Shannon divergence.} An \emph{exponential tilt} of a distribution $p$
is a one-parameter reweighting $p_+(x)\propto p(x)\,e^{\theta\,\phi(x)}$ that shifts $p$ in a chosen direction
$\phi$ while letting us dial its KL distance to $p$ continuously from $0$ up to the target $\Delta$ we need. The \emph{Jensen--Shannon divergence} $\mathrm{JSD}(p,q)=\tfrac12\KL{p}{m}+\tfrac12\KL{q}{m}$
with $m=\tfrac12(p+q)$ is a symmetric, bounded divergence; it equals the mutual information between a fair coin
choosing $p$ or $q$ and the resulting sample, and so is the per-token information rate of our tilting code.
(For \cref{asm:free}, ``distortion-free'' or ``undetectable'' means the key-averaged law equals the natural
one, $\pbar=\pnat$, so no key-blind test beats chance; the recoverable signal is then limited by the
generation entropy rather than by a distortion budget.)

\begin{fact}[Level-0 detection rate; \citealp{cai2024better,li2024statistical,huang2023optimal}]\label{fact:det}
Under the standard independent/stationary-token detection model with per-token budget $\Delta$, the optimal
detector achieves target error $\delta$ at $\ntask{\det}=\Theta(\log(1/\delta)/\Delta)$, and the single-step
detectability floor is $1-\sqrt{1-e^{-D}}$ at per-step divergence $D$ \citep[Prop.~3.1]{cai2024better}. We use
only the monotone consequence that detection error $\to 0$ as $n\Delta\to\infty$.
\end{fact}

\begin{fact}[Distortion-constrained fingerprinting converse; \citealp{moulin2008universal,somekhbaruch2005capacity}]\label{fact:moulin}
The capacity of a fingerprinting/watermarking channel under a fidelity constraint tends to $0$ as the
constraint tightens; equivalently, identifying one of $N$ users from a length-$n$ output costs
$\Omega(\log N\,/\,\text{(per-symbol fidelity budget)})$. The single-user, collusion-free case is
\cref{fact:fano} applied to $I(U;X)$ under the budget, which is the only case we invoke.
\end{fact}

\begin{fact}[General channel capacity; \citealp{verdu1994general}]\label{fact:verduhan}
For an information-stable channel with memory, the capacity equals $\lim_n \tfrac1n\sup_{P_{X^n}}I(X^n;Y^n)$,
and reliable decoding of $M$ messages is possible at blocklength $n$ once $\log M\le(1-o(1))nC$; the
information density $\tfrac1n\sum_i i(X_i;Y_i\mid \text{past})$ concentrates under stationarity and mixing.
\end{fact}

\subsection{The forensic-recovery budget}

\begin{statementbox}
\ThmBudget*
\end{statementbox}
\begin{proof}
\emph{Idea.} Everything the carrier reveals about the secret passes through the key, and the golden formula
prices that information against any reference law; choosing the natural law turns the price into the
distortion budget, while choosing the source itself turns it into the entropy ceiling.

Since $X$ depends on the secret $S$ only through the key $K=K(S)$, the chain $S\to K\to X$ is Markov. All
quantities are read conditionally on the registry $\kappa$ (\cref{sec:prelim}); given $\kappa$ the secret
determines its enrolled key, so $\bar d$ below is the realized per-key distortion and the sums over $S$ range
over the $N$ enrolled keys (the realized key-mixture is $\hat p_\kappa$; $\pbar=\E_\kappa\hat p_\kappa$ stays
the unconditional marginal). The data-processing inequality (\cref{fact:golden})
gives $\Phibud(n)=I(S;X)\le I(K;X)$ (mutual informations conditional on $\kappa$ throughout). For any distribution $Q$
on $\Sigma^n$,
\[
  \E_K\,\KL{P_{X\mid K}}{Q}
  =\E_K\!\sum_x P_{X\mid K}(x)\log P_{X\mid K}(x)-\sum_x \hat p_\kappa(x)\log Q(x),
\]
since $\E_K P_{X\mid K}(x)=\hat p_\kappa(x)$. The second term is the cross-entropy $H(\hat p_\kappa,Q)$,
minimized over $Q$ at $Q=\hat p_\kappa$ (Gibbs' inequality), and the minimum value makes the right-hand side
equal to $I(K;X)$; this is the
golden formula of \cref{fact:golden}, and in particular, for every fixed $Q$,
\begin{equation}\label{eq:golden}
  I(K;X)\ \le\ \E_K\,\KL{P_{X\mid K}}{Q}.
\end{equation}

\emph{(a)} Put $Q=\pnat$ in \eqref{eq:golden}: under \cref{asm:bounded}, pushing the key average through to
the secret and applying the chain rule for relative entropy,
\[
  I(K;X)\ \le\ \E_K\KL{\pwm^{(K)}}{\pnat}\ =\ \E_S\KL{\pwm(\cdot\mid S)}{\pnat}\ =\ \sum_{t=1}^n\bar d(t)\ \le\ n\Delta,
\]
which uses only the avg-form cap $\max_t\bar d(t)\le\Delta$ (indeed only $\sum_t\bar d(t)\le n\Delta$); with
the data-processing step $\Phibud(n)\le n\Delta$.

\emph{(b)} Working conditionally on $\kappa$, mutual information never exceeds a conditional entropy, and
conditioning on $\kappa$ only reduces entropy; under \cref{asm:free} $\pbar=\pnat$, so
\[
  \Phibud(n)\ =\ I(S;X\mid\kappa)\ \le\ I(K;X\mid\kappa)\ \le\ H(X\mid\kappa)\ \le\ H(X)\ =\ H(\pbar)\ =\ H(\pnat)\ =\ nh_n,
\]
the step $H(X\mid\kappa)\le H(X)$ being the Jensen (concavity-of-entropy) bound
$\E_\kappa H(\hat p_\kappa)\le H(\E_\kappa\hat p_\kappa)=H(\pnat)$ on the realized key-mixture
$\hat p_\kappa:=\tfrac1N\sum_v\pwm^{(k_v)}$ (\cref{app:entropyclose}), and
$h_n:=H(\pnat^{(n)})/n$ is the finite-$n$ entropy
($h_n\downarrow h$ for stationary sources, so asymptotic statements may use the rate $h$).

\emph{(c)} Under $\mathcal E(m)$ the per-secret kernels equal the natural ones off $T$, so
$X_t\perp S\mid X_{<t}$ for every $t\notin T$ and $\nu(t)=I(S;X_t\mid X_{<t})=0$ there
(\cref{prop:profile}(iii); conditioning on the \emph{full} past screens autoregressive propagation). The
chain rule then gives $\Phibud(n)=\sum_t\nu(t)=\nu(T)$, and each remaining term is at most
$H(X_t\mid X_{<t})\le\log q$, whence $\Phibud(n)\le|T|\log q\le m\log q$.
\end{proof}

\noindent\emph{Scope of (b).} The entropy bound $\Phibud\le H(X)$ needs no undetectability assumption, but the
evaluation $H(X)=nh_n$ uses \emph{statistical} undetectability $\pbar=\pnat$. Computational undetectability gives no two-sided
entropy bound: for $\pnat=\Unif(\{0,1\}^n)$, a pseudorandom generator with seed length $s=\mathrm{polylog}(n)$
has output $\approx_c\pnat$ yet entropy $\le s\ll nh$. Computationally undetectable schemes (Christ--Gunn,
PRC) are therefore governed by seed entropy, not $nh_n$.

\subsection{Attribution converse and the forensic gap}

\begin{statementbox}
\ThmSeparation*
\end{statementbox}
\begin{proof}
\emph{Converse.} Apply \cref{thm:budget}(a) with $S=U\sim\Unif[N]$: each $\KL{\pwm^{(k_u)}}{\pnat}\le n\Delta$,
so $I(U;X)\le n\Delta$. Fano's inequality (\cref{fact:fano}), together with $H(U\mid X)=H(U)-I(U;X)=\log
N-I(U;X)\ge\log N-n\Delta$, gives
\[
  \log N-n\Delta\ \le\ H(U\mid X)\ \le\ \log 2+P_e\log N,
\]
which rearranges to the converse
\[
  P_e\ \ge\ 1-\frac{n\Delta+\log 2}{\log N}.
\]
The bound is non-trivial once $n\Delta<\log N-\log 2$, giving $\ntask{\att}=\Omega(\log N/\Delta)$. This is the
single-user case of the distortion-constrained fingerprinting converse of
\citet{moulin2008universal,somekhbaruch2005capacity} (\cref{fact:moulin}), specialized to the per-token KL
channel.

\emph{Forensic gap.} Call the scheme \emph{mixture-detectable} if the key-unknown likelihood-ratio test of
$\pbar$ against $\pnat$ has error $\to 0$ as $n\to\infty$, and write $\ntask{\det}^{\mathrm{mix}}$ for the
length at which it does. In the stationary-token model this is implied, when $\pbar$ is ergodic, by a positive per-token mixture
divergence $\tfrac1n\KL{\pbar}{\pnat}\ge c>0$: by the relative-entropy (Shannon--McMillan--Breiman) theorem
(\cref{fact:stein}(ii)) the normalised log-likelihood ratio $\tfrac1n\log(\pbar/\pnat)$ concentrates at a
positive rate under $\pbar$ and a negative rate under $\pnat$, so the likelihood-ratio test separates the
hypotheses with vanishing error. (For a non-ergodic key-mixture $\pbar$ one needs the divergence rate positive
on every ergodic component; a component that coincides with $\pnat$ would leave a constant miss probability.) (We do
\emph{not} use the i.i.d.\ Chernoff--Stein lemma applied to a single length-$n$ sample, for which a diverging KL
alone would not force the error to $0$.) On the window $\ntask{\det}^{\mathrm{mix}}\le n<(1-\eta)\log N/\Delta$,
detection error $\to 0$ while $n\Delta<(1-\eta)\log N$, so the converse gives $P_e\ge
1-\frac{(1-\eta)\log N+\log 2}{\log N}=\eta-o(1)$. The window is non-empty when $\log
N=\omega(\Delta\,\ntask{\det}^{\mathrm{mix}})$. At $n^\star=\lceil\log\log N/\Delta\rceil$ (assumed
$\ge\ntask{\det}^{\mathrm{mix}}$, so detection still succeeds), $I(U;X)\le n^\star\Delta\le\log\log
N+\Delta=\log\log N+O(1)$ ($\Delta$ a fixed per-scheme constant), whence $P_e\ge 1-\frac{\log\log
N+O(1)}{\log N}\to 1$.
\end{proof}

\noindent\emph{Necessity of mixture-detectability.} The hypothesis $c>0$ is not free. The parity scheme on
$\Sigma^n=\{0,1\}^2$ with two keys, $\pwm^{(0)}=\Unif\{x:x_1\oplus x_2=0\}$ and
$\pwm^{(1)}=\Unif\{x:x_1\oplus x_2=1\}$, has $\KL{\pwm^{(b)}}{\pnat}=\log 2$ (within budget) yet
$\pbar=\tfrac12(\pwm^{(0)}+\pwm^{(1)})=\Unif(\{0,1\}^2)=\pnat$ exactly, so the key-unknown detector is blind.
For such (distortion-free) schemes ``detectable but not attributable'' is automatic, not a separation.

\subsection{Two quality models}

\begin{statementbox}
\ThmDichotomy*
\end{statementbox}
\begin{proof}
\emph{Idea.} The lower bound is the budget once more: $\mathcal B(\Delta)$ caps the mass at
$n\Delta$, so attribution needs $n\gtrsim\log N/\Delta$. The two upper bounds spend that mass in the two
extremal shapes, thin-and-wide (a tilting code on every token) and tall-and-narrow (the digit stamp on a few).

Write $d_i^{(k)}$ for the per-key conditional distortions, so
$\sum_i d_i^{(k)}=\KL{\pwm^{(k)}}{\pnat}$ by the chain rule for relative entropy.

\emph{(a) Lower bound.} A scheme in $\mathcal B(\Delta)$ (sup-form: every per-secret conditional within
$\Delta$ at every prefix) has $d_i^{(k)}\le\Delta$ for every $i,k$, so
$\KL{\pwm^{(k)}}{\pnat}=\sum_i d_i^{(k)}\le n\Delta$ and \cref{thm:separation} applies (the converse needs
only this avg-form consequence), giving
$\ntask{\att}=\Omega(\log N/\Delta)$; the $\ell$-bit version replaces $U\sim\Unif[N]$ with
$W\sim\Unif\{0,1\}^\ell$ in Fano, giving $\ntask{\ext}=\Omega(\ell/\Delta)$.

\emph{(a) Upper bound (conditional tilting code).} For a fixed $\Delta$, let $p_+^{(i)}(\cdot\mid x_{<i})$ be
the order-preserving exponential tilt of the natural conditional $\pnat(\cdot\mid x_{<i})$ with
$\KL{p_+^{(i)}}{\pnat^{(i)}}=\Delta$. Assign user $u$ an i.i.d.\ codeword $c_u\in\{0,1\}^n$; generate
$X_i\sim p_+^{(i)}(\cdot\mid X_{<i})$ when $c_u[i]=1$ and $X_i\sim\pnat(\cdot\mid X_{<i})$ otherwise. By the KL
chain rule the total distortion is exactly $(\#\{i:c_u[i]=1\})\,\Delta\le n\Delta$ for \emph{any}
autoregressive $\pnat$, so the scheme lies in $\mathcal B(\Delta)$ (each token within $\Delta$). The induced
per-position channel $c_u[i]\mapsto X_i$ (conditioned on the past) has mutual information
$\E\,\JSD(p_+^{(i)},\pnat^{(i)})=\Theta(\Delta)$; under a mild mixing assumption on $\pnat$ the conditional
information density $\tfrac1n\sum_i i(c_u[i];X_i\mid X_{<i})$ concentrates, and the general channel-with-memory
coding theorem (Verdú--Han, \cref{fact:verduhan}) gives reliable maximum-likelihood decoding of $N$ codewords once
$n=\Theta(\log N/\Delta)$. The same code carries $\ell$ message bits, giving $\ntask{\ext}=\Theta(\ell/\Delta)$.

\emph{(b) Footprint upper bound (the stamp).} Let $m=\lceil\log N/\log q\rceil$ and $\Lmap_k(x)=\{1,\dots,m\}$.
Encode $u\in[N]$ by its $m$ base-$q$ digits and set $X_1,\dots,X_m$ to those digits; generate
$X_{m+1},\dots,X_n\sim\pnat$. Then the per-secret kernels for $i>m$ are exactly natural, so the scheme is in
$\mathcal E(m)$ with marked set $T=\{1,\dots,m\}$, and the
extractor reading $X_{1:m}$ recovers $u$ with zero error, so $\Phibud=\log N$ and the region has size
$m=\Theta(\log N/\log q)$, independent of any per-token cap. Fidelity holds by the footprint $m/n\to0$; if a
KL reading is wanted, assuming $\pnat(\text{digit}_i\mid\text{past})\ge p_0>0$ gives total distortion
$\le m\log(1/p_0)$, an $n$-independent constant, so the average per-token distortion $\to 0$. The $\ell$-bit
version forces $\lceil\ell/\log q\rceil$ positions.
\end{proof}

\begin{statementbox}
\CorBridge*
\end{statementbox}
\begin{proof}
In $\mathcal E(m)$ the per-secret kernels off $T$ are natural, so the likelihood factorizes as
\[
  \pwm(x\mid s)\;=\;\prod_{t\in T}\pwm(x_t\mid s,x_{<t})\cdot\prod_{t\notin T}\pnat(x_t\mid x_{<t}),
\]
and the ratio $\pwm(x\mid s)/\pnat(x)=\prod_{t\in T}\pwm(x_t\mid s,x_{<t})/\pnat(x_t\mid x_{<t})$ depends on
$s$ only through the pairs $((x_{<t},x_t))_{t\in T}$. By the factorization criterion, $X_{\le\max T}$ is a
sufficient statistic for $S$; and $\nu=0$ off $T$ (\cref{prop:profile}(iii)) gives $I(S;X)=\nu(T)$. When $T$
is a prefix $\{1,\dots,m\}$ (the stamp), sufficiency is exactly the Markov chain $S\to X_T\to X$. (For
general $T$ the chain $S\to X_T\to X$ can fail: the readout needs the conditioning prefixes, i.e.\
$X_{\le\max T}$, not $X_T$ alone.) In all cases $\Lmap_k:=T$ localizes the mark: the $s$-dependence of the
likelihood lives on $T$, read conditionally on the (natural) prefix.

\emph{Converse leg.} Suppose the scheme also satisfies the avg-form cap $\max_t\bar d(t)\le\Delta$ (which
the sup-form implies). Then $\nu(t)\le\bar d(t)\le\Delta$ at every $t$ ($\nu(t)\le\bar d(t)$ by the golden
formula; cf.\ \cref{prop:profile}(ii)), so vanishing-error attribution forces, by Fano,
$(1-o(1))\log N\le I(U;X)=\nu(T)\le m\Delta$, i.e.\ $m\ge(1-o(1))\log N/\Delta$: localization does not buy
$\Delta$-freeness --- dropping the per-token cap does.
\end{proof}

\subsection{Level 3: footprint--resolution and the rate region}

\begin{statementbox}
\ThmCrop*
\end{statementbox}
\begin{proof}
\emph{Idea.} A crop window disjoint from the dilated footprint sees only unmarked text and cannot detect, so
the footprint must meet every length-$w$ window; a set meeting every window cannot be too sparse, which
lower-bounds its size.

\emph{(a)/(b).} Let $R^{+}=\{i:\mathrm{dist}(i,R)\le r\}$ be the dilated support, $|R^{+}|\le(2r+1)|R|$.
Suppose a
length-$w$ window $W$ is disjoint from $R^{+}$. By the coupling of \cref{asm:local}, the marked carrier agrees
on $W$ with an unmarked carrier $X'\sim\pnat$, so $x_W=x'_W$ has law exactly $\pnat$; a detector observing only
$x_W$ then tests $\pnat$ against $\pnat$ and cannot exceed chance, contradicting $w$-crop-robust detection.
Hence every length-$w$ window meets $R^{+}$. There are $n-w+1$ such windows, and a single element $i$ lies
in at most $w$ of them (those starting in $[i-w+1,i]$), so $|R^{+}|\ge(n-w+1)/w$, tight for
$R^{+}=\{w,2w,\dots\}$. Combining with $|R^{+}|\le(2r+1)|R|$ gives $|R|\,w\ge(n-w+1)/(2r+1)$, i.e.\
$|R|\,w=\Omega(n)$ for $r=O(1)$ (for $w>n/2$ a nonempty $R$ gives $|R|\,w>n/2$ directly); for a scheme in
$\mathcal E(m)$ whose natural law factorizes across positions, the off-$T$ kernels are prefix-independent,
an exact $r=0$ coupling with $R=T$ exists, and the bound reads $m\,w=\Omega(n)$ (for correlated sources an
$R=T$ coupling can fail to exist: a scheme stamping one symbol of a perfectly correlated source is
$w$-crop-robust for every $w$ with $|T|=1$, and is excluded exactly because its carrier footprint under any
valid coupling is $\Omega(n)$, not $|T|$). (i) and (ii) are
the regimes $|R|=o(n)$ (in particular $m=o(n)$ under $\mathcal E(m)$) and $w=O(1)$.

\emph{(c).} Here we take $\pnat$ to factorize across the length-$w$ blocks (the i.i.d.\ / whitened regime;
the general correlated case is discussed below). Partition $[n]$ into $N_w=\lfloor n/w\rfloor$ disjoint
length-$w$ blocks $W_1,\dots,W_{N_w}$; crop to a uniformly random block $J\sim\Unif[N_w]$. Crop-robust
localization must output $J$ from $x_{W_J}$ alone, so by Fano (\cref{fact:fano}) $I(J;x_{W_J})\ge(1-o(1))\log
N_w$. For the upper bound, the golden formula (\cref{fact:golden}) with the common natural block law
$\pnat^{W}$ as reference gives $I(J;x_{W_J})\le\E_J\,\KL{P_{x_{W_J}\mid J}}{\pnat^{W}}$. Because $\pnat$ is
block-independent, its within-block conditionals coincide with the full-prefix conditionals, so by the
chain rule over the block and joint convexity of relative entropy each block term obeys
$\KL{P_{x_{W_J}\mid J}}{\pnat^{W}}\le\sum_{i\in W_J}d_i$ (the marked within-block conditional, with its
pre-block prefix marginalized, stays within $d_i:=\E_K d_i^{(K)}=\bar d(i)$ of the natural conditional by
convexity). Hence
\[
  I(J;x_{W_J})\ \le\ \E_J\!\sum_{i\in W_J}d_i
  \ =\ \frac{1}{N_w}\sum_{i\in[n]}d_i
  \ \le\ \frac{n\Delta}{N_w}\ =\ w\Delta,
\]
using only the \emph{global} budget $\sum_i d_i\le n\Delta$ of \cref{asm:bounded} (avg-form), not a
per-token cap. Combining, $\log N_w\le(1+o(1))\,w\Delta$, i.e.\ $w\Delta\ge(1-o(1))\log(n/w)$.

\emph{The block-independence hypothesis is load-bearing.} For a correlated $\pnat$ the cropped block
$x_{W_J}$ marginalizes out its pre-block prefix, and $\KL{P_{x_{W_J}\mid J}}{\pnat^{W}}$ need not be
bounded by $\sum_{i\in W_J}d_i$: the divergence of the cropped block against the stationary block law can
exceed the prefix-conditioned per-token budget actually spent there (a binary first-order Markov source
already exhibits this). The conclusion $I(J;x_{W_J})\le w\Delta$ nonetheless held throughout a numerical
search over correlated stationary sources and localizing schemes (the ratio $I/(w\Delta)$ peaking at $1$,
attained at the i.i.d.\ point), so we expect the relation to be general; a tight proof for correlated
$\pnat$ is left open.
\end{proof}

\begin{statementbox}
\PropRegion*
\end{statementbox}
\begin{proof}
\emph{Converse.} With $S=(U,W)$ and $U\perp W$, the chain rule and \cref{thm:budget}(a) give
\[
  I(U;X)+I(W;X\mid U)\ =\ I(U,W;X)\ \le\ n\Delta.
\]
Vanishing-error attribution forces $I(U;X)\ge(1-o(1))\log N$ by Fano (\cref{fact:fano}), and vanishing-error
extraction forces $I(W;X\mid U)\ge(1-o(1))\ell\log2$ (Fano applied to $W\in\{0,1\}^\ell$ conditioned on each
value of $U$ and averaged); adding,
\[
  R_{\att}+R_{\ext}\ \le\ (1+o(1))\,n\Delta.
\]

\emph{Achievability.} Concatenate the user identity and the payload into a single message of
$R_{\att}+R_{\ext}$ nats and transmit it with the conditional tilting code of \cref{thm:dichotomy}(a), whose
per-token capacity is $C=\Theta(\Delta)$. Under the regularity/mixing hypothesis the information density
concentrates, so reliable decoding holds once $nC\ge(1+o(1))(R_{\att}+R_{\ext})$; the decoder splits the
recovered message into $(\hat U,\hat W)$.
\end{proof}

\subsection{Proof of \texorpdfstring{\cref{thm:entropyclose}}{the entropy-column theorem}}
\label{app:entropyclose}

\begin{statementbox}
\ThmEntropyClose*
\end{statementbox}
\begin{proof}
\emph{Idea.} The true text is $\pnat$-typical, so its surprisal clears any sub-$h$ threshold (completeness, by
SMB); an innocent key's string is itself $\pnat$-distributed, so it coincides with the typical text only with
$\pnat$-mass $\le e^{-\tau}$ (soundness, a change of measure). The threshold $\tau=\log(N/\delta)$ is what
balances the two sides.

Throughout, the true user's key makes the text law exactly $\pnat$ (the inverse-CDF sampler is
distortion-free: marginally over the key, $X\sim\pnat$). \emph{Consistency} of a key $k_v$ with a text $X$
means that $k_v(t)$ lands in the realized token's quantile cell of $\pnat(\cdot\mid X_{<t})$ at every $t$;
for the deterministic sampler this says exactly $X(k_v)=X$, where $X(k_v)$ is the single string key $k_v$
generates.

\emph{(i) False positives, pointwise version (deterministic sampler).} Fix the realized innocent keys. Each
innocent $k_v$ determines one string $X(k_v)$, and $\hat U=v$ requires $X=X(k_v)$ and $\imath(X)>\tau$. The
only
randomness left is the true text $X\sim\pnat$, so
\[
  \Pr[\hat U=v]\ =\ \pnat\big(X(k_v)\big)\,\mathbf 1\{\imath(X(k_v))>\tau\}\ \le\ e^{-\tau}\ =\ \delta/N,
\]
since on $\{\imath>\tau\}$ we have $\pnat=e^{-\imath}<e^{-\tau}$ \emph{pointwise}. This holds for every
realized key
set (no averaging over the key draw) and for arbitrary $\pnat$ (no ergodicity), at every $n$.

\emph{(i$'$) Averaged version (randomized distortion-free kernels).} Given $X$, an independent wrong key
$k_v$ is consistent with probability $\prod_t\pnat(X_t\mid X_{<t})=\pnat(X)=e^{-\imath(X)}$ (each $k_v(t)$
must
land in the realized token's quantile cell, and $k_v(t)\sim\Unif[0,1]$ independently). Hence
\[
  \Pr[\hat U=v]\ \le\ \E_{X\sim\pnat}\!\big[e^{-\imath(X)}\,\mathbf 1\{\imath(X)>\tau\}\big]
  \ =\!\!\sum_{X:\,\pnat(X)<e^{-\tau}}\!\!\pnat(X)^2\ \le\ e^{-\tau}\sum_X\pnat(X)\ =\ e^{-\tau}.
\]
This is the watermark-forensics instance of the folklore likelihood-ratio Markov bound
$\Pr_{\pnat}[\log(P^{k_v}/\pnat)>\tau]\le e^{-\tau}$ (Barron-style change of measure): for any randomized
distortion-free kernel only the key-averaged bound survives, which is how the ``any scheme'' reading must be
scoped.

\emph{(ii) Completeness.} The true key is always consistent, so the decoder finds $U$ iff $\imath(X)>\tau$.
By the
Shannon--McMillan--Breiman theorem \citep{cover2006elements}, $\imath(X)/n\to h$ almost surely for stationary
ergodic $\pnat$; at $n\ge(1+\alpha)\tau/h$ we have $\tau/n\le h/(1+\alpha)<h$, so
$\Pr[\imath(X)>\tau]=\Pr[\imath(X)/n>\tau/n]\to1$. \emph{Quantifier order.} Completeness needs $\tau/n$
bounded away
from $h$: with $\delta$ fixed and $N\to\infty$ this holds along $n=(1+\alpha)\tau/h$ with the $o(1)$ uniform;
for jointly $\delta\to0$ one needs $\log(1/\delta)=O(\log N)$, keeping $\tau=\Theta(\log N)$ --- this is the
asymptotic regime stated in the theorem.

\emph{(iii) Rate.} Combining, the average error is at most
$\Pr[\imath\le\tau]+\sum_{v\ne U}\Pr[\hat U=v]\le o(1)+N\cdot\delta/N=\delta+o(1)$ at
$n=(1+\alpha)(\log N+\log(1/\delta))/h$, any fixed $\alpha>0$.

\emph{(iii) Converse.} Write $P^{k}:=\pwm^{(k)}$ for the marked law under key $k$, and
$\hat p_\kappa:=\tfrac1N\sum_{v}P^{k_v}$ for the realized $N$-key mixture given the key set
$\kappa=\{k_1,\dots,k_N\}$; each $P^{k_v}$ is a random measure with $\E[P^{k_v}]=\pnat$ (the sequence-level
condition $\pbar=\pnat$; for the witness this follows from the kernel-level form because its keys are
i.i.d.\ across positions, so inverse-CDF sampling is distortion-free at the sequence level, cf.\ \cref{asm:free}),
and since the keys are i.i.d., $\E_\kappa[\hat p_\kappa]=\pbar=\pnat$. Condition on the realized key set $\kappa$: $U\sim\Unif[N]$ and
$X\sim P^{k_U}$; since $U\perp\kappa$ and $U\sim\Unif[N]$, the $\kappa$-conditional marginal of $X$ is
$\hat p_\kappa$, so $H(X\mid\kappa)=H(\hat p_\kappa)$ and $I(U;X\mid\kappa)\le H(X\mid\kappa)
=H(\hat p_\kappa)$. Fano (\cref{fact:fano}) at each $\kappa$ gives
$P_e(\kappa)\ge1-(I(U;X\mid\kappa)+\log2)/\log N$; averaging over the i.i.d.\ key draw (by concavity of
entropy and Jensen, $\E_\kappa H(\hat p_\kappa)\le H(\E_\kappa\hat p_\kappa)=H(\pnat)=nh_n$; the inequality
is textbook and we claim only its use)
\[
  P_e\ =\ \E_\kappa\,P_e(\kappa)\ \ge\ 1-\frac{\E_\kappa\,I(U;X\mid\kappa)+\log2}{\log N}
  \ \ge\ 1-\frac{nh_n+\log2}{\log N}.
\]
Any rule with error $\le2\delta$ (the threshold decoder achieves $\le\delta+o(1)$) therefore needs
$n\ge(1-2\delta)\log N/h_n-O(1/h)$. Two symbols do the bookkeeping here: completeness (ii) uses the entropy
\emph{rate} $h$ via Shannon--McMillan--Breiman, while the converse uses the finite-$n$ quantity
$h_n=H(\pnat^{(n)})/n$ exactly; for stationary sources $h_n\downarrow h$, so $nh_n=(1+o(1))nh$ and the two
sides meet asymptotically. Since $\alpha>0$ is arbitrary, $\ntask{\att}=\Theta(\log N/h)$ for every fixed
$\delta\in(0,\tfrac12)$ (achievability alone holds for all $\delta\in(0,1)$; the converse coefficient
$1-2\delta$ needs $\delta<\tfrac12$); letting $\delta\to0$ with $\log(1/\delta)=o(\log N)$ squeezes the two sides together,
$\ntask{\att}=(1+o(1))\log N/h$, closing the column.

\emph{Fixed-key variant.} For a \emph{fixed} deployed key set $\kappa$ the averaging step is unavailable;
$I(U;X)\le H(\hat p_\kappa)$ still holds, and under the carried hypothesis
$H(\hat p_\kappa)\le nh_n+o(\log N)$ (the slack $o(\log N)$ as $N\to\infty$)
the same Fano computation gives $P_e\ge1-(nh_n+o(\log N)+\log2)/\log N$, hence
$\ntask{\att}=\Omega(\log N/h)$
and, with $W\sim\Unif\{0,1\}^\ell$ in place of $U$, $\ntask{\ext}=\Omega(\ell/h)$. (When $\pnat$ is the
maximum-entropy law the hypothesis is automatic.)
\end{proof}

\begin{statementbox}
\CorExtraction*
\end{statementbox}
\begin{proof}
Identify the payload $W\sim\Unif\{0,1\}^\ell$ with $U\sim\Unif[N]$ for $N=2^\ell$, so $\log N=\ell\log2$
nats. \Cref{thm:entropyclose} gives $\ntask{\ext}=\Theta(\ell/h)$ for fixed $\delta$, and
$\ntask{\ext}=(1+o(1))\,\ell\log2/h$ as $\delta\to0$ with $\log(1/\delta)=o(\ell)$.
\end{proof}

\subsection{The overlap identity and the candidate list}
\label{app:overlap}

\begin{lemma}[Scheme-independent collision identity]\label{lem:overlap}
For any statistically distortion-free scheme satisfying the kernel-level condition
$\E_{k\sim\mu}\,p(x\mid k,x_{<t})=\pnat(x\mid x_{<t})$ for every prefix $x_{<t}$ and every $x$ (randomized
kernels included), and independent keys $k,k'\sim\mu$,
\[
  \E_{k,k'}\Big[\sum_x p(x\mid k,x_{<t})\,p(x\mid k',x_{<t})\Big]\ =\ \sum_x\pnat(x\mid x_{<t})^2 .
\]
The expected pairwise per-token overlap is therefore invariant across all distortion-free schemes.
\end{lemma}
\begin{proof}
For independent $k,k'$, $\E[AB]=\E[A]\,\E[B]$ termwise in $x$; then apply the kernel-level distortion-free
condition twice.
\end{proof}

For the inverse-CDF witness of \cref{thm:entropyclose} the kernels are point masses given the prefix, and the
identity is also a direct sequence-level computation:
$\E_{k,k'}\sum_X P^{k}(X)P^{k'}(X)=\sum_X\pnat(X)^2$. With an i.i.d.\ source,
$\sum_X\pnat(X)^2=e^{-nr_2}$ where $r_2:=-\log\sum_{x\in\Sigma}\pnat(x)^2$ is the R\'enyi-2 rate, so the
expected number $M$ of falsely consistent keys is
\[
  \E[M]\ =\ (N-1)\,\E_{X\sim\pnat}[\pnat(X)]\ =\ (N-1)\sum_X\pnat(X)^2\ =\ (N-1)\,e^{-nr_2},
\]
which is the exponentially large candidate list of \cref{rem:listsize} in the window
$\log N/h\le n\ll\log N/r_2$.

\emph{Abstain-on-tie has zero false positives (closed world).} The rule outputs $v$ only when $v$ is the
\emph{unique} consistent key. A false positive requires a wrong $v$ to be uniquely consistent; but the true
key is always consistent with the text it generated --- contradiction. Hence $\mathrm{FP}=0$ exactly, at
every $n$ and every source, \emph{within} the closed-world model $U\in[N]$; on out-of-set text (no true key
present) the argument vanishes. Its completeness $\E_X[(1-\pnat(X))^{N-1}]\to1$ at the threshold
$(1+o(1))\log N/h$ is supported numerically and is \emph{not}
claimed as a theorem; \cref{thm:entropyclose} does not use it.

\subsection{Two definitional witnesses: why the per-secret recut is necessary}
\label{app:witnesses}

The v2 definitions of \cref{sec:prelim} state the quality caps on the \emph{per-secret} kernels
$\pwm(\cdot\mid s,x_{<t})$, not on the key-averaged marginal distortion profile. Two two-token examples show
that the marginal profile and the information profile disagree in both directions, so the recut is forced,
not cosmetic.

\paragraph{Parity (marginal $\delta\equiv0$, yet $\nu(2)=\log2$).} On $\Sigma=\{0,1\}$, $n=2$,
$\pnat=\Unif(\{0,1\}^2)$, secret $S\sim\Unif\{0,1\}$ with $\pwm^{(s)}=\Unif\{x:x_1\oplus x_2=s\}$. The
key-averaged law is $\pbar=\Unif(\{0,1\}^2)=\pnat$ exactly, so the marginal per-token distortion vanishes at
both positions, and a key-averaged footprint reading would place the scheme in $\mathcal E(0)$ --- absurd,
since $\Phibud=I(S;X)=\log2$: given $X_1$, the token $X_2$ determines the parity, so
$\nu(2)=I(S;X_2\mid X_1)=\log2$. Per-secret, the kernel at $t=2$ is a point mass $\ne\Unif$, so $t=2$ is
marked, as it must be.

\paragraph{Kernel bend (marginal $\delta_2>0$, yet $\nu\equiv0$).} Same $\Sigma$, $n=2$, $\pnat=\Unif$;
\emph{every} secret generates $X_1\sim\Unif\{0,1\}$ and then $X_2\sim\mathrm{Ber}(3/4)$, regardless of $S$.
Then $\delta_2=\KL{\mathrm{Ber}(3/4)}{\mathrm{Ber}(1/2)}>0$ while $X\perp S$, so $\nu\equiv0$: position $2$
carries distortion but no secret. Distortion support does not certify information.

Together: the marginal profile both under-counts (parity) and over-counts (bend) the secret-bearing
positions, while the per-secret profile $\bar d$ dominates both ($\nu(t)\le\bar d(t)$ by the golden formula
and $\delta^{\mathrm{marg}}(t)\le\bar d(t)$ by joint convexity of KL), which is what makes the v2 caps sound
and the screening property $\supp(\nu)\subseteq T$ (\cref{prop:profile}(iii)) a theorem rather than a
convention.

\section{How deployed watermarks fit the framework}
\label{sec:schemefit}

We now place each scheme of \cref{tab:taxonomy} inside the model of \cref{sec:prelim}: we identify its
per-secret footprint $T$ (hence its side of \cref{def:embedbias}), the quality model it obeys
(\cref{asm:bounded} or \cref{asm:free}, or the footprint cap $\mathcal E(m)$), its readout resolution, and the
forensic levels it reaches. The recurring lesson is that almost every deployed scheme is \emph{biasing}:
full-support, and localizing (when it does) by fine readout rather than by a small footprint.

\subsection{Language-model watermarks}

\paragraph{Green-list \citep{kirchenbauer2023watermark}.} At each step the key hashes the previous token to
split the vocabulary into a green/red list and adds a fixed logit bias toward green; detection counts green
tokens. Every token's conditional is tilted under every key, so the per-secret footprint is $T=[n]$ and the scheme is \emph{biasing}
(\cref{def:embedbias}) under the bounded-distortion model $\mathcal B(\Delta)$, with $\Delta$ the per-token
green-bias divergence. Readout is global (one green-fraction statistic), so it reaches Level 0 and, in its
multi-bit and many-user descendants, Levels 1--2, but not crop-robust Level 3. Its secret \emph{is} recoverable
from a long window, which is why the support, not a read-from criterion, is the right cut (\cref{sec:prelim}).

\paragraph{Unigram-Watermark \citep{zhao2024unigram}.} The context-width-$1$ ablation of green-list: a single
key-determined global green list (no context hashing), a fixed logit bias toward it, and a
one-proportion $z$-test on the green count. Every token's conditional is tilted under every key, so $T=[n]$:
\emph{biasing} under $\mathcal B(\Delta)$, global readout, Level 0. Its contribution is a \emph{provable}
edit- and paraphrase-robustness guarantee, and the fixed list buys it by making the mark maximally
delocalized: robustness here is redundancy of a full-support mark, not a small footprint.

\paragraph{SynthID-Text \citep{dathathri2024synthid}.} Tournament sampling reweights the next-token
distribution through keyed elimination rounds; in its non-distortionary mode the key-averaged law is the
model's own, $\pbar=\pnat$. It is therefore \emph{biasing} in the distortion-free regime (\cref{asm:free}):
full support, global readout, Level-0 (zero-bit) detection.

\paragraph{Distribution-preserving reweighting \citep{wu2024dipmark,hu2024unbiased}.} DiPmark and the
unbiased-watermark family reweight the next-token distribution by a key-conditioned function whose
key-average is the identity: $\E_K\,\pwm(\cdot\mid K,x_{<t})=\pnat(\cdot\mid x_{<t})$ at every prefix --- the
\emph{kernel-level} form of \cref{asm:free}, not merely the sequence-level one. Under each key the reweighting
touches every step, so $T=[n]$: \emph{biasing} (d-free), global readout (a keyed score summed over the
sequence), Level 0.

\paragraph{Multi-bit watermarks \citep{yoo2023multibit,arcmark2026}.} These carry an $\ell$-bit payload by
position-dependent green-lists (MPAC, bounded-distortion $\mathcal B(\Delta)$) or an optimal-transport coupling
(ArcMark, distortion-free, \cref{asm:free}); either way the mark perturbs every token's per-secret conditional, so $T=[n]$
and the scheme is \emph{biasing}, with bucketed (per-position-class) rather than fine readout. The family
realises Levels 0 and 2, at the $\Theta(\ell/\Delta)$ payload cost of \cref{tab:ladder}.

\paragraph{BiMark \citep{feng2025bimark}.} An \emph{unbiased multilayer} multi-bit scheme: a bit-flip
reweighting whose key-average is the identity ($\E_K\,\pwm(\cdot\mid K,x_{<t})=\pnat(\cdot\mid x_{<t})$, the
kernel-level form of \cref{asm:free}), stacked across several layers to raise detectability while preserving
quality (lower perplexity). Every step is reweighted, so $T=[n]$: full support, \emph{biasing} in the
distortion-free regime. Multi-bit messages are carried by an XOR-enhanced position-allocation code, read back
bit by bit from a voting matrix without the message space, so the readout is bucketed (per-position-class)
rather than fine; the scheme reaches Levels 0 and 2. Being distortion-free, its payload cost is governed by the
entropy column rather than a distortion budget: \cref{cor:extraction-h} caps message-agnostic $\ell$-bit
extraction at $\Theta(\ell/h)$ tokens.

\paragraph{Many adaptive users \citep{cohen2024many}.} This construction is built generically from a
distortion-free zero-bit scheme plus an erasure-robust (Boneh--Kiayias--Montgomery) fingerprinting code,
inheriting full support and $\pbar=\pnat$: \emph{biasing} in the entropy regime (\cref{asm:free}). It embeds a
message and traces among many users, reaching Levels 0--2; its detect-short / trace-long guarantee is the
achievability companion to our Level-0/Level-1 converse (\cref{thm:separation}).

\paragraph{Distortion-free watermarks \citep{kuditipudi2023robust}.} An exponential-minimum-sampling rule
(after Aaronson's exponential-sampling proposal for GPT, \citealp{aaronson2022watermark}) aligns the output
with a key-dependent random sequence while keeping $\pbar=\pnat$; the alignment signal is
spread over the whole sequence. Full support, \emph{biasing} (\cref{asm:free}), global readout, Level 0. Its
robustness to cropping is redundancy of a delocalised mark, not a small footprint.

\paragraph{Undetectable watermarks \citep{christ2024undetectable}.} A pseudorandom signal is embedded in the
high-entropy steps of generation so that, without the key, the output is computationally indistinguishable
from $\pnat$. Full support, distortion-free \emph{biasing} (\cref{asm:free}); the governing resource is the
generation entropy (the paper certifies completeness through the empirical entropy of the realized output),
and the scheme reaches Level 0.

\paragraph{Pseudorandom error-correcting codes \citep{christ2024pseudorandom}.} A PRC encodes a pseudorandom
codeword spread across the output, decodable with the key even after bounded corruption. Full support,
distortion-free \emph{biasing} (\cref{asm:free}); the error-correction buys robustness, not a small footprint.
The language-model watermark it instantiates is zero-bit, so it reaches Level 0; multi-bit message embedding is
a separate steganographic application of the same code.

\paragraph{SemStamp \citep{hou2024semstamp}.} Sentence-level rejection sampling: a keyed locality-sensitive
hash partitions a sentence-embedding space into regions, and each candidate sentence is resampled until it
lands in a watermarked region; detection counts in-region sentences. The extractor's basis is the
sentence-embedding space, so the footprint is read at sentence granularity --- and there every sentence's
per-secret conditional is the natural law \emph{conditioned} on the keyed region, hence marked: full support,
\emph{biasing} under $\mathcal B(\Delta)$ (the per-sentence divergence is $\log(1/\gamma)$ for region mass
$\gamma$ under $\pnat$), global readout, Level 0. The basis change swaps tokens for sentences without
shrinking the footprint's fraction: basis-relativity (\cref{sec:taxonomy}, pattern 3) in the opposite
direction from Tree-Ring, for which the change of basis is what creates the $o(n)$ object.

\subsection{Image and audio watermarks}

\paragraph{Tree-Ring \citep{wen2023treering}.} A ring is written into a few concentric frequency shells of the
diffusion model's initial-noise Fourier spectrum; detection inverts the generation and matches the ring. In
the latent-noise spatial basis the inverse-DFT of the ring is delocalized over all noise coordinates (and the
noise-to-image map is the full diffusion sampler, not a transform), so there it is full support, biasing; but
in the \emph{Fourier basis of the latent} the mark is an $o(n)$ object on a handful of shells: it is \emph{embedding}
there, the one deployed scheme with a genuine small footprint. Its localization map is the ring's support and
its readout is global (a matched filter), so it reaches Level 0 and, with key-distinct rings, Level 1, but not
spatial Level 3. This is the basis-dependence of \cref{def:embedbias} made concrete.

\paragraph{RingID \citep{ci2024ringid}.} The Tree-Ring successor built for multi-key
\emph{identification} --- Level 1 in practice, instantiated in the same latent-Fourier dictionary, the regime
of \cref{thm:dichotomy}(b)'s attribution question. Two verified design changes decide its classification.
The ring values are discretized (each ring carries $\pm\alpha$ rather than a draw from the Gaussian Fourier
prior), so the key-averaged law on the shells is no longer natural (RingID itself documents the
distribution shift this family incurs), and the d-free reading survives only partially. And to make many
keys discriminable it adds a heterogeneous second component: a key-specific Gaussian noise pattern occupying
an \emph{entire} latent channel, read jointly with the rings by a nearest-key $\ell_1$ match after DDIM
inversion. That channel is a constant fraction of the latent coordinates, so the per-secret footprint is
$\Omega(n)$: \emph{biasing}, global readout, Levels 0--1. The instructive part is the trade: the ring
component alone is the deployed object closest to the small-footprint attribution of \cref{thm:dichotomy}(b),
and RingID buys reliable multi-key capacity precisely by leaving the embedding class.

\paragraph{Gaussian Shading \citep{yang2024gaussianshading}.} The latent prior is replaced by a key-determined
sample that is still marginally Gaussian, so the output is performance-lossless ($\pbar=\pnat$) while carrying
a message. Full support in the latent, distortion-free \emph{biasing} (\cref{asm:free}), global readout;
Levels 0 and 2.

\paragraph{ROBIN \citep{huang2024robin}.} An explicit watermark is planted in the Fourier spectrum of an
\emph{intermediate} diffusion latent (around steps 200--300 of 1000), and an adversarially optimized hiding
prompt-guidance signal steers the remaining generation to conceal it; verification DDIM-inverts to the
embedding step and thresholds an $\ell_1$ distance to the planted pattern. The optimized pattern spans
concentric rings covering roughly $70\%$ of the frequency plane, so even in its own Fourier basis the
footprint is $\Omega(n)$: full support, \emph{biasing} under $\mathcal B(\Delta)$ (the optimization
explicitly budgets retention against an unwatermarked anchor image, and the steered law is not $\pnat$),
global readout, Level 0 (zero-bit).

\paragraph{HiDDeN \citep{zhu2018hidden}.} The classic learned image watermark and the ancestor of Stable
Signature: an encoder network maps a cover image and an $\ell$-bit message to a visually indistinguishable
encoded image, and a jointly trained decoder recovers the message through differentiable noise layers (blur,
crop, dropout, JPEG). The learned residual spans the whole image, so $T$ is full: \emph{biasing} under
$\mathcal B(\Delta)$, global readout of one bitstring; Levels 0 and 2.

\paragraph{Stable Signature \citep{fernandez2023stablesig}.} The latent decoder is fine-tuned so every image
it produces carries a fixed bitstring, recovered by a trained extractor. The signature perturbs the whole
image's features (full support); it is \emph{biasing} under $\mathcal B(\Delta)$ with global readout of one
bitstring, reaching Levels 0 and 2 but, lacking per-region decoding, not Level 3.

\paragraph{SEAL \citep{seal2025}.} A semantic locality-sensitive hash of the image is bound into the watermark
and read \emph{per patch}, which lets it localize tampering at high AUC where latent schemes score zero. The
mark covers the whole carrier (full support, distortion-free \emph{biasing}, \cref{asm:free}); what reaches
Level 3 is its \emph{fine readout resolution}, not a small footprint, exactly as \cref{thm:crop} requires of
any crop-robust localizer. Levels 0--3.

\paragraph{EditGuard \citep{editguard2024}.} A dual localization-plus-copyright watermark embeds a fragile,
spatially-indexed mark across all pixels and recovers a tamper mask from reconstruction error. Full support,
\emph{biasing} under $\mathcal B(\Delta)$, but with \emph{fine} per-region readout, so it reaches Levels 0
and 3 (and carries a copyright bitstring, a Level-2 channel).

\paragraph{AudioSeal \citep{sanroman2024audioseal}.} A detector network reads a per-sample localization signal
embedded throughout the waveform, flagging AI-generated segments at sample resolution. Full support,
\emph{biasing} under $\mathcal B(\Delta)$, fine readout; Levels 0 and 3 (a 16-bit message variant adds Level 2).

\paragraph{The digit-stamp (ours).} For contrast, the decoupling witness of \cref{thm:dichotomy}(b) writes the
secret into the first $m=\Theta(\log N/\log q)$ tokens and leaves the rest natural: $T=[m]$ with $|T|=o(n)$, so it
is genuinely \emph{embedding} in token space and attributes at a $\Delta$-free cost. It is also the only such
construction here, and it is not robust to a single edit, which is precisely the gap between the
information-theoretic decoupling and a deployable scheme (\cref{sec:discussion}).

\section{Additional related work}
\label{sec:morerelated}

\Cref{sec:related} positioned the paper against the work its theorems build on or answer. This appendix
broadens the lens to recent watermarking research adjacent to the framework, read, where the connection is
real, through the profile $\nu$ and the ladder. Independent and concurrent, \citet{liu2026implicitidentity}
survey a related landscape under the name \emph{implicit identity}, organizing fingerprinting and
watermarking across datasets, models, and generated content into a lifecycle taxonomy; our contribution is
orthogonal in kind, not degree: rather than taxonomize the space of mechanisms, we isolate one
generation-time mark and derive exact sample-complexity rates for each forensic capability it supports.

\paragraph{Design and detection theory for language-model watermarks.} A fast-moving line treats scheme
design as an optimization problem. \citet{he2025dawa} jointly optimize the watermarking scheme and its
detector under a worst-case false-positive constraint and distortion control, with closed-form
distribution-adaptive solutions, and \citet{ai2026pasa} watermark over semantic clusters in embedding space
with provable robustness to semantic-invariant (paraphrase) attacks; \citet{tsur2025heavywater} cast
distortion-free design for low-entropy
generation (code, structured output) as optimal transport, the regime where the entropy column prices
attribution most dearly ($h$ small makes $\log N/h$ large); \citet{zhao2025permuteflip} design a decoder
(Permute-and-Flip) with a provably better quality--stability trade-off and a Gumbel-analogous watermark
tailored to it; and \citet{gloaguen2026dlmwatermark} give the first watermark for diffusion language models,
whose arbitrary-order generation is a carrier process our autoregressive profile does not directly cover
(re-cutting $\nu$ for order-agnostic generation is a natural extension). \citet{liu2026incontext} move the
marking capability from the sampler to the prompt itself; in our terms the deposit schedule is then chosen
without kernel access, a constrained class the framework does not yet price. On the detection side,
\citet{fernandez2023bricks} consolidate deployed tests with non-asymptotic false-positive guarantees, the
Level-0 discipline whose attribution-side analogue is the per-innocent $\delta/N$ form of
\cref{thm:entropyclose}(i); \citet{li2025trgof} prove a truncated goodness-of-fit test adaptively optimal for
detecting the Gumbel-max watermark under human edits, robust Level-0 where our open problem asks for robust
Level 1. Closest in spirit to our currencies, \citet{he2026tradeoffs} derive information-theoretic
trade-offs (false alarm, detection error, distortion, information rate) for multi-bit watermarking of
stationary processes; like the concurrent \citet{distribembed2025}, the analysis lives at the level of
aggregate rates, and the positional decomposition $\nu(t)$, on which our shape and localization results
turn, is again the difference. A further question in the same currency is answered by
\citet{chai2026gumbelproportion}: estimating what \emph{fraction} of a document is watermarked under the
Gumbel-max mechanism, with matching information-theoretic lower bounds for both a full-observation and a
pivotal-statistic estimator---a mixed-authorship reading of the mass $\|\nu\|_1$ our framework has not
itself posed. Finally, \citet{golowich2024edit} construct watermarks that are
simultaneously undetectable and robust to a constant fraction of adversarial edits via indexing
pseudorandom codes: the strongest known step toward the regime of our sharpest open problem, since it
achieves edit-robust Level 0, and whether an indexing-style construction supports $N$-user attribution at
the $h$-rate is exactly the question \cref{sec:conclusion} poses.

\citet{qin2026increfa} attack a different notion of attribution entirely: attributing an image to its
generating model from learned artifacts alone, with no embedded key and an open, ever-growing set of
candidate generators (their IncreFA casts this as incremental learning against 28 models spanning
2022--2025). \Cref{thm:entropyclose} instead prices attribution \emph{among a fixed, enrolled key set}; the
two notions bound complementary settings, passive generator identification versus keyed user attribution,
and neither theorem transfers to the other's regime.

\paragraph{Attribution in agent-trajectory carriers.} Moving beyond token streams, \citet{gao2026trace}
watermark the \emph{trajectories} of LLM agents---their logs of tool calls and actions---for robust ownership
attribution, splitting the mark across two complementary channels: a selection channel keyed on local content
and a tally channel keyed on log structure that stays invariant under rewriting. Read through the ladder this
is Level-1 attribution in a non-text carrier, and the design speaks directly to our sharpest open problem: the
structural tally channel is a deposit schedule whose support is engineered to survive step deletion and
paraphrase, empirical evidence that the edit-robust attribution \cref{sec:conclusion} calls for becomes
attainable once the carrier admits a rewriting-invariant coordinate. In a related but distinct carrier,
\citet{lu2026bicot} watermark a model's own \emph{chain-of-thought}: an ownership signal is aligned with
high-saliency structural anchors in the internal geometry of the reasoning trace, and a model-free
registration test verifies it, robust to fine-tuning, quantization, model perturbation, and output-level
attack. Where \citet{gao2026trace} split the mark across an agent's external action log, here the carrier is
the reasoning process a single model produces on its own, a second illustration that attribution generalizes
past the plain generated token sequence our profile $\nu$ is defined over.

\paragraph{Image and multimodal schemes, and localization.} \citet{li2025shallowdiffuse} embed
diffusion-image watermarks in a low-dimensional subspace approximately orthogonal to the generation
manifold; in the taxonomy's terms this is a latent-basis choice, and the basis relativity of
\cref{def:embedbias} (only the mass $\norm{\nu}_1$ is invariant) is the right lens for such subspace
designs. \citet{gao2026slice} bind distinct semantic factors of an image to designated regions of the
initial noise, obtaining false-accept guarantees together with tamper localization, per-factor readout on
the axis \cref{thm:crop} prices. \citet{zeng2026lava} layer audio-visual anti-tampering marks with
cross-modal alignment for deepfake detection and localization, Level-3 readout in a multimodal carrier,
company to EditGuard and AudioSeal in \cref{tab:taxonomy}'s fine-readout rows. Two further constructions
narrow to a single carrier each: \citet{milis2026audio} watermark synthetic speech alone, exploiting
token-vocabulary redundancy for a gradient-free scheme robust to audio-domain edits, and
\citet{chen2026spherical} watermark panoramic imagery alone, coupling third-order $SO(3)$ representations so
the mark survives arbitrary 3-D rotation of the sphere. Both are carrier-specific instances of the same
footprint-versus-basis question \cref{sec:taxonomy} poses in general.

\paragraph{Attacks and removal.} Removal is the adversarial complement of our lower bounds: the rates price
the decoder's task when the mark survives, and attacks decide whether it survives (the impossibility line of
\citealp{zhang2023sand, francati2026codinglimits}). \citet{bao2026shift} give a training-free removal attack
on diffusion watermarks that deflects the latent trajectory by stochastic resampling, succeeding against
nine schemes at once because they share a reliance on trajectory reconstruction; \citet{gao2026breaking}
defeat semantic-aware image watermarks by LLM-guided, coherence-preserving semantic injection, attacking in
the semantic basis the scheme reads, a reminder that the footprint, like sparsity, is basis-relative
(\cref{sec:taxonomy}). \citet{ni2025fragility} track how the mutual information a robust mark carries decays
under iterated diffusion editing --- an information-theoretic reading of removal on the same quantity $\nu$
whose deposit our converses price.

\section{What is a watermark? A philosophical and historical discussion}
\label{sec:philosophy}

The question in this section's title is older than the technology that prompts it, and its oldest answer is an
image. In the paper mills of Fabriano, around 1282, a figure of twisted wire was sewn onto the mould on which
the pulp was couched; where the wire lay the sheet formed thinner, and the design (a crown, a hand, a
unicorn) appeared not on the page but \emph{through} it, when the leaf was lifted to the light
\citep{hunter1947papermaking}. The mark was in the very substance of the paper and yet nowhere on its surface;
it named the maker, the mill, the year, and it did so to guard against forgery, so that the sheet could vouch
for its own origin. Nearly everything essential to the objects of this paper is already present in that
gesture: a mark laid into a carrier at the moment of its making, imperceptible to ordinary reading, and
legible only from a particular standpoint --- held, as it were, against the light.

Two ancient impulses meet in it. The first is to \emph{authenticate}: the seal in wax, the potter's mark, the
assayer's hallmark struck on silver, the signature, each a way of binding a made thing to the one who made
it, of letting an object carry its own provenance. The second is to \emph{conceal}. The art of hidden writing
is at least as old as \citet[5.35, 7.239]{herodotus1998histories}, who tells that Histiaeus pricked a message
into the shaved scalp of a slave and waited for the hair to regrow before sending him, and that Demaratus
scraped the wax from a tablet, wrote on the wood beneath, and sealed the words again under fresh wax.
Antiquity had no name for the practice; it was Trithemius, at the close of the fifteenth century, who supplied
one (\emph{steganographia}, covered writing) and bequeathed the modern art its name
\citep{kahn1967codebreakers}. A generative watermark is the union of these two impulses, and a strange union
it is, for it authenticates \emph{by} concealing: the distortion-free mark of \cref{thm:entropyclose} is a
seal that leaves the wax unbroken, a signature in an ink that, to every eye but one, is the paper itself.

What kind of thing, then, is such a mark? Here the historical question turns metaphysical, and the standing
temptation is to seek an object where there is none. Shown a marked passage and asked to point to the
watermark, one hunts for a token, a phrase, a buried signature, and, for the schemes that matter, finds
nothing. The error is one of category. It is the error \citet{ryle1949} diagnosed in the visitor who, shown
the colleges, the library, and the halls, asks to be shown the University in addition: the University is not a
further building beside the others but the manner in which they are ordered, and to look for it among them is
to mistake its logical type. In an older vocabulary one seeks a substance and is offered an accident, a way
that a substance is \citep[Categories, chs.~2, 5]{aristotle1963categories}. A biasing watermark is exactly
this: not an object in the carrier but a property of the law that generated it, a disposition of the text to
answer a certain test, present at every position and locatable at none. To ask which token \emph{is} the mark
is to ask which molecule of a warm gas carries its heat. That a property with no place can nonetheless be
measured, named, and priced is the bequest of \citet{shannon1948mathematical}: once information is reconceived
as distinguishability rather than as stuff, a mark need not be a thing in order to be real, and the profile
$\nu$ becomes the exact record, position by position, of how far the generating law has been bent to
carry the secret, whether that bending is gathered into an object or dissolved into a property.

If the watermark is not an object, what relation does it bear to its origin? The precise answer was furnished,
in another connection, by \citet{peirce1931collected}, who divided the ways a sign can stand for its object
into three. An \emph{icon} signifies by resemblance, a portrait, a map; a \emph{symbol} by convention,
a word, a flag; but an \emph{index} signifies ``by virtue of being really affected by that Object'' (CP
2.248): smoke by fire, a weathercock by the wind, a footprint by the foot that pressed it. A watermark is
neither icon nor symbol. It does not resemble its origin, and it refers to it by no agreement; it is an
\emph{index}. The correlations a decoder reads were physically laid down by the key at the instant of
generation, and to detect or to attribute is to read backward along that causal thread to the act that
deposited it. This is why forensics is possible at all (forensics is the reading of indices), and it is
why the forensic-recovery budget of \cref{thm:budget} has the form it has. An index can carry only the causal
information that flowed into it; no reading, however ingenious, recovers a trace the cause did not leave. Data
processing is nothing but this ancient fact set in the language of information: a trace cannot be manufactured
downstream of the event that makes it.

An index that points to a person becomes a name. The passage from the trace to the name was made by
\citet{kripke1980}, for whom a name is fastened to its bearer by an \emph{initial baptism} and thereafter
transmitted along a causal--historical chain, referring not through any description its user holds but through
the unbroken line back to the baptismal act. The key, drawn at generation time, is such a baptism; the marked
text carries the chain folded into its statistics; and to attribute is to trace the chain to the name that
fixed it. Detection and attribution are thus two questions of unequal depth put to a single index. Detection
asks only whether a mark \emph{exists} --- a question of being, indifferent to whose it is. Attribution asks
\emph{whose}, and this is the harder question, for it demands not existence but individuation: what makes this
text the work of \emph{this} origin and no other. The schoolmen had a word for the metaphysical ground of such
thisness: the \emph{haecceity} of \citet{scotus1994}, the ultimate difference by which a thing is this one
and not another of its kind. Existence is cheap; haecceity is dear, and \cref{thm:entropyclose} prices it
exactly, at $\log N$ nats in the appropriate currency, for to single one origin out of $N$ is to pay $\log N$
nats of thisness.

There remains a paradox at the core of every mark of authenticity, and it is the deepest lesson the theory
holds. A mark can authenticate only if it can be checked, and to be checked it must in some measure be
reproducible --- the verifier must be able to regenerate or to recognize it. But a mark that is reproducible
is, in the same measure, forgeable. \citet{derrida1988limited} found this tension in the signature, the very
emblem of a singular and present author: a signature works as the sign of that singularity only because it has
``a repeatable, iterable, imitable form,'' and the iterability that lets it function at all is the same that
makes it forgeable. A watermark inherits the paradox in its sharpest form and answers it with the \emph{key}.
The mark is made iterable for exactly one party, the holder of the key, and the security of the scheme is
nothing but the size of the gap between the iterability granted to the verifier and the iterability withheld
from the forger. The distortion-free construction is, in this light, a resolution of the signature's paradox
by information: the mark is wholly public in its effect (the text circulates freely, indistinguishable from
unmarked text) and wholly private in its ground (only the key can read it), and the entire force of
\cref{thm:entropyclose} lives in that difference. Here, too, \citet{leibniz1989} returns. His principle that
no two things are indiscernible (that to be numerically distinct is to differ in some property) appears,
in the distortion-free regime, to fail: two users' outputs are, as laws, statistically identical, told apart
by no test whatever. The principle is not broken but \emph{relativized}. Indiscernibility is always
indiscernibility \emph{to} a standpoint; the texts no test can separate are separated at once by the holder of
the registry, and the watermark is precisely the information that lives in the difference between what one
standpoint can tell apart and what another cannot. This is the whole content of the registry-conditioning of
\cref{sec:prelim}: against the world the mark's mass is zero, against the key it is positive, and the scheme
is nothing but that difference.

We may now answer the question of the title. A watermark is not a mark in a text but a structure in the
relation among three terms: a text, the origin that generated it, and a standpoint from which that origin
can be read. It is a trace, in Peirce's exact sense, made legible by a key; a name, in Kripke's, deposited at
the text's birth and recoverable along a causal thread; a commitment, in the end, by which the words are made
to answer for where they came from. The text is unchanged; what the mark alters is not what the words say but
what they \emph{owe} --- how much of their origin they can be compelled to confess, and from where. The
profile $\nu$ is the ledger of that debt, entered once at generation and never after, and the four forensic
questions of this paper are the four questions one has always been able to put to a trace: that it is
(detection), whose it is (attribution), what it says (extraction), and where it lies (localization). Their
prices, which the foregoing has computed, are the prices of reading a mark that was made, like the figure in
the Fabriano sheet, to be found only by one who knows how to hold it to the light.

\end{document}